\documentclass[journal,11pt]{article}
\usepackage[a4paper,includeheadfoot,margin=2.5cm]{geometry}
\usepackage{etex}
\usepackage{graphicx}
\usepackage{stfloats}
\graphicspath {{Figs/}}
\usepackage{pgfplotstable, booktabs}
\usepackage{array}
\usepackage[square, numbers, comma, sort&compress]{natbib}
\usepackage{amsmath,amsfonts,amssymb,amscd,amsthm,xspace}
\usepackage{mathrsfs}
\usepackage{mathtools} 
\usepackage{ifpdf}
\usepackage[scriptsize]{subfigure}
\usepackage{array}
\usepackage[utf8]{inputenc}
\usepackage[english]{babel}
\usepackage[kerning=true]{microtype} 
\usepackage{float}  
\usepackage{fancyhdr}
\usepackage{stmaryrd}
\usepackage{color}
\usepackage{bbm}
\usepackage{dsfont}
\usepackage{hyperref}
\hypersetup{
	colorlinks   = true,
	citecolor    = red,
	linkcolor    = blue
}
\usepackage{tikz}
\usetikzlibrary{positioning}
\usetikzlibrary{calc}
\usetikzlibrary{arrows}
\usepackage[many]{tcolorbox} 
\usepackage{siunitx,booktabs}
\usepackage{setspace}
\newcommand{\msgs}{\mathcal{W}}
\newcommand{\mn}{\left|\mathcal{W}\right|}
\newcommand{\msg}{w}
\newcommand{\Msg}{W}
\newcommand{\esg}{v}
\newcommand{\Esg}{V}
\newcommand{\usg}{u}
\newcommand{\Usg}{U}
\newcommand{\cl}{n}

\newcommand{\clwd}{x}
\newcommand{\cwd}{\mathbf{x}}
\newcommand{\Clwd}{X}
\newcommand{\Cwd}{\mathbf{X}}

\newcommand{\hlwd}{\overline{x}}

\newcommand{\Hlwd}{\overline{X}}

\newcommand{\cset}{\code\left({\msg}\right)}
\newcommand{\Cset}{\Code\left(\msg\right)}

\newcommand{\Csetl}{\Codel\left(\msg\right)}

\newcommand{\csetl}{\codel\left(\msg\right)}

\newcommand{\cwdp}{{\mathbf{x}_{\leq\dpr}}}
\newcommand{\Cwdp}{{\mathbf{X}_{\leq\dpr}}}

\newcommand{\optCwdp}{{\mathbf{X}_{\leq\optdpr}}}

\newcommand{\cwdl}{{\mathbf{x}_{>\dpr}}}
\newcommand{\Cwdl}{{\mathbf{X}_{>\dpr}}}

\newcommand{\optCwdl}{{\mathbf{X}_{>\optdpr}}}

\newcommand{\optHwdl}{{\overline{\mathbf{X}}_{>\optdpr}}}

\newcommand{\rwdp}{{\mathbf{y}_{\leq\dpr}}}
\newcommand{\Rwdp}{{\mathbf{Y}_{\leq\dpr}}}

\newcommand{\optrwdp}{{\mathbf{y}_{\leq\optdpr}}}
\newcommand{\optRwdp}{{\mathbf{Y}_{\leq\optdpr}}}

\newcommand{\rwdl}{{\mathbf{y}_{>\dpr}}}

\newcommand{\optrwdl}{{\mathbf{y}_{>\optdpr}}}
\newcommand{\optRwdl}{{\mathbf{Y}_{>\optdpr}}}

\newcommand{\cwdpc}{\mathbf{x}_{\leq\dpc}}
\newcommand{\Cwdpc}{\mathbf{X}_{\leq\dpc}}

\newcommand{\cwdlc}{\mathbf{x}_{>\dpc}}
\newcommand{\Cwdlc}{\mathbf{X}_{>\dpc}}

\newcommand{\rwdpc}{\mathbf{y}_{\leq\dpc}}
\newcommand{\Rwdpc}{\mathbf{Y}_{\leq\dpc}}
\newcommand{\rwdlc}{\mathbf{y}_{>\dpc}}
\newcommand{\Rwdlc}{\mathbf{Y}_{>\dpc}}

\newcommand{\optStp}{{\mathbf{S}_{\leq\optdpr}}}

\newcommand{\optStl}{{\mathbf{S}_{>\optdpr}}}

\newcommand{\rlwd}{y}
\newcommand{\rwd}{\mathbf{y}}
\newcommand{\Rlwd}{Y}
\newcommand{\Rwd}{\mathbf{Y}}
\newcommand{\rwds}{\mathcal{Y}}
\newcommand{\st}{\mathbf{s}}
\newcommand{\St}{\mathbf{S}}
\newcommand{\slt}{s}
\newcommand{\Slt}{S}

\newcommand{\Zt}{\mathbf{Z}}

\newcommand{\Zlt}{Z}

\newcommand{\apoe}{\mathbbm{P}^{\cl}_\mathrm{Avg}}
\newcommand{\mpoe}{\mathbbm{P}^{\cl}_\mathrm{Max}}

\newcommand{\ampoe}{\overline{\mathbbm{P}}^{\cl}_{\mathrm{Max}}}
\newcommand{\mpoec}{\mathbbm{P}^{\cl}_{\mathrm{Max}|\mathcal{C}}}
\newcommand{\ballP}{\mathcal{X}}
\newcommand{\ballN}{\mathcal{S}}
\newcommand{\ballPN}{\mathcal{Y}}
\newcommand{\ballx}{\mathcal{B}\left(\cwdlc,\bost\right)}
\newcommand{\Codep}{{\Code}_{\leq\dpc}}
\newcommand{\Codel}{{\Code}_{>\dpc}}

\newcommand{\rwdsp}{\underline{\rwds}}
\newcommand{\rwdsl}{\overline{\rwds}}

\newcommand{\code}{\mathcal{C}}
\newcommand{\Code}{\mathscr{C}}
\newcommand{\codep}{\mathcal{C}_{\leq\dpc}}
\newcommand{\codel}{\mathcal{C}_{>\dpc}}

\newcommand{\listxs}{{\mathcal{X}}_{>\dpc}}
\newcommand{\listx}{{\mathcal{X}}_{>\dpc}\left(\codel,\cbd,\listas,\msg,\dpc\right)}
\newcommand{\listv}{{\mathcal{W}}_{>\dpc}\left(\codel,\cbd,\listas,\rwdlc,\dpc\right)}
\newcommand{\listvs}{{\mathcal{W}}_{>\dpc}}

\newcommand{\listVs}{{\mathcal{W}}_{>\dpc}}

\newcommand{\listd}{{\mathcal{W}}_{>\dpc_{1} }}
\newcommand{\lista}{{\mathcal{W}}_{\leq\dpc}\left(\codep,\cbd,\rwdpc,\dpc\right)}
\newcommand{\listas}{{\mathcal{W}}_{\leq\dpc}}
\newcommand{\listAxtend}{{\mathcal{W}}_{\leq\dpc}\left(\codep,\cbg,\rwdpc,\dpc\right)}
\newcommand{\clistf}{{M}_{\leq\dpc}\left(\Codep,\cbd,\rwdpc,\dpc\right)}
\newcommand{\clistfs}{{M}_{\leq\dpc}}
\newcommand{\clistvs}{{M}_{>\dpc}}

\newcommand{\clistVs}{{M}_{>\dpc}}
\newcommand{\clista}{{M}_{\leq\dpc}\left(\codep,\cbd,\rwdpc,\dpc\right)}
\newcommand{\clistas}{{M}_{\leq\dpc}}
\newcommand{\ent}{\mathbbm{H}}
\newcommand{\mut}{\mathbbm{I}}
\newcommand{\spa}{\mathbf{P}}
\newcommand{\npa}{\mathbf{N}}
\newcommand{\spad}{\mathbf{\Phi}}
\newcommand{\npad}{\mathbf{\Psi}}
\newcommand{\spas}{\mathcal{P}}
\newcommand{\npas}{\mathcal{N}}
\newcommand{\uspas}{\mathcal{P}(\nu)}
\newcommand{\unpas}{\mathcal{N}(\nu)}
\newcommand{\spasd}{{\mathcal{I}}}
\newcommand{\npasd}{{\mathcal{J}}}

\newcommand{\npase}{\mathcal{N}_{\tau}\left(\dpr,\spa\right)}
\newcommand{\npaseu}{\mathcal{N}\left(\dpr,\spa\right)}
\newcommand{\optnpase}{\mathcal{N}_{\tau}\left(\optdpr,\spa\right)}
\newcommand{\npasde}{\mathcal{J}_{\gamma}\left(\dpc,\spad\right)}

\newcommand{\indd}{\mathcal{S}_{\dpr}}
\newcommand{\expc}{\mathbbm{E}}
\newcommand{\cul}{K}
\newcommand{\dpr}{m}
\newcommand{\optdpr}{m^*}
\newcommand{\dpc}{\mu}

\newcommand{\indT}{T}
\newcommand{\indt}{t}
\newcommand{\cbd}{\mathbf{F}}
\newcommand{\cbg}{{\mathbf{F}_{\Delta}}}
\newcommand{\cost}{{{F}}^{\Delta}_{\dpc}}
\newcommand{\ost}{F}
\newcommand{\aost}{{{F}}_{\dpc+1}}
\newcommand{\bost}{{{F}}_{\dpc}}
\newcommand{\dost}{{{F}}_{\dpc_{1} +1}}
\newcommand{\post}{{{F}}_{\dpc_{1} }}
\newcommand{\npal}{N}
\newcommand{\spal}{P}
\newcommand{\npadl}{\mathit{\Psi}}
\newcommand{\spadl}{\mathit{\Phi}}

\newtheorem{theorem}{Theorem}
\newtheorem{lemma}{Lemma}
\newtheorem{corollary}{Corollary}
\newtheorem{claim}{Claim}
\theoremstyle{definition}
\newtheorem{definition}{Definition}

\newtheorem{remark}{Remark}

\newtcolorbox{mybox}[1]{%
	tikznode boxed title,
	enhanced,
	arc=0mm,
	interior style={white},
	attach boxed title to top center= {yshift=-\tcboxedtitleheight/2},
	fonttitle=\bfseries,
	colbacktitle=white,coltitle=black,
	boxed title style={size=normal,colframe=white,boxrule=0pt},
	title={#1}}

\title{Quadratically Constrained Channels with Causal Adversaries}

\author{Tongxin LI \\
	Caltech\\
\textsf{tongxin@caltech.edu}
\and
Bikash Kumar Dey\\
IIT Bombay\\
\textsf{bikash@ee.iitb.ac.in}
\and
Sidharth Jaggi\\
CUHK\\
\textsf{jaggi@ie.cuhk.edu.hk}
\and
Michael Langberg\\
SUNY Buffalo\\
\textsf{mikel@buffalo.edu}
\and
Anand D. Sarwate\\
Rutgers\\
\textsf{asarwate@ece.rutgers.edu}
}

\begin{document}
\maketitle


\begin{abstract}
We consider the problem of communication over a channel with a {\it{causal}} jamming adversary subject to {\it{quadratic constraints}}. A sender Alice wishes to communicate a message to a receiver Bob by transmitting a real-valued length-$\cl$ codeword $\cwd=\left(\clwd_1,\ldots,\clwd_\cl\right)$ through a communication channel. Alice and Bob do not share common randomness. Knowing Alice’s encoding strategy, a jammer James chooses a real-valued length-$\cl$ adversarial noise sequence $\st=\left(\slt_1,\ldots,\slt_\cl\right)$ in a causal manner: each $\slt_\indt$ ($1\leq\indt\leq\cl$) can {\it{only}} depend on $\left(\clwd_1,\ldots,\clwd_{\indt}\right)$. Bob receives $\rwd$, the sum (over $\mathbbm{R}$) of Alice’s transmission {$\cwd$} and James’ jamming vector {$\st$}, and is required to reliably estimate  Alice’s message from this sum. 

In this work we characterize the channel capacity for such a channel as the limit superior of the optimal values $C_\cl\left(\frac{\spal}{\npal}\right)$ of a series of optimizations.
Upper and lower bounds on $C_\cl\left(\frac{\spal}{\npal}\right)$ are provided both analytically and numerically. Interestingly, unlike many communication problems, in this causal setting Alice’s optimal codebook may {\it not} have a uniform power allocation — for certain $\mathrm{SNR}$ a codebook with a {\it two-level uniform} power allocation results in a strictly higher rate than a codebook with a uniform power allocation would.

\end{abstract}

\providecommand{\keywords}[1]{\textbf{\textit{Index terms---}} #1}


\section{Introduction}
\label{sec:1}
A transmitter Alice wishes to reliably send a message $\msg$ to a receiver Bob. To do this, she first encodes the message $\msg$ to a codeword $\cwd$. The codeword $\cwd$ is set to be a length-$\cl$ real-valued sequence $\cwd=\clwd_1,\ldots,\clwd_{\cl}$ (satisfying the power constraint specified later in~\eqref{eq:1.3}. Alice then transmits the coordinates $\clwd_1,\clwd_2,\ldots,\clwd_{\cl}$ one-by-one through a channel. During each time-step $\indt\in \{1,\ldots,\cl\}$, the coordinate $\clwd_\indt$ is transmitted.  However, there is an adversarial jammer James sitting in between Alice and Bob. James maliciously controls the communication channel and he is able to add (coordinate-by-coordinate over ${\mathbbm R}$) a sequence $\st=\slt_1,\ldots,\slt_{\cl}$ of real-valued noise to $\cwd$. Prior to transmitting a specific $\cwd$, both James and Bob know the potential corresponding codewords for each message and their distributions\footnote{In the general scenario considered in this work, the message and codeword are random variables $\Msg$ and $\Cwd$. Alice is able to \textit{stochastically} select codewords for each message according to the distribution $p_{\Cwd|\Msg}$. In this case, we assume the distribution $p_{\Cwd|\Msg}$ is accessible to both James and Bob. Hence there is \textit{no} shared private randomness between Alice and Bob. Prior to the communication, everything that Bob knows, so does James.}.  But since Alice sends the codeword sequentially, James is restricted to choose $\st$ in the following \textit{causal} manner: at each time-step $\indt$, the noise $\slt_\indt$ has to be decided at the current time right after James observes the coordinates $\clwd_1,\clwd_2$ up to $\clwd_\indt$, \textit{i.e.,} each $\slt_\indt$ should be selected as a function of $\left(\clwd_1,\ldots,\clwd_\indt\right)$ without knowing the coordinates $\clwd_{\indt+1},\ldots,\clwd_{\cl}$ in the future. To decode the transmitted message, Bob waits until he receives the corrupted codeword $\rwd=\cwd+\st$.

Let $\spal$ and $\npal$ be positive constants. The codeword $\cwd$ and the sequence $\st$ must satisfy the following \textit{quadratic constraints}:
\begin{align}
\label{eq:1.3}
&\sum_{\indt=1}^{\cl}\clwd_{\indt}^2\leq \cl\spal,\\
\label{eq:1.4}
&\sum_{\indt=1}^{\cl}\slt_{\indt}^2\leq \cl\npal.
\end{align}

The constants $\spal$ and $\npal$ can be considered respectively as the \textit{signal power} and \textit{noise power} for Alice and James. 

\begin{figure*}[b]	
	\centering
	\includegraphics[scale=0.245]{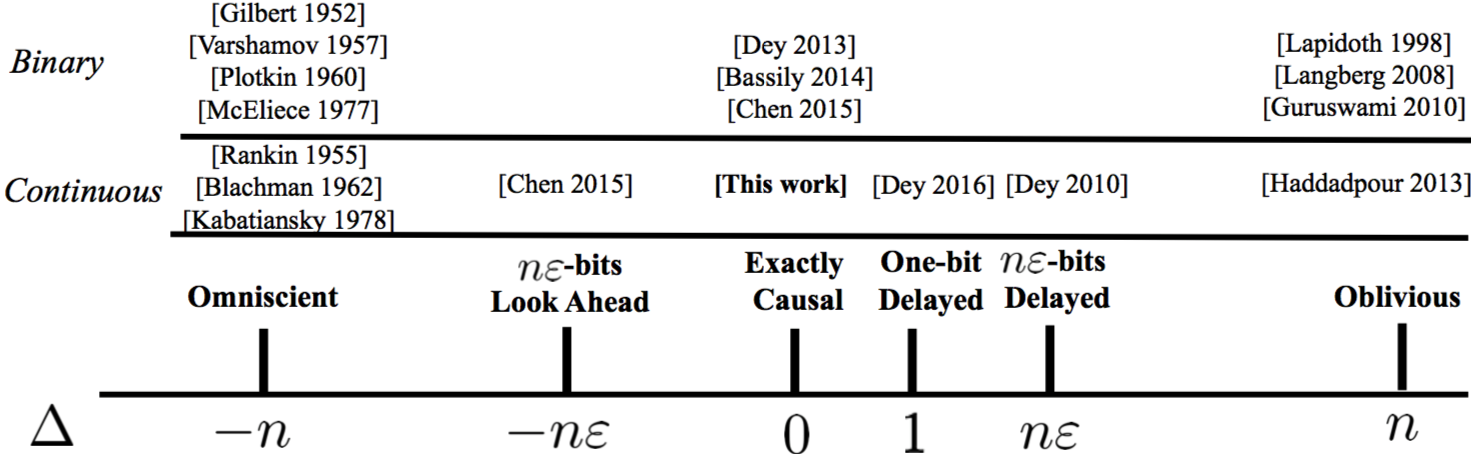}
	\caption{A graphical representation of known results for AVCs without common randomness, parametrized by the delay parameter $\Delta$.}
		\label{Fig:0}
\end{figure*}

The fundamental problem for point-to-point communication over such a channel is as follows. 
\textit{``How can Alice transmit as many messages as possible while ensuring Bob is able to decode them with high probability regardless of James’s jamming actions?''}
In this work we characterize the channel capacity for such a channel. Denoting by $C_{\mathsf{cau}}\left(\frac{\spal}{\npal}\right)$ (as a function of $\mathrm{SNR}=\frac{\spal}{\npal}$) the corresponding capacity for a causal channel with quadratic constraints, we show that for all $\spal>0$ and $\npal>0$,
\begin{align}
\label{eq:0}
C_{\mathsf{cau}}\left(\frac{\spal}{\npal}\right)=
\limsup_{\cl\rightarrow\infty}C_\cl\left(\frac{\spal}{\npal}\right).
\end{align}
Here, each $C_\cl\left(\frac{\spal}{\npal}\right)$ is the optimal value of an optimization problem presented in (P1.1)-(P1.6).
\subsection{Related Work}
\subsubsection{Quadratic Constraints}
%

The quadratic constraints presented in (\ref{eq:1.3}) and (\ref{eq:1.4}) are traditional assumptions in Shannon's communication model~\cite{shannon1949communication}. 
A particular signal can be considered as a point in an $\cl$-dimensional space where the dimension $\cl$ depends on the time interval.  In this sense, the quadratic constraints $\spal$ and $\npal$ provide natural restrictions on the power of Alice and James respectively.

A standard assumption in classical coding theory  is that the $\cl$-dimensional noise vector chosen by the adversary James depends on the entire $\cl$-dimensional vector transmitted by Alice, ignoring the impact of James’ potential lack of foreknowledge of Alice’s future transmissions. 

A finer classification of adversarial jamming problems based on causality assumptions, is as follows:

\subsubsection{{Classes of Adversaries}}
\label{sec:1.2.1}


Two types of adversaries that are not considered in this paper are also widely studied in the community.

\paragraph{Oblivious}
Suppose James is not aware of the truly transmitted codeword $\cwd$. In this case, he has to select the adversarial noise $\st$ entirely \textit{blindly}.  Lapidoth~\cite{ganti2000mismatched} studied additive noise channels with power-constrained (but arbitrary) noise.Using the same terminology in~\cite{haddadpour2013avcs}, we say James is an \textit{oblivious} adversary. The corresponding channel capacity is known in~\cite{hughes1987gaussian,haddadpour2013avcs}  to be $C_{\mathsf{ob}}=\frac{1}{2}\log_2\left(1+\frac{\spal}{\npal}\right)$ if $\spal>\npal$ and $0$ otherwise. 


\paragraph{Omniscient}
On the other hand, suppose James knows exactly the transmitted codeword $\cwd$ before selecting the adversarial noise $\st$. We say such an adversary is \textit{omniscient}. In this case, the exact channel capacity $C_{\mathsf{omni}}$ is still unknown. 
Some upper and lower bounds can be found in~\cite{blachman1962capacity,rankin1955closest}. In particular, a lower bound corresponding to the quadratically constrained version of the Gilbert–Varshamov (GV) bound can be found in~\cite{blachman1962capacity}, and an upper bound corresponding to the quadratically constrained version of the Plotkin bound can be found in~\cite{rankin1955closest}.
The former shows that a positive rate is possible even against an omniscient adversary if $\spal>2\npal$, and the latter shows that the capacity of a channel with an omniscient adversary is zero when $\spal\leq 2\npal$. Generalizations of both results will be useful in our causal quadratically constrained arguments as well.


Kabatiansky and Levenshtein in~\cite{kabatiansky1978bounds} derived the tightest known outer bounds\footnote{The linear programming (LP) bound they derived was for spherical codes, but it can be directly extended to an outer bound on sphere packings.}. In Figure~\ref{Fig:1}, we plot the GV-type bound in~\cite{blachman1962capacity} and the LP bound in~\cite{kabatiansky1978bounds} , together with the omniscient capacity $C_{\mathsf{ob}}$, as references for our results in this work.

\paragraph{Causal}
The primary focus of this work is on {\it causal adversaries}.
Channels with causal adversaries can be considered as a special case of arbitrarily varying channels (AVCs)~\cite{blackwell1960capacities,lapidoth1998reliable}. The \textit{causality} assumption is physically reasonable in many engineering situations. In~\cite{csiszar2011information} (\textit{c.f.,} in page $224$), an arbitrarily ``star'' varying channel is introduced such that at each time step $\indt$, the state $\slt_\indt$ can \textit{only} depend on previously transmitted symbols $\clwd_1,\ldots,\clwd_{\indt-1}$. Yet, as mentioned in~\cite{csiszar2011information}, the general problem has not been fully tackled. 
It turns out that techniques from previous works on AVCs cannot be applied directly to channels with causal adversaries. The capacities of channels with causal adversaries are known in some special cases. For example, recent papers by Dey \textit{et al.}~\cite{dey2013upper} and Chen \textit{et al.}~\cite{chen2015characterization} characterize the capacity region of binary bit-flip channels with causal (online) adversaries. The results in~\cite{chen2016capacity} extend these techniques to characterize the capacities of $q$-ary additive-error/erasure channels. Also, the capacities of binary erasure channels with causal adversaries are known by~~\cite{chen2015characterization,bassily2014causal}. Dey \textit{et al.}~\cite{dey2010coding,dey2016bit} considered a ``delayed'' causal adversary such that the state $\slt_\indt$ is only decided by $\clwd_1,\ldots,\clwd_{\indt-\Delta}$ where $\Delta$ can be an arbitrarily small (but constant) fraction of
the code block-length $\cl$.


At the risk of missing much of the relevant literature, we summarize below the known results, parametrized by the delay parameter $\Delta$ going from $-\cl$ to $\cl$.

In general, adversaries with causality constraints may be weaker than omniscient adversaries, and stronger than oblivious adversaries. Hence characterization of the capacity of causal adversaries may help with characterization of the hard open problem of characterizing the communication capacity in the presence of omniscient adversaries. 
For notational simplicity, in the remaining sections we often omit the subscript and write $C=C_{\mathsf{cau}}$. 




 Our main contributions are presented as below.
\subsection{Main Contributions}
\label{sec:1.3}
As our main result, we show that the capacity is the $\limsup$ of the sequence $\{C_\cl\}_{\cl \geq 1}$ as given in (\ref{eq:0}).
Here, each $C_\cl$ is the optimal objective value of the following optimization:
\begin{mybox}{Scaled Babble-Push Optimization with Optimal Value $C_\cl$}
	\begin{align}
	\tag{P1.1}
	\underset{\spa}{\text{sup}} \quad \underset{\npa}{\text{inf}} \ \underset{1\leq \dpr\leq\cl}{\text{min}} \quad
	\frac{1}{2\cl}\sum_{\indt=1}^{\dpr}&\log\frac{\spal_{\indt}}{\npal_{\indt}}\\
	\tag{P1.2}
	\text{subject to } \qquad \	\sum_{\indt=1}^{\cl}&\spal_{\indt}\leq\cl \spal,\\
	\tag{P1.3}
	\sum_{\indt=1}^{\cl}&\npal_{\indt}\leq\cl \npal,\\
	\tag{P1.4}
	\cl \spal-\sum\limits_{\indt=1}^{\dpr}\spal_{\indt} \leq 2&\cl\npal-\sum\limits_{\indt=1}^{\dpr}2\npal_{\indt},\\
	\tag{P1.5}
	\npal_{\indt}\leq \spal_{\indt}, \text{for all } &\indt=1,\ldots,\dpr,\\
	\tag{P1.6}
	\spal_{\indt}>0, \npal_{\indt}>0, \text{for all } &\indt=1,\ldots,\cl.
	\end{align}
\end{mybox}
The vectors $\spa=\spal_1,\ldots,\spal_\cl$ and $\npa=\npal_1,\ldots,\npal_\cl$ are both length-$\cl$ vectors with non-negative values from ${\mathbbm R}$, and respectively chosen from the \textit{signal power set} $\spas$ and \textit{noise power set} $\npas$ defined in Section~\ref{sec:3.1}. They 
can be considered as the {\it per coordinate average power allocation} for Alice’s codewords $\cwd$, and James’ jamming vectors $\st$, satisfying respectively:
\begin{align*}
    \spal_{\indt},\npal_\indt&>0, \quad \text{ for } \indt=1,\ldots,\cl, \mbox{ and }\\
	\sum_{\indt=1}^{\cl}\spal_{\indt}&\leq\cl \spal,\\
	\sum_{\indt=1}^{\cl}\npal_{\indt}&\leq\cl \npal,.
\end{align*}
The inner minimization in the optimization problem in (P1.1) is over all coordinates $\dpr \in \{1,\ldots,\cl\}$. Further, note that besides satisfying the signal and jamming power constraints in (P1.2) and (P1.3), and the positivity constraints in (P1.6), Alice and James’ power allocations are also required to satisfy the {\it energy bounding condition} (P1.4). Constraint (P1.5) guarantees that the objective function is always nonnegative.

Some intuition connecting the physical meaning of these constraints and corresponding optimization, and the underlying causal capacity problem, will be provided in Section~\ref{sec:3.1.3} later. Furthermore, Figure~\ref{Fig:3} in Section~\ref{sec:3.1.3} gives a pictorial representation.

 In particular, based on the optimization above, we show the following converse and achievability, proved in Section~\ref{sec:4} and Section~\ref{sec:5} respectively.
\begin{itemize}
\item \textit{\underline{Converse}:}
Borrowing the idea of the ``babble-and-push'' attack from~\cite{dey2013upper} (also~\cite{sarwate2012avc,blachman1962capacity}), we design a new attack for James that is called ``scaled babble-and-push'' in Section~\ref{sec:4}. Based on that attack, we show that for sufficiently large block-length $\cl$, any code (either deterministic or stochastic) with rate larger than ${C}_{\cl}$ 
has a non-vanishing \textit{average} probability of error. Hence, $C\leq\limsup_{\cl\rightarrow\infty}{C}_{\cl}$. We summarize our claim formally in Theorem~\ref{thm:1}.
\item \textit{\underline{Achievability}:}
Motivated by the stochastic encoder designed in~\cite{chen2015characterization}, we construct an ensemble of concatenated codes (with independent stochasticity in each chunk) in Section~\ref{sec:5}. We show that for $\cl$ sufficiently large, there exists a concatenated stochastic code with rate smaller than ${C}_{\lfloor\sqrt{\cl}\rfloor}$ such that the corresponding \textit{maximal} probability of error is asymptotically zero (in $\cl$). Hence, $C\geq\limsup_{\cl\rightarrow\infty}{C}_{\lfloor\sqrt{\cl}\rfloor}$. We summarize our claim formally in Theorem~\ref{thm:2}.
	\item \textit{\underline{Channel Capacity}:}
Corollary~\ref{corollary:1} combines the achievability and the converse to show a tight characterization of the channel capacity, as the $\limsup_{\cl\rightarrow\infty}C_\cl$ (when $\spal>2\npal$). However, it is not immediately clear that this optimization is numerically tractable even for fixed (but large) $\cl$, since in principle it would involve optimizing over ${\mathbbm R}^\cl\times {\mathbbm R}^\cl \times \{1,\ldots,\cl\}$. We thus provide both upper and lower bounds on $C_\cl$ in Section~\ref{sec:3.e}.  Interestingly, the optimizing codebook for may not have a uniform power allocation. 
	
Figure~\ref{fig:1} below summarizes the known results. The two dotted curves represent upper/lower bounds on $C_{\mathsf{omni}}$. The solid curve is of values equal to the oblivious capacity $C_{\mathsf{ob}}=\frac{1}{2}\log_2\left(1+\frac{\spal}{\npal}\right)$ when $\spal\geq 2\npal$ and zero otherwise. The blue and red curves are lower bound and upper bound respectively on $C_\cl$ with $\cl=500$. 
\end{itemize}

\begin{figure}[h]	
	\centering
	\includegraphics[scale=0.61]{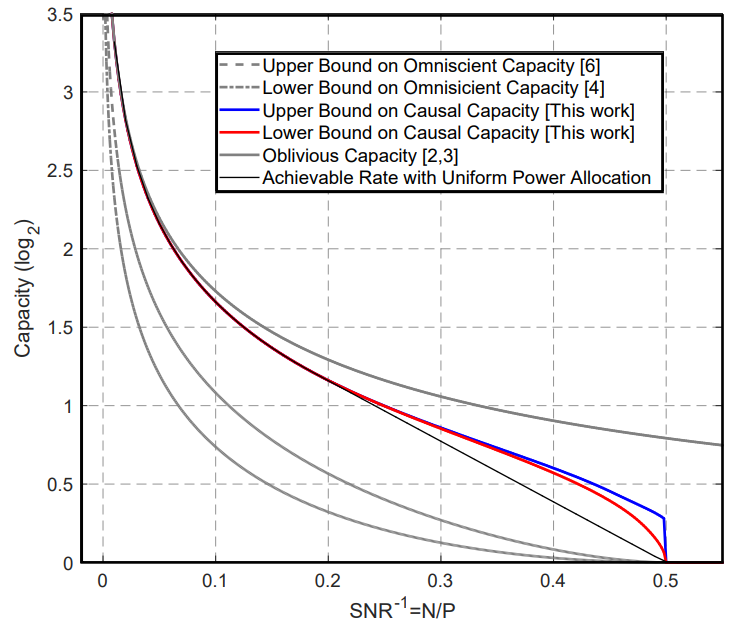}
	\caption{Summary of known bounds for quadratically constrained channels with omniscient/oblivious adversaries together with bounds on $C_\cl$ ($\cl=500$) derived in this work.}
	\label{Fig:1}
\end{figure}

\subsection{Outline}
The remaining content is organized as follows. In Section~\ref{sec:2} we specify our model. Our main results are summarized in Section~\ref{sec:3}. In Section~\ref{sec:4}, we describe the ``scaled babble-and-push'' attack for James and provide a sketch of proof for our converse. A stochastic concatenated code construction is given in Section~\ref{sec:5}, which implies the achievability. The detailed proofs are provided in Appendix~\ref{app:4}.

\section{Model and Preliminaries}
In this section, we describe the channel model being considered. Before going to the definitions, we first formalize the notation for the remaining subsections.

\subsection{Notation}
Let $||\cdot||$ denote the $\ell_2$-norm. Throughout the paper, let $\log\left(\cdot\right)$ denote the logarithm with base $2$. We use $\expc\left[\cdot \right]$ to denote the expectation of random variables if the underlying probability distribution is clear. The symbol $\left(\cdot\right)\triangleq\left(\cdot\right)$ appears in the first time when the notation on the LHS is defined to be that of the RHS.

\subsubsection{Sequences}
Boldface lowercase letters $\cwd$,  $\rwd$ and $\st$ are used to represent particular realizations of the transmitted and received codewords. 
Correspondingly, when capital letters such as $\Cwd$, $\Rwd$ or $\St$ are used, we consider the transmitted and received codewords as random sequences. We will follow this convention strictly in the remaining subsections. Unhiglighted lowercase letters such as letter $\clwd,\rlwd$ and $\slt$ represent the coordinates of the aforementioned sequences. The distribution or probability density function of a random variable $\Cwd$ is written as $p_{\Cwd}$. Sometimes for a well-specified event $\mathcal{E}$, $\Pr_{\Cwd}\left(\mathcal{E}\right)$ means the probability that $\mathcal{E}$ occurs with the randomness of $\Cwd$. Sometime the subscripts are omitted if there is no confusion. In our problem, since $\Cwd$, $\St$ and $\Rwd$ are over continuous alphabets, the probabilities considered are often Lebesgue integrals. 
We only consider those distributions for the random sequences $\Cwd$, $\St$ and $\Rwd$ such that all integrals in the following contexts are well-defined. 

Let $\dpr$ be an integer between $1$ and $\cl$. The \textit{$\dpr$-prefix} $\Cwdp=\Clwd_1\ldots \Clwd_m$ of $\Cwd$ is a subsequence that contains the first $m$ coordinates of $\Cwd$. The \textit{$\dpr$-suffix} $\Cwdl=\Clwd_{\dpr+1}\ldots \Clwd_\cl$ of $\Cwd$ represents the remaining last $\cl-\dpr$ coordinates. The symbol  $\circ$ is used to denoted the concatenation of two sequences. For example, we can write $\Cwd=\Cwdp\circ\Cwdl$.

\subsubsection{Sets}
We use calligraphic symbols to denote sets, \textit{e.g.,} $\ballP,\ballN$, $\ballPN$ and $\code$. When the elements in a set are stochastically generated, we use a slightly different symbol for the set. For example, we use the symbol $\Code$ as the stochastic version for $\code$.

We refer to the three $\cl$-dimensional balls below frequently:
\begin{align*}
\ballP&\triangleq\left\{\cwd\in\mathbbm{R}^\cl:\left|\left|\cwd\right|\right|^2\leq \cl \spal\right\},\\
\ballN&\triangleq\left\{\st\in\mathbbm{R}^\cl:\left|\left|\st\right|\right|^2\leq \cl \npal\right\},\\
\ballPN&\triangleq\left\{\rwd\in\mathbbm{R}^\cl:\left|\left|\rwd\right|\right|^2\leq \cl\left(\sqrt{P}+\sqrt{N}\right)^2\right\}.
\end{align*}


\subsection{Communication Model}
\label{sec:2}

As mentioned above, a channel with quadratic constraints $\spal$ and $\npal$ is a communication system comprising of three parties--a \textit{transmitter} Alice, an \textit{adversary} James, and a \textit{receiver} Bob. During transmission, first, a message $\Msg$ is chosen uniformly at random. Next, based on this \textit{selected message}, Alice encodes the message into a \textit{transmitted codeword} $\Cwd$. Then she transmits it through a contaminated channel where James attacks casually with an additive \textit{adversarial noise} $\St$. The \textit{received codeword} $\Rwd=\Cwd+\St$ then arrives at Bob's side and an \textit{estimated message} $\Esg$ is finally decoded. 

Let $\cl>0$ be an integer representing the \textit{block-length} of $\Cwd,\St$ and $\Rwd$. Next, we give formal definitions of $\Msg,\Cwd,\St,\Rwd$ and $\Esg$.

\subsubsection{\underline{Selected Message $\Msg$}}
Let $\msgs$ denote the source message set, a discrete set with each message $\msg\in\msgs$ having equal probability of being generated by the source. 
 Denote by $W$ the \textit{selected message} defined as a random variable in $\msgs$ with a probability mass function $p_{\Msg}$ such that for all $\msg\in\msgs$,
\begin{align*}
p_{\Msg}\left(\msg\right)=\frac{1}{\mn}.
\end{align*}
\subsubsection{\underline{Transmitted Codeword $\Cwd$}}

For each message $\msg\in\msgs$, the corresponding \textit{transmitted codeword} $\Cwd\left(\msg\right)=\Clwd_1\left(\msg\right),\ldots,\Clwd_{\cl}\left(\msg\right)$ is a random sequence in $\ballP$ specified by a probability density function $p_{\Cwd|\Msg}$ satisfying the property that for all $\msg\in\msgs$, $\cwd\in\ballP$, we have
\begin{align*}
\int_{\cwd\in\ballP} p_{\Cwd|\Msg}\left(\cwd|\msg\right)\mathrm{d}\cwd=1.
\end{align*}
We sometime use the alternative notation $p_{\Cwd\left(\msg\right)}$ for $p_{\Cwd|\Msg}$. We denote $\code\subseteq\ballP$  the \textit{collection of codewords} containing all of the transmitted codewords. 
The {\it collection of codewords} is the set $\code=\bigcup_{\msg=1}^{\mn}\cset$, where each {\it partial collection} $\cset$ in the union is the set containing all possible codewords for a fixed message $\msg$. 

%

\subsubsection{\underline{(Causal) Adversarial Noise $\St$}}
The {\it adversarial noise sequence} $\St=\Slt_1\ldots,\Slt_{\cl}$ is a random sequence in $\ballN$, with a potential dependence (described below) on the transmitted code $\Cwd$.
The corresponding probability density function $p_{\St|\Cwd}$ satisfies for all $\cwd\left(\msg\right)\in\ballP$,
\begin{align*}
\int_{\st\in\ballN} p_{\St|\Cwd}\left(\st|\cwd\right)\mathrm{d}\st=1.
\end{align*}

Moreover, $p_{\St|\Cwd}\left(\st|\cwd\right)$ can be decomposed into a product of $\cl$ conditional probabilities such that for all $\cwd\in\ballP$ and $\st\in\ballN$,
\begin{align}
\label{eq:1.2}
p_{\St|\Cwd}\left(\st|\cwd\right)
&\triangleq\prod_{\indt=1}^{\cl}p_{\Slt_\indt|\St_{\leq \indt-1},\Cwd}\left(\slt_\indt\big |\st_{\leq \indt-1}, \cwd\right).
\end{align}

In particular, we assume the following \textit{causality property}:

\begin{definition}[\textit{Causality Property}]
\label{def:0}
A probability density function $p_{\St|\Cwd}\left(\st|\cwd\right)$ is said to have the \textit{causality property} if each conditional probability $$p_{\Slt_\indt|\St_{\leq \indt-1},\Cwd}\left(\slt_\indt\big |\st_{\leq \indt-1}, \cwd\right)$$ in (\ref{eq:1.2}) above satisfies that for all $\cwd\in\ballP, \st\in\ballN$ and $\indt=1,\ldots,\cl$,
\begin{align}
\nonumber
&p_{\Slt_\indt|\St_{\leq \indt-1},\Cwd}\left(\slt_\indt\big |\st_{\leq \indt-1}, \cwd\right)\\
=&p_{\Slt_\indt|\St_{\leq \indt-1},\Cwd_{\leq \indt}}\left(\slt_\indt\big |\st_{\leq\indt-1},\cwd_{\leq \indt}\right).
\end{align}
\end{definition}

In other words, the $\indt$-th adversarial noise $\Slt_{\indt}$ is independent of the future coordinates $\Cwd_{>\indt}$ conditioned on $\left(\Cwd_{\leq \indt},\St_{\leq \indt-1}\right)$, \textit{i.e.,} $\Slt_{\indt}\longleftrightarrow  \left(\Cwd_{\leq\indt},\St_{\leq \indt-1}\right)\longleftrightarrow  \Cwd_{>\indt}$ is a Markov chain\footnote{In our achievability, we obtain a stronger result by allowing James to know the transmitted message $\msg$. In that case, his strategy is specified by $p_{\St|\Cwd,\Msg}$, and causality means that $\Slt_{\indt}\longleftrightarrow  \left(\Cwd_{\leq\indt},\St_{\leq \indt-1},\Msg\right)\longleftrightarrow  \Cwd_{>\indt}$ is a Markov chain}.


Denote by $\mathsf{P}$ the set of all probability density functions $p_{\St|\Cwd}$ satisfying the causality property above. Moreover, noting that $\Rwd=\Cwd+\St$, we let $\mathsf{Q}$ be the set of all probability density functions $p_{\Rwd|\Cwd}$ with an underlying causal distribution $p_{\St|\Cwd}$.

\subsubsection{\underline{Received Codeword $\Rwd$}}
The \textit{received codeword} $\Rwd$ is a summation of $\Cwd$ and $\St$. By definition, random sequence $\Rwd$ in $\ballPN$ has a marginal distribution $p_{\Rwd}$:
\begin{align}
\label{eq:1.0}
p_{\Rwd}\left(\rwd\right)=\frac{1}{\mn}\sum_{\msg=1}^{\mn}p_{\Rwd|\Msg}\left(\rwd|\msg\right)
\end{align}
where for all $ \rwd\in\ballN$ and $\msg\in\msgs$, the conditional distribution $p_{\Rwd|\Msg}\left(\rwd|\msg\right)$ is
\begin{align}
\label{eq:1.01}
p_{\Rwd|\Msg}\left(\rwd|\msg\right)\triangleq\int_{\cwd\in\cset}
&p_{\Cwd|\Msg}\left(\cwd |\msg\right)p_{\St|\Cwd}\left(\rwd-\cwd|\cwd\right)\mathrm{d}\cwd.
\end{align}


\subsubsection{\underline{Estimated Message $\Esg$}}
Given a received codeword $\rwd\in\ballPN$, an \textit{estimate} $\esg\in\msgs$ is chosen\footnote{Note that we can generalize this definition such that the estimate is not necessarily in the message set $\msgs$. As later in Section~\ref{sec:5.b}, when a decoding error occurs, a symbol $\mathsf{error}$ outside the message set $\msgs$ will be decoded.} according to the conditional distribution $p_{\Esg|\Rwd}$. The marginal distribution $p_{\Esg}$ of the \textit{estimated message} denoted by $\Esg$ is 
\begin{align}
p_{\Esg}\left(\esg\right)=\int_{\rwd\in\ballPN}p_{\Esg|\Rwd}\left(\esg|\rwd\right)p_{\Rwd}\left(\rwd\right)\mathrm{d}\rwd, \quad \ \text{for all} \ \esg\in\msgs.
\end{align}
For fixed block-length $\cl$, number of messages $\mn$ and signal power $\spal$, a $\left(\mn,\cl,\spal\right)$-\textit{code} represented by $\left(p_{\Cwd|\Msg},p_{\Esg|\Rwd}\right)$ is a pair of two distributions for encoding and decoding respectively. There may be more than one codeword corresponding to a single message. We call such a code~\textit{stochastic}, in contrast to a \textit{deterministic} code with a one-to-one mapping between messages and codewords. 

\begin{remark}
	We assume that the distribution $p_{\Cwd|\Msg}$ and $p_{\Esg|\Rwd}$ are known to every party in the system. In other words, there is no secrecy between Alice and Bob and they cannot share any randomness with each other.
\end{remark}

The communication model described above is illustrated as a schematic diagram in Figure~\ref{fig:1}.
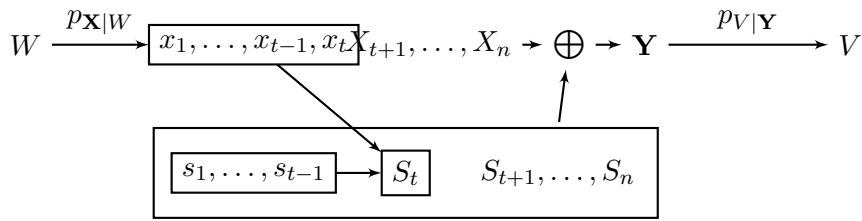
\begin{figure*}[b]
	\begin{center}
		\begin{tikzpicture}[node distance=2cm,auto,>=latex']
		\node (a)  {$\Msg$};
		\node (b1) [right of=a, node distance=3cm] {$\clwd_1 , \ldots ,  \clwd_{\indt-1}, \clwd_\indt$};
		\node (b2) [right of=b1, node distance=2.3cm] {$\Clwd_{\indt+1}, \ldots, \Clwd_{\cl}$};
		\node (c) [right of=b2, node distance=1.85cm] {$\bigoplus$};
		\node (e) [right of=c, node distance=1cm] {$\Rwd$};
		\node (d) [right of=e, node distance=2.7cm] {$\Esg$};
		\node (e0) [below of=b1, node distance =1.7 cm] {$\slt_1 ,\ldots, \slt_{\indt-1}$};
		\node (e1) [right of=e0, node distance =2 cm] {$\Slt_\indt$};
		\node (e2) [right of=e1, node distance =2 cm] {$\Slt_{\indt+1}, \ldots, \Slt_{\cl}$};
		\draw[thick,->] (a) edge node [name=p] {$p_{\Cwd|\Msg}$} (b1);
		\node [coordinate] (end) [above of=p, node distance=1cm]{};
		\draw[thick,->] (b2) edge node {} (c);
	\draw[thick,->] (c) edge node {} (e);
		\draw[thick,->] (e) edge node {$p_{\Esg|\Rwd}$} (d);
		\draw[thick,->] (b1) edge node {} (e1);
		\draw[thick,->] (e0) edge node {} (e1);
		\draw[thick,-]  ($(e1.north west)+(0,0)$) rectangle ($(e1.south east)+(0,0)$);
		\draw[thick,-]  ($(e1.north west)+(-3,0.3)$) rectangle ($(e1.south east)+(3,-0.3)$);
		\draw[thick,-]  ($(e0.north west)$) rectangle ($(e0.south east)$);
		\node (virtual) [above of=e2, node distance=0.54cm] {};
		\draw[thick,->] (virtual) edge node {} (c);
		\draw[thick,-]  ($(b1.north west)$) rectangle ($(b1.south east)$);
		\end{tikzpicture}
	\end{center}
	\caption{System diagram for the $\indt$-th transmission.}	
	\label{fig:1}
\end{figure*}

\subsubsection{\underline{Probability of Error}}
Fix a selected message $\msg\in\mathcal{W}$. Consider an $\left(\mn,\cl\right)$-code $\left(p_{\Cwd|\Msg},p_{\Esg|\Rwd}\right)$. Consider a $\left(\mn,\cl,\spal\right)$-code $\left(p_{\Cwd|\Msg},p_{\Esg|\Rwd}\right)$. According to our definitions,
\begin{align}
\nonumber
&\Pr\left(\Esg\neq\msg|\Msg=\msg \right)=\\
&
\sum_{\esg\neq \msg}\int_{\st}\int_{\cwd}
\label{eq:1}
p_{\Cwd|\Msg}\left(\cwd|\msg\right)p_{\St|\Cwd}\left(\st|\cwd\right)p_{\Esg|\Rwd}\left(\esg|\cwd+\st\right)\mathrm{d}\cwd\mathrm{d}\st
\end{align}
where the integrals are taken over $\st\in\ballN$ and $\cwd\in\code\left(\msg\right)$ respectively.

We now define the two notions of probability of error considered in this work.

\begin{definition}[\textit{Probability of Error}]
	\label{def:1}
	The \textit{average probability of error} is defined to be
	\begin{align}
	\label{eq:1.1}
		\apoe\triangleq\sup_{p_{\St|\Cwd}\in\mathsf{P}}\frac{1}{\mn}\sum_{\msg=1}^{\mn}\Pr\left(\Esg\neq\msg|\Msg=\msg \right).
	\end{align}	
	The \textit{maximal probability of error} considered in this work is defined to be
	\begin{align*}
	\mpoe\triangleq\max_{\msg\in\msgs}\sup_{p_{\St|\Cwd}}\Pr\left(\Esg\neq\msg|\Msg=\msg \right).
	\end{align*}
\end{definition}

\begin{remark}
	Note that the maximal probability of error is always greater than the average probability of error. Therefore, to obtain a slightly stronger result, Theorem~\ref{thm:1} (converse) in Section~\ref{sec:3} involves the average probability of error while the maximal probability of error is used in the Theorem~\ref{thm:2} (achievability).
\end{remark}		
\begin{remark}
	Note that in the average probability $\apoe$, the supremum over the set of causal distributions $p_{\St|\Cwd}$ is taken before the choice of the message $\msg$ and in the average probability $\apoe$, the corresponding supremum is taken after the selection of $\msg$. This gives us a stronger result since without knowing $\msg$, James is still able to conduct his attack in the proof of Theorem~\ref{thm:1}. On the other hand, Bob's decoder used for the proof of Theorem~\ref{thm:2} still works even if James knows the transmitted message a priori.
\end{remark}

With the notions of probability of error as defined above, we are ready to define the achievable rate for a causal channel with quadratic constraints.

\begin{definition}[\textit{Achievable Rate}]
A rate $R(\spal,\npal)$ is {\emph{achievable under an average probability of error criterion}} for a causal channel with quadratic constraints if for any $\varepsilon>0$,
there exists an infinite sequence of $(\mn,\cl,\spal)$-codes (not necessarily one for each $\cl$) satisfying 
\begin{align}
\label{eq:2.3}
\frac{1}{\cl}\log \mn=R(\spal,\npal)
\end{align}
 such that the corresponding average probability of error is bounded from above as
$\apoe<\varepsilon$
 for every positive integer $k$. Analogously, a rate $R(\spal,\npal)$ is {\it achievable under a maximal probability of error criterion} by replacing $\apoe$ with $\mpoe$.
\end{definition}

The \textit{causal capacity} $C_{\mathsf{cau}}$ is defined as the supremum of all achievable rates.

A rate is \textit{achievable} if and only if there exists a code such that are infinitely many block-lengths satisfy equality (\ref{eq:2.3}). 
Thus if we can find a sequence of $(\mn,\cl)$-codes with $\frac{1}{\cl}\log \mn=C_\cl$ such that the corresponding maximal probability vanishes as $\cl$ goes to infinity, then we can bound the capacity from below by $\limsup_{\cl\rightarrow\infty}C_\cl$, for the reason that there always exists a subsequence $\{C_{\cl_k}\}$ of $\{C_\cl\}$ such that
\begin{align*}
\lim_{k\rightarrow\infty} C_{\cl_k}=\limsup_{\cl\rightarrow\infty}C_\cl.
\end{align*}

\section{Main Results}
\label{sec:3}

\subsection{Optimization Formalism}
\label{sec:3.1}
Fix any sufficiently large block-length $\cl\geq 1$.

In this work we convert the problem of characterizing the capacity region of a causal channel with quadratic constraints, to one of optimizing a certain function under certain constraints. 


\begin{definition}[\textit{Signal Power Set}]
	The \textit{signal power set}  $\spas$ denotes the set containing all length-$\cl$ non-negative real sequences $\spa=\spal_1,\ldots,\spal_{\cl}$ satisfying
	\begin{align*}
	\spal_{\indt}&>0, \quad \indt=1,\ldots,\cl,\\
	\sum_{\indt=1}^{\cl}\spal_{\indt}&\leq\cl \spal.
	\end{align*}
\end{definition}

\begin{definition}[\textit{Noise Power Set}]
	The \textit{noise power set} $\npas$ denotes the set containing all length-$\cl$ non-negative real sequences $\npa=\npal_1,\ldots,\npal_{\cl}$ satisfying
	\begin{align*}
	\npal_{\indt}&>0, \quad \indt=1,\ldots,\cl,\\
	\sum_{\indt=1}^{\cl}\npal_{\indt}&\leq\cl \npal.
	\end{align*}
\end{definition}

Based on the two sets above, we reprise the optimization stated in Section~\ref{sec:1.3}.

\begin{mybox}{Reference Optimization with Optimal Value $C_\cl$}
	\begin{equation}
	\label{eq:3.2.7}
					\tag{P1}
	\begin{aligned}
	& &\underset{\spa}{\text{sup}} \ \underset{\npa}{\text{inf}} \ \underset{1\leq \dpr\leq\cl}{\text{min}} \quad \frac{1}{2\cl}\sum_{\indt=1}^{\dpr}&\log\frac{\spal_{\indt}}{\npal_{\indt}}\\
	& &\text{subject to }  \qquad \qquad \spa &\in\spas,\\
	& &  \npa &\in\npas,\\
	& &  \npal_{\indt} &\leq \spal_{\indt} \\
	& & \text{for all } \indt&=1,\ldots,\dpr,\\
	& &  \cl \spal-\sum\limits_{\indt=1}^{\dpr}\spal_{\indt} \leq& 2\cl\npal-\sum\limits_{\indt=1}^{\dpr}2\npal_{\indt}.
	\end{aligned}
	\end{equation}
\end{mybox}

Instead of directly describing the optimization problem (\ref{eq:3.2.7}), we consider a closely related optimization problem (\ref{eq:3.2.5}) defined in Section~\ref{sec:robust}. This optimization problem corresponds to a specific jamming strategy that James can follow. The only difference between (\ref{eq:3.2.7}) and (\ref{eq:3.2.5}) is a slackness parameter $\tau$, discussed below.

We now provide some intuition to motivate the connection between the optimization problem above and the underlying physical communication problem.

\subsection{Intuition behind the Formalism}

\label{sec:3.1.3}
Consider the sets $\spas$ and $\npas$. The sequences $\spa\in\spas$ and $\npa\in\npas$ can be regarded as per coordinate average power allocations decided by Alice (for her transmitted codebook) and James (for his jamming sequence). For them, the total budgets of power are $\cl\spal$ and $\cl\npal$ respectively. That is, $\spal_\indt$ is the average (over the message probability distribution $p_\Msg$ and the codeword probability distribution $p_{\Cwd|\Msg}$) of $(\Clwd_\indt)^2$, and $\npal_\indt$ is the average (over Alice’s message probability distribution $p_\Msg$, the codeword probability distribution $p_{\Cwd|\Msg}$, and James' (causal) jamming distribution $p_{\St|\Cwd}$) of $(\Slt_\indt)^2$. Note that Alice has to design her codebook without knowing the specific jamming sequence that James will instantiate, hence the sequence $\{\spal_1,…,\spal_\cl\}$ is a function only of $p_\Msg$ and $p_{\Cwd|\Msg}$. On the other hand, James can choose his jamming sequence as a function (satisfying the causality condition) of $p_\Msg$, $p_{\Cwd|\Msg}$ and $p_{\St|\Cwd}$. 

Due to the causality constraint on James, he does not know the transmitted codeword $\cwd$ exactly. However, he can still learn information from the probability distribution of $\Cwd$, given the transmitted message. 
Since James knows the distribution $p_{\Cwd|\Msg}$, each $\indt$-th expectation $\expc[\left|\Clwd_{\indt}\right|^2]$ is available to James. 
\begin{definition}[$\indt$-th \textit{Average Power}]
	\label{def:7}
	The $\indt$-th \textit{average power} denoted by $\spal_{\indt}$ of an $\left(\mn,\cl\right)$-code $\left(p_{\Cwd|\Msg},p_{\Esg|\Rwd}\right)$ is defined as the expectation of each $\Clwd_{\indt}$
	\begin{align}
	\label{eq:3.3}
	\spal_{\indt}\triangleq \expc\left[\left|\Clwd_{\indt}\right|^2\right]=\frac{1}{\mn}\sum_{\msg=1}^{\mn}\int_{\clwd_{\indt}\in\mathbbm{R}}\left|\clwd_{\indt}\right|^2 p(\clwd_\indt|\msg)\mathrm{d}\clwd_{\indt}.
	\end{align}
\end{definition}

Abusing the notation $\spa$, denote by $\spa=\spal_{1},\ldots,\spal_{\cl}$ the sequence of average powers. We call $\spa$ the \textit{average power allocation sequence}. Since the transmitted codeword $\cwd$ is in the $\cl$-dimensional ball $\ballP$, it follows that $\spa\in\spas$. 

Knowing $\spa$, James selects his own power allocation represented by a sequence $\npa=\npal_1,\ldots,\npal_\cl$ in $\npas$. Moreover, James uses a two-stage attack that we call \textit{scaled-babble and push}. He first selects a \textit{division point} $\dpr$. Based on the point $\dpr$, James attacks the $\dpr$-prefix $\cwdp$ and the $\dpr$-suffix $\cwdl$ differently.

For all $\indt \leq \dpr$ in the prefix, James chooses each coordinate of his jamming sequence $\Slt_\indt$ as the sum of a deterministic and a stochastic component. The deterministic component corresponds to $-\alpha \Clwd_\indt$ (justifying the word ``scaled” in the name scaled-babble and push, with $\alpha$ equaling ${\npal^*_\indt}/{\spal_\indt}$ where $\spa$ is the average power allocation sequence defined in Definition~\ref{def:7} above and $\npa^*$ corresponds to the optimal solution of~(\ref{eq:3.2.5}) with $\spa$ fixed. The stochastic component corresponds to zero-mean Gaussian noise (justifying the word ``\textit{babble}”) with variance $\npal^*_{\indt}(1-{\npal^*_{\indt}}/{\spal_{\indt}})$.

For each coordinate $\indt > \dpr$ in the suffix, James tries to actively ``push” the suffix of $\clwd$ towards some codeword $\hlwd$ corresponding to a message other than the one Alice is actually transmitting. Specifically, the noise $\slt_\indt$ added equals $(\hlwd_\indt - \clwd_\indt)/2$. 

Analytically, we show in Lemma~\ref{lemma:3} (Inequality~(\ref{eq:mutual})) that the objective function $\frac{1}{2\cl}\sum_{\indt=1}^{\dpr}\log\frac{\spal_{\indt}}{\npal_{\indt}}$ corresponds to the normalized mutual information between transmitted codeword prefix $\Cwdp$ and and the received codeword prefix $\Rwdp$ over the AWGN instantiated by the first stage (scaled-babble) of James’ attack. This can be used to show that with significant probability, the $\hlwd$ chosen by James to push in the second stage corresponds to a message $\usg$ different than Alice’s true message $\msg$.

Then in the push stage, the energy-bounding condition plays a part in our optimization framework since we can show (in Lemma~\ref{lemma:3}) that if the following inequality
\begin{align}
\label{eq:energy-bounding}
\cl \spal-\sum\limits_{\indt=1}^{\dpr}\spal_{\indt} \leq 2\cl\npal-\sum\limits_{\indt=1}^{\dpr}2\npal_{\indt}
\end{align}
is satisfied, then James's attack can be successful with positive probability, as Theorem~\ref{thm:1} states.
Roughly speaking, this is because one can use a version of the Plotkin bound~\cite{plotkin1960binary} to show that ``not too many” pairs of codewords can be ``too far apart”.

Since the division point $\dpr$ of the stages can be chosen anywhere between $1$ and $\cl$ by James, by first minimizing over $1\leq\dpr\leq\cl$ and the sequence $\npa\in\npas$ and then maximizing over the sequence $\spa\in\spas$ (or, this can be considered as a maximization over the distribution of $\Cwd$), we form the optimizations~(\ref{eq:3.2.5}) and~(\ref{eq:3.2.6}) in Section~\ref{sec:robust}. Solving~(\ref{eq:3.2.5}) therefore gives an upper bound on the causal capacity, since if Alice tried to transmit at a rate higher than the optimizing value of~(\ref{eq:3.2.5}), the arguments above ensure that James’ scaled-babble and push attack works with positive probability.

Arguing that essentially the same rate (as the optimal value of~(\ref{eq:3.2.6})) is achievable {\it regardless} of which causal jamming strategy James employs requires a different argument. To this end, we analyze yet another optimization problem (\ref{eq:3.2.6}) given in Section~\ref{sec:robust}.

Again, note the differences between (\ref{eq:3.2.6}) and (\ref{eq:3.2.7})/(\ref{eq:3.2.5}) (apart from changes in the names of variables). One (minor) change is that the slackness $\gamma$ is in the opposite direction from the slackness $\tau$ in the converse optimization (\ref{eq:3.2.5}) — this ``slackness reversing” phenomenon is common in many communication problems. Another change is that the code construction arising from (\ref{eq:3.2.6}) will have the $\cl$ coordinates distributed into $\cul$ chunks\footnote{The specific value of $\cul$ does not matter too much, and can be chosen from a wide range — for concreteness, we later set it to equal $\sqrt{n}$.}, and hence the variables $\spadl$ and $\npadl$ actually denote the average power per chunk, rather than per coordinate.

Our code comprises of $\cul$ chunks with independent stochasticity in each chunk. Specifically, for each message $\msg\in\msgs$, and each chunk $\indT\in \{1,\ldots,\cul\}$, we choose $2^{\beta}$ ($\beta$ is a constant specified in (\ref{eq:5.20})) codewords uniformly at random from the surface of a $\cl/\cul$-dimensional ball of radius $\sqrt{\spadl^*_\indT}$. Here $\spadl^*_\indT$ is the value for $\spadl_\indT$ arising from the optimization (\ref{eq:3.2.6}). 
Hence for each message $\msg$, in each chunk there are $2^\beta$ many potential codewords that Alice can transmit, and indeed, Alice chooses one of them uniformly at random to transmit.

Before discussing Bob’s decoder, a short discussion of {\it list-decoding} in the context of quadratically constrained channels is in order. List-decoding is a powerful primitive introduced by Elias~\cite{elias1957list} that guarantees that if a suitable code is used by Alice and Bob, even an omniscient jammer James is unable to confuse Bob ``too much” — he can {\it always} ensure that given his observation, Bob can use an appropriate decoder and ensure that Alice’s transmitted message is within a small list. In our specific scenario with quadratic constraints, it turns out that the objective function in (\ref{eq:3.2.6}) corresponds to the list-decoding capacity for a transmission of length $\dpc$. This is not a coincidence — as shown in~\cite{sarwate2012list} (albeit not for continuous alphabet channels, and without an input power constraint), the list-decoding capacity of an AVC can be written as the mutual information between the transmission $\Cwd$ and the received vector $\Rwd$ (minimized over all possible stochastic channels that can be instantiated by James). Further, as shown by Sarwate in~\cite{sarwate2012avc}, this mutual information has per-coordinate form $(1/2)\log(\spal_\indt/\npal_\indt)$. There are indeed some technical differences between ``usual” list-decoding and the notion we need in this work\footnote{(i) The chunked structure of our codes is somewhat different than the ``usual” random code ensemble used in the analysis of list-decoding. Nonetheless, careful analysis shows that list-decoding is possible even for most codes drawn from the ensemble of codes in this work. 
(ii) Since we use chunk-wise stochastic encoding, but Bob only cares about Alice’s message, not her transmitted codeword, we distinguish between {\it message list-decoding} and {\it codeword list-decoding}. It turns out that the former suffices for our purpose — indeed, the latter is in general not possible, since James can add ``a lot of noise” to some chunks, and therefore introduce ``a lot of confusion” about the specific stochastic codeword transmitted in those chunks. (iii) Another  relatively straightforward issue pertains to the fact that Alice and James may use non-uniform power allocations. Concavity properties of the logarithm function allow us to generalize list-decoding even to such non-uniform distributions.}, but these can be resolved with some thought.

We then show in Lemma~\ref{lemma:4} and~\ref{lemma:5}  that regardless of the jamming sequence James chooses, there always exists a certain critical chunk index $\dpc_1$ (potentially but not necessarily related to James’ division point $\dpc$) such that Bob can list-decode to a ``small” set of messages using the prefix $\rwdp$, and the suffix $\rwdl$ satisfies the energy-bounding condition in~\ref{eq:energy-bounding} (with appropriate slack).

This suggests a structure for Bob’s decoder — that he list-decodes using the prefix, and then somehow uses the suffix to whittle down the list to a single element. However, two challenges immediately present themselves. For one, it is unclear what value Bob should use for $\dpc$, since based on his observation of the sum of Alice’s transmission and James’ jamming sequence it isn’t clear how he can estimate $\dpc$. For another, even if he were able to estimate $\dpc$, it isn’t clear how Bob can reduce his list-size — as coding theory shows us, omniscient jammers can be very powerful/confusing.

To handle the first issue, Bob’s decoder is structured as an {\it iterative} decoder — Bob first ``guesses” a potential value $\dpc_0$ (this is called the \textit{starting point of decoding} of the iterations defined in Definition~\ref{def:8}) for $\dpc$, and, knowing Alice’s message rate calculates an upper bound ${F}_{\dpc_0}$ (this is the first coordinate of the  \textit{budget reference sequence} $\cbd$ defined in Definition~\ref{def:3}) on how much adversarial noise up to chunk $\dpc_0$ could be tolerated if Bob were to list-decode at this chunk. He then checks to see if this value $\dpc_0$ is ``plausible”. ``Plausible'' here means that if Bob tried to list-decode using $\dpc_0$, then there should be a  ``reasonable” suffix $\Clwd_{\leq\dpc_0}$ (in Alice’s code) corresponding to {\it exactly one} message in the list of messages obtained via to the prefix. ``Reasonable” means that this suffix $\Clwd_{>\dpc_0}$ is relatively close to Bob’s observed suffix $\Rlwd_{>\dpc_0}$, about $\sqrt{\cl\npal-{F}_{\dpc_0}}$. If this chunk does turns out to be plausible, then Bob outputs message corresponding to this ``reasonable” suffix. If not, Bob increments $\dpc_0$ by one and repeats.

Our analysis of this encoder-decoder pair then proceeds by showing two facts. First, we show that if Bob’s estimate of $\dpc$ is indeed “correct, then with high probability over the stochasticity in Alice’s encoding, Bob’s decoder outputs the correct answer. Second, we show that for incorrect values of $\dpc$, with high probability over the stochasticity in Alice’s encoding, Bob’s decoder detects that this $\dpc$ is implausible.

Both these arguments rely critically on the fact that when James is choosing his jamming sequence in chunk $\indT$ he has no way of knowing the stochasticity that Alice will use in future chunks (even though James may well be able to decode Alice’s message fairly early on). Hence, regardless of the list he chooses to impose on Bob via the prefix, with high probability over the specific suffix that Alice transmits, there will be no ``reasonably close” suffix for {\it any} $\usg\neq \msg$ in this list. Proving such a fact requires one to prove a somewhat subtle ``code goodness” property, analogous to the one in~\cite{chen2015characterization}, which may be viewed as a generalization of a Gilbert-Varsahmov-type property\footnote{The potential subtlety, and the difference from a Gilbert-Varshamov-type {\it worst-case} distance guarantee arises from the fact that we only require such code goodness for {\it most} possible suffixes conditioned in the prefix James observes, rather than {\it all} prefixes. Indeed, requiring a ``for all” rather than a ``for most” guarantee would be too ambitious, as can be seen by noting that if James were to know the specific suffix that Alice would transmit, then he could potentially choose the prefix list of messages to impose on Bob in a manner so that there would be a corresponding suffix for some $\usg\neq \msg$ in the list that would {\i also} be ``reasonably” close. Hence James’ causal restriction is critically used here.}.

Below we visualize the energy-bounding condition by plotting Alice’s remaining energy $\sum_{\indt=\indt_1+1}^{\cl}\spal_\indt$, and James’ remaining energy $\sum_{\indt=\indt_1+1}^{\cl}\npal_\indt$ multiplied by $2$  as two decreasing functions parametrized by $\indt_1\in \{1,\ldots,\cl\}$.
 In this way, the first index where these two curves intersect equals  the optimal \textit{division point} with fixed power allocations $\spa$ and $\npa$.

\begin{figure}[H]	
	\centering
	\includegraphics[scale=0.45]{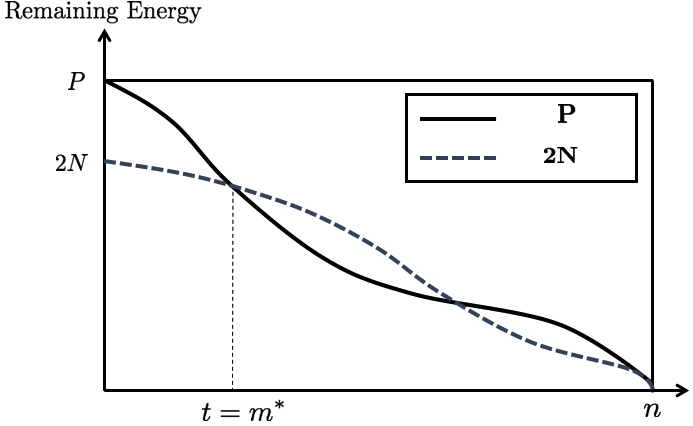}
	\caption{
		A graphical representation of a sample realization of Alice’s power allocation sequence $\spa$ and James’ power allocation sequence (normalized by a factor of 2) $\npa$. The $x$-axis denotes the time index $\indt$, the solid curve denotes Alice’s residual power ($\cl\spal-\sum_{i=1}^\indt \spal_i$), and the dashed curve denotes James’ (normalized) residual power ($2\cl\npal-2\sum_{i=1}^{\indt} \npal_i$).
For the given P and N sequences, the first point of intersection of the two curves is the minimizing $\dpr$ in optimization (\ref{eq:3.2.7}). When considering optimizations (\ref{eq:3.2.5}) and (\ref{eq:3.2.6}), the plots of $\spa$ and $2\npa$ should be slightly tweaked to reflect the slackness parameters $\tau$ and $\gamma$ respectively — we do not depict these impacts in this figure.}
	\label{Fig:3}
\end{figure}

\subsection{Converse}
\label{sec:3.3}
Based on the optimal value ${C}_{\cl}$, below we state our converse result for causal channels with quadratic constraints.

\begin{theorem}[\textit{Converse}]
	\label{thm:1}
	Consider a causal channel with quadratic constraints $\spal>0$ and $\npal>0$. Let $\varepsilon>0$. For any code with rate satisfying $R={C}_{\cl}+3\varepsilon$, the corresponding average probability of error can always be bounded from below as $\apoe=\Omega\left(\varepsilon\right)$ for any block-length $\cl$ sufficiently large.
\end{theorem}

Note that from Theorem~\ref{thm:1} we can deduce that any rate $R\geq C_{\cl}+3\varepsilon$ is not achievable, since the corresponding average probability of error is always bounded from below by a constant that is independent of the block-length $\cl$. 

\textit{\underline{Outline of Proof}:}

The proof follows by specifying a causal adversary based on the optimizing $\npa$ in (\ref{eq:3.2.7}). The adversary uses a particular attack, called \textit{scaled-babble and push}. This attack successes with a constant probability given a large enough block-length $\cl$. The details can be found in Section~\ref{sec:4}. 
\hfill $\blacksquare$

\subsection{Achievability}
We also have the following achievability result.


\begin{theorem}[\textit{Achievability}]
	\label{thm:2}
	Consider a causal channel with quadratic constraints $\spal>0$ and $\npal>0$. Let $\varepsilon>0$. There exists a code with rate satisfying $R={C}_{\lfloor\sqrt{\cl}\rfloor}-\varepsilon$ and the corresponding maximal probability of error satisfying
	\begin{align*}
	\mpoe=&\exp\left({-{\Omega\left(\left(\cl+\ln\varepsilon\right)e^{\sqrt{\cl}}-\frac{\cl}{\varepsilon}\right)}}\right)+\exp\left({-\Omega\left(\cl\right)}\right)
	\end{align*}
	for any block-length $\cl$ sufficiently large.
\end{theorem}

\textit{\underline{Outline of Proof}:}

Using the optimizing power allocation sequence $\spa$ in optimization (\ref{eq:3.2.7}), Alice generate a stochastic code by concatenating independent chunks of sub-codewords. We show that the generated code ensures a vanishing maximal probability of error under any possible causal attack of James.

The proof sketches of Theorem~\ref{thm:1} and Theorem~\ref{thm:2} above can be found in Section~\ref{sec:4} and Section~\ref{sec:5} respectively with detailed proofs provided in Appendix~\ref{app:2}.
\hfill $\blacksquare$

Corollary~\ref{corollary:1} combines the achievability and the converse to show a tight characterization of the channel capacity. 
\begin{corollary}[\textit{Channel Capacity}]
	\label{corollary:1}
	Consider a causal channel with quadratic constraints $\spal>0$ and $\npal>0$. The channel capacity $C_{\mathsf{cau}}$ satisfies
	\begin{align*}
	C_{\mathsf{cau}}\left(\frac{\spal}{\npal}\right)=
	\limsup_{\cl\rightarrow\infty}C_\cl\left(\frac{\spal}{\npal}\right).
	\end{align*}
\end{corollary}

\subsection{Analytical Bounds on $C_\cl$}
\label{sec:3.e}
Next, we provide both lower and upper bounds on $C_\cl$, by restricting the sets corresponding to maximization and minimizations respectively. We consider the following subset of $\spas$:
\begin{definition}[\textit{Two-level Power Sets}]
	\label{def:5}
	A \textit{restricted signal power set} $\uspas$ denotes a subset of $\spas$ that contains all \textit{two-level} length-$\cl$ non-negative real sequences $\spa=\spal_1,\ldots,\spal_{\cl}$ satisfying
	\begin{align*}
	&\spal_{\indt}=\begin{cases}
	\underline{\spal}, &\quad  1\leq \indt\leq\nu\\
	\overline{\spal}, &\quad \nu<\indt\leq\cl
	\end{cases}, \quad \spa\in\spas
	\end{align*}
	for some constants $\underline{\spal},\overline{\spal}>0$ and some \textit{transition point} $\nu\in\left\{1,\ldots,\cl\right\}$.
	
	Similarly, a \textit{restricted noise power set} $\unpas$ consists of all sequences $\npa=\npal_1,\ldots,\npal_{\cl}$ satisfying
	\begin{align*}
	&\npal_{\indt}=\begin{cases}
	\underline{\npal}, &\quad  1\leq \indt\leq\nu\\
	\overline{\npal}, &\quad \nu<\indt\leq\cl
	\end{cases}, \quad \npa\in\npas
	\end{align*}
	for some constants $\underline{\npal},\overline{\npal}>0$ and some \textit{transition point} $\nu\in\left\{1,\ldots,\cl\right\}$
\end{definition}
The subset $\uspas$ ($\unpas$) contains all \textit{``two-level''} sequences with coordinates only taking two possible values \textit{consecutively}. Figure~\ref{Fig:4} below illustrates a typical two-level sequence in $\uspas$.
\begin{figure}[H]	
	\centering
	\includegraphics[scale=0.45]{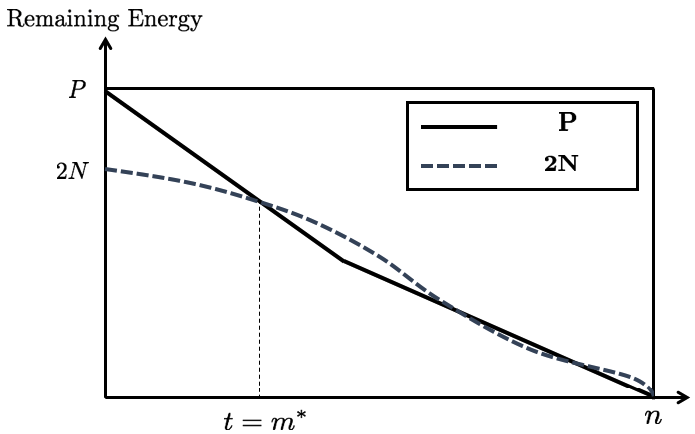}
	\caption{A sample realization of a two-level uniform power allocation for Alice. Note that in this figure the solid curve comprises of two line-segments — the (negative of the) slope of the first line segment corresponds to the first power-level $\underline{\spal}$, and the (negative of the) slope of the second line segment corresponds to the second power-level $\overline{\spal}$. The first power-level $\underline{\spal}$ is not necessarily larger than the second power-level $\overline{\spal}$.}
	\label{Fig:4}
\end{figure}

Based on the notion of the restricted signal power set $\uspas$, we bound $C_\cl$ from both below and above in the following theorem:
\begin{theorem}[\textit{Upper and Lower Bounds}]
	\label{thm:4}
For any block-length $\cl$, $C_\cl$ can be bounded as
	\begin{align*}
	C_\cl&\leq \overline{C}_\cl\triangleq \adjustlimits\max_{1\leq\nu\leq\cl}\sup_{\spa\in\uspas}\inf_{\npa\in\mathcal{N}\left(\nu,\spa\right)}
	\frac{1}{2\cl}\sum_{\indt=1}^{\nu}\log\frac{\spal_{\indt}}{\npal_{\indt}},\\
	C_\cl&\geq \underline{C}_\cl\triangleq \adjustlimits\max_{1\leq\nu\leq\cl}\sup_{\spa\in\uspas}\min_{1\leq\dpr\leq\cl}\inf_{\npa\in\npaseu}
	\frac{1}{2\cl}\sum_{\indt=1}^{\dpr}\log\frac{\spal_{\indt}}{\npal_{\indt}}
	\end{align*}
where $\npaseu\subseteq\npas$ denotes the set containing all $\npa$ satisfying the constraints in (\ref{eq:3.2.7}):
	\begin{align*}
	&\npal_{\indt}\leq \spal_{\indt}, \quad \indt=1,\ldots,\dpr,\\
	&\cl \spal-\sum\limits_{\indt=1}^{\dpr}\spal_{\indt} \leq 2\cl\npal-\sum\limits_{\indt=1}^{\dpr}2\npal_{\indt},\\
	&\cl \spal-\sum\limits_{\indt=1}^{\indt_0}\spal_{\indt} > 2\cl\npal-\sum\limits_{\indt=1}^{\indt_0}2\npal_{\indt}, \quad\quad\text{for all }\indt_0<\dpr.
	\end{align*}
\end{theorem}

Note that if $\npaseu=\emptyset$ (or $\uspas=\emptyset$), the corresponding objective value is set to be positive infinite.

\textit{\underline{Outline of Proof}:}

The upper bound follows by showing that for fixed $\spa$ and $\dpr$, there is always a two-level sequence in $\spas(\nu)$ attaining the same objective value. Intuitively, this can be regarded as replacing the minimization over $\dpr$ to a maximization and combining it with the supremum to form a restricted set.

The lower bound follows by directly restricting the set of all $\spa$ to a set consisting of all two-level sequences.

The proof of the theorem above can be found in Appendix~\ref{app:1}.
\hfill $\blacksquare$

Using both the restricted signal and noise power sets, we obtain another upper bound on $C_\cl$.	
	
\begin{theorem}[\textit{Upper Bound}]
	\label{thm:5}
For any block-length $\cl$, $C_\cl$ can be bounded from above as
\begin{align*}
	C_\cl\leq &\widetilde{C}_\cl\triangleq \adjustlimits\min_{1\leq \dpr\leq \cl}\sup_{\spa\in\spas(\dpr)}\inf_{\npa\in\npas(\dpr)}
	\\&
	\begin{cases}
	\frac{\dpr}{2\cl}\log\left({\underline{\spal}}/{\underline{\npal}}\right)
	\quad &\text{ if } \overline{P}\leq 2\overline{N}\\
	\frac{\dpr}{2\cl}\log\left({\underline{\spal}}/{\underline{\npal}}\right)
	+\frac{\cl-\dpr}{2\cl}\log\left({\overline{\spal}}/{\overline{\npal}}\right)
	\quad &\text{ otherwise }
	\end{cases}\\
	&\text{ subject to } \\
	& \quad \underline{\spal}\geq \underline{\npal}.
\end{align*}
\end{theorem}

\textit{\underline{Outline of Proof}:}

The bound above can be derived by bringing the minimization in optimization (\ref{eq:3.2.7}) to the front. The proof is provided in Appendix~\ref{app:6}.
\hfill $\blacksquare$

\subsection{Experimental Results}
We provide numerical calculations of the bounds  $\underline{C}_\cl$ and $\overline{C}_\cl$ for $\cl=50,100,150$ and $200$.  The quantization level used for the signal to noise ratio is fixed to be $0.005$. The numerical values of $\underline{C}_\cl$ and $\overline{C}_\cl$ sampled at ${\npal}/{\spal}=0.1,0.2,0.3$ and $0.4$ are summarized in the tables below. By comparing the numerical values of them with different $\mathrm{SNR}$, we find that they converge fast as the block-length $\cl$ increases and the lower bound $\underline{C}_\cl$ and upper bound $\widetilde{C}_\cl$ are close. Thus by setting $\cl=500$, the curves plotted in Figure~\ref{Fig:1} give an acceptable characterization of the capacity region.

\begin{table}[H]
\begin{center}
	\pgfplotstableread{ 
		0.100   1.6610   1.6610  1.6610   1.6610
		0.200   1.1603   1.1607  1.1608   1.1608
		0.300   0.8529   0.8531  0.8532   0.8531
		0.400   0.5700   0.5705  0.5707   0.5707
	}\datatable
	
	\pgfplotstabletypeset[
	columns/0/.style={column name={$\npal/\spal$}}, 
	columns/1/.style={
		column name={$n=50$},  
		dec sep align,      
		/pgf/number format/fixed zerofill,  
		/pgf/number format/precision=4     
	},
	columns/2/.style={
		column name={$n=100$},
		dec sep align,
		/pgf/number format/fixed zerofill,
		/pgf/number format/precision=4
	},
	columns/3/.style={
		column name={$n=150$},
		dec sep align,
		/pgf/number format/fixed zerofill,
		/pgf/number format/precision=4
	},
	columns/4/.style={
		column name={$n=200$},
		dec sep align,
		/pgf/number format/fixed zerofill,
		/pgf/number format/precision=4
	},
	every head row/.style={
		before row=\toprule,    
		after row=\midrule
	},
	every last row/.style={
		after row=\bottomrule
	}]{\datatable}
	\end{center}
		\caption{Numerical Values of $\underline{C}_{\cl}$.}
	\end{table}
	
	\begin{table}[H]
		\begin{center}
			\pgfplotstableread{ 
				0.100   1.6610   1.6610  1.6610   1.6610
				0.200   1.1610   1.1610  1.1610   1.1610
				0.300   0.8685   0.8685  0.8685   0.8685
				0.400   0.6610   0.6610  0.6610   0.6610
			}\datatable
			
			\pgfplotstabletypeset[
			columns/0/.style={column name={$\npal/\spal$}}, 
			columns/1/.style={
				column name={$n=50$},  
				dec sep align,      
				/pgf/number format/fixed zerofill,  
				/pgf/number format/precision=4     
			},
			columns/2/.style={
				column name={$n=100$},
				dec sep align,
				/pgf/number format/fixed zerofill,
				/pgf/number format/precision=4
			},
			columns/3/.style={
				column name={$n=150$},
				dec sep align,
				/pgf/number format/fixed zerofill,
				/pgf/number format/precision=4
			},
			columns/4/.style={
				column name={$n=200$},
				dec sep align,
				/pgf/number format/fixed zerofill,
				/pgf/number format/precision=4
			},
			every head row/.style={
				before row=\toprule,    
				after row=\midrule
			},
			every last row/.style={
				after row=\bottomrule
			}]{\datatable}
		\end{center}
		\caption{Numerical Values of $\overline{C}_{\cl}$.}
	\end{table}

	\begin{table}[H]
	\begin{center}
		\pgfplotstableread{ 
			0.100   1.6610   1.6610  1.6610   1.6610
			0.200   1.1610   1.1610  1.1607   1.1610
			0.300   0.8591   0.8607  0.8591   0.8594
			0.400   0.6026   0.6024  0.6021   0.6022
		}\datatable
		
		\pgfplotstabletypeset[
		columns/0/.style={column name={$\npal/\spal$}}, 
		columns/1/.style={
			column name={$n=50$},  
			dec sep align,      
			/pgf/number format/fixed zerofill,  
			/pgf/number format/precision=4     
		},
		columns/2/.style={
			column name={$n=100$},
			dec sep align,
			/pgf/number format/fixed zerofill,
			/pgf/number format/precision=4
		},
		columns/3/.style={
			column name={$n=150$},
			dec sep align,
			/pgf/number format/fixed zerofill,
			/pgf/number format/precision=4
		},
		columns/4/.style={
			column name={$n=200$},
			dec sep align,
			/pgf/number format/fixed zerofill,
			/pgf/number format/precision=4
		},
		every head row/.style={
			before row=\toprule,    
			after row=\midrule
		},
		every last row/.style={
			after row=\bottomrule
		}]{\datatable}
	\end{center}
	\caption{Numerical Values of $\widetilde{C}_{\cl}$.}
\end{table}

\section{Robustness of the Optimization}
\label{sec:robust}
Due to technical necessities, we employ slightly different optimizations for proving converse and achievability respectively by introducing slacknesses for both of them. Next we present the first optimization used for Theorem~\ref{thm:1} in Section~\ref{sec:4}. 

\subsection{Optimization~(\ref{eq:3.2.5}) for Converse}
\label{sec:3.1.1}
Let $\tau>0$ be an arbitrary constant.
With the sets $\spas$ and $\npas$ defined above, we state the following optimization problems optimizing over all sequences $\spa$ and $\npa$ in the two sets $\spas$ and $\npas$ respectively. The feasible set is compact and therefore an optimal solution $\left(\spa^*,\npa^*\right)$ exists and an optimal value can be attained.
\begin{mybox}{Tweaked Optimization with Optimal Value $\overline{C}_{\cl}^{\tau}$}
	\begin{equation}
	\label{eq:3.2.5}
	\begin{aligned}
	& &\underset{\spa}{\text{sup}} \ \underset{\npa}{\text{inf}}  \ \underset{1\leq \dpr\leq\cl}{\text{min}} \qquad \frac{1}{2\cl}\sum_{\indt=1}^{\dpr}&\log\frac{\spal_{\indt}}{\npal_{\indt}}\\
	& &\text{subject to}  \qquad \qquad  \qquad \spa&\in\spas,\\
	& &  \npa &\in\npas,\\
	& &\npal_{\indt} &\leq \spal_{\indt}\\
	& & \text{for all } \indt&=1,\ldots,\dpr,\\
	& &   \cl \spal-\sum\limits_{\indt=1}^{\dpr}\spal_{\indt} \leq \left(1-\tau\right)\Big(2\cl\npal&-\sum\limits_{\indt=1}^{\dpr}2\npal_{\indt}\Big).
	\end{aligned}
	\tag{P2}
	\end{equation}
\end{mybox}
Denote by $\overline{C}_{\cl}^{\tau}$ the optimal value of the optimization problem~(\ref{eq:3.2.5}). Later in Section~\ref{sec:4} we shall show that for any $(\mn,\cl)$-code, a rate greater than $\overline{C}_{\cl}^{\tau}$ can never be achievable.

\subsection{Optimization~(\ref{eq:3.2.6}) for Achievability}
\label{sec:3.1.2}
Due to technical issues, we also need a second optimization that is slightly different from the previous one. 
Let $\theta>0$ be an integer denoting the \textit{chunk size} (specified later in Eq.~(\ref{eq:5.20})) and without loss of generality, suppose the number of chunks $\cul={\cl}/{\theta}$ is an integer.
Additionally, in the following, we define two sets of real-valued positive length-$\cul$ sequences.
\begin{definition}[\textit{Chunked Signal Power Set}]
	A \textit{chunked signal power set} $\spasd$ denotes a set containing all length-$\cul$ non-negative real sequences $\spad=\spadl_1,\ldots,\spadl_{\cul}$ satisfying
	\begin{align*}
	&\spadl_{\indT}>0, \quad \indT=1,\ldots,\cul,\\
	&\sum_{\indT=1}^{\cul}\spadl_{\indT}\leq\cl \spal.
	\end{align*}
\end{definition}

\begin{definition}[\textit{Chunked Noise Power Set}]
	A \textit{chunked noise power set} $\npasd$ denotes a set containing all length-$\cul$ non-negative real sequences $\npad=\npadl_1,\ldots,\npadl_{\cul}$ satisfying
	\begin{align*}
	&\npadl_{\indT}>0, \quad \indT=1,\ldots,\cul,\\
	&\sum_{\indT=1}^{\cul}\npadl_{\indT}\leq\cl \npal.
	\end{align*}
\end{definition}

Let $\gamma>0$ be a constant. 
Optimizing over all sequences $\spad$ and $\npad$ in the two sets $\spasd$ and $\npasd$ respectively, we get a similar optimization~(\ref{eq:3.2.6}) as we have in~(\ref{eq:3.2.5}):

\begin{mybox}{Chunked Optimization with Optimal Value $\underline{C}_{\cul}^{\gamma}$}
	\begin{equation}
	\label{eq:3.2.6}
	\begin{aligned}
	& &\underset{\spad}{\text{sup}} \ \underset{\npad}{\text{inf}}  \ \underset{1\leq \dpc\leq\cul}{\text{min}} \qquad \frac{1}{2\cul}\sum_{\indT=1}^{\dpc}&\log\frac{\spadl_{\indT}}{\npadl_{\indT}}\\
	& &\text{subject to}  \qquad \qquad  \qquad \spad&\in\spasd,\\
	& &  \npad &\in\npasd,\\
	& &\npadl_{\indT} &\leq \spadl_{\indT}\\
	& & \text{for all } \indT&=1,\ldots,\dpc,\\
	& &  (1-\gamma)\cl \spal  -\sum\limits_{\indT=1}^{\dpc}\spadl_{\indT}\leq 2\cl\npal&-\sum\limits_{\indT=1}^{\dpc}2\npadl_{\indT}.
	\end{aligned}
	\tag{P3}
	\end{equation}
\end{mybox}

Denote by $\underline{C}_{\cul}^{\gamma}$ the optimal value of the optimization problem~(\ref{eq:3.2.6}). The optimal value $\underline{C}_{\cul}^{\gamma}$ exists since the feasible set of~(\ref{eq:3.2.6}) is non-empty. 

The optimization~(\ref{eq:3.2.6}) above is basically a chunked version of~(\ref{eq:3.2.5}). It is useful since to prove the achievability, it is convenient for us to consider chunk-wise encoding and decoding.
Later in Section~\ref{sec:5}, we shall show that any rate less than $\underline{C}_{\cul}^{\gamma}$ is achievable. To prove this, we use the same construction of codes in~\cite{chen2015characterization}. First, an encode transmits a concatenation of $\cul$ chunks of $\theta$-length codewords. Then a decoder estimates iteratively based on the received codeword. 

\subsection{Equivalent Forms}
Equivalently, we can develop alternative expressions of the optimal values $\overline{C}_{\cl}^{\tau}$ of optimization (\ref{eq:3.2.5}) and $\underline{C}_{\cul}^{\gamma}$ of the chunked optimization (\ref{eq:3.2.6}).

We consider a  fixed \textit{division parameter} $\dpr$ in $\left\{1,\ldots,\cl\right\}$. Then for a given $\spa\in\spas$, it is helpful to define a new set of feasible power allocation sequences for James.
Let $\npase\subseteq\npas$ be the set containing all $\npa$ satisfying the following constraints:
\begin{align}
\label{eq:3.25}
&\npal_{\indt}\leq \spal_{\indt}, \quad \indt=1,\ldots,\dpr,\\
\label{eq:3.26}
&\cl \spal-\sum\limits_{\indt=1}^{\dpr}\spal_{\indt} \leq \left(1-\tau\right)\Big(2\cl\npal-\sum\limits_{\indt=1}^{\dpr}2\npal_{\indt}\Big),\\
\label{eq:3.28}
&\cl \spal-\sum\limits_{\indt=1}^{\indt_{0}}\spal_{\indt} > \left(1-\tau\right)\Big(2\cl\npal-\sum\limits_{\indt=1}^{\indt_0}2\npal_{\indt}\Big)\\
\nonumber
&\qquad\text{for all } \indt_{0}<\dpr.
\end{align}

The first two set of inequalities (\ref{eq:3.25}) and (\ref{eq:3.26}) come from the constraints in optimization~(\ref{eq:3.2.5}). The last inequality (\ref{eq:3.28}) guarantees that the fixed division parameter $\dpr$ is the optimizer since every $\indt_0<\dpr$ violates the energy-bounding condition.

Similarly, we can do the same for optimization~(\ref{eq:3.2.6}). Fix a division parameter  $\dpc$ in $\left\{1,\ldots,\cul\right\}$.  For a given $\spad\in\spasd$, let $\npasde\subseteq\npasd$ be the set containing all $\npad$ satisfying the following:
\begin{align}
\label{eq:3.15}
&\npadl_{\indT}\leq \spadl_{\indT}, \quad \indT=1,\ldots,\dpc,\\
&\left(1-\gamma\right)\cl \spal-\sum\limits_{\indT=1}^{\dpc}\spadl_{\indT} \leq 2\cl\npal-\sum\limits_{\indT=1}^{\dpc}2\npadl_{\indT},\\
\label{eq:3.16}
&\left(1-\gamma\right)\cl \spal-\sum\limits_{\indT=1}^{\indT_0}\spadl_{\indT}> 2\cl\npal-\sum\limits_{\indT=1}^{\indT_0}2\npadl_{\indT}\\
\nonumber
&\qquad \text{for all } \indT_{0}<\dpc.
\end{align}

Note that for a certain $\dpr$, the set $\npase$ above may be empty. In that case, there is no feasible solution and the objective value is set to be positive infinite. This will not change the optimal value and solutions since it can be verified that for any $\spa$, at least for some $\dpr$, the set $\npase$ is non-empty\footnote{For instance, $\dpr=\cl$ always guarantees the energy bounding condition (the last constraint) in ~\ref{eq:3.2.5}. Therefore, the feasible set is non-empty and an \textit{optimal} value $\overline{C}_{\cl}^{\tau}$ exists. This also holds for $\npasde$ defined in (\ref{eq:3.15})-(\ref{eq:3.16}).}. Minimizing over all possible $1\leq\dpr\leq\cl$ and maximizing over all power allocation $\spa$ for Alice, we know the optimal value $\overline{C}_{\cl}^{\tau}$  of the optimization problem (\ref{eq:3.2.5}) can be written as the following one-line form:
\begin{lemma}
	For any block-length $\cl$, the optimal value $\overline{C}_{\cl}^{\tau}$  of the optimization problem (\ref{eq:3.2.5}) equals to
	\begin{align}
	\label{eq:3.24}
	\overline{C}_{\cl}^{\tau}=\adjustlimits
	\sup_{\spa\in\spas}\min_{1\leq\dpr\leq\cl}\inf_{\npa\in\npase}
	\frac{1}{2\cl}\sum_{\indt=1}^{\dpr}\log\frac{\spal_{\indt}}{\npal_{\indt}}.
	\end{align}
\end{lemma}

Similarly, we have the following lemma:
\begin{lemma}
		For any block-length $\cl$, the optimal value $\underline{C}_{\cul}^{\gamma}$  of the optimization problem (\ref{eq:3.2.6}) equals to
	\begin{align}
	\label{eq:3.14}
	\underline{C}_{\cul}^{\gamma}=\adjustlimits
	\sup_{\spad\in\spasd}\min_{1\leq\dpc\leq\cul}\inf_{\npad\in\npasde}
	\frac{1}{2\cul}\sum_{\indT=1}^{\dpc}\log\frac{\spadl_{\indT}}{\npadl_{\indT}}.
	\end{align}
\end{lemma}

For the two optimal values, we can get rid of slacknesses $\tau>0$ and $\gamma>0$ by studying their asymptotic behaviors. later, we show that the two optimizations~(\ref{eq:3.2.5}) and~(\ref{eq:3.2.5}) are indeed \textit{robust} such that once the slackness $\tau>0$ and $\gamma>0$ are small enough, then the corresponding optimal values do not differ significantly from the reference optimization~\ref{eq:3.2.7}.

%
%
	\subsection{Robustness of the optimization}
	
	The two optimal values $\overline{C}_{\cl}^{\tau}$ and $\underline{C}_{\cul}^{\gamma}$ correspond to optimization problems tweaked slightly by slacknesses $\tau>0$ and $\gamma>0$. 
	We note that the positive slacknesses  $\tau$ and $\gamma$ can be arbitrarily small.  As the first step of characterizing the channel capacity $C$, our first theorem states that optimization~(\ref{eq:3.2.5}) and~(\ref{eq:3.2.6}) are robust such that if $\tau>0$ and $\gamma>0$ are small enough, then the corresponding optimal values do not differ a lot from the reference optimization~(\ref{eq:3.2.7}) without slackness.

	Recall that $C_\cl$ denotes the corresponding optimal value of optimization~(\ref{eq:3.2.7}).  We have the following theorem providing the desired robustness of optimizations~(\ref{eq:3.2.5}) and~(\ref{eq:3.2.6}). The proof is presented in Appendix~\ref{app:1}.
	
	\begin{theorem}[\textit{Robustness}]
		\label{thm:3}
		Let $C_{\cl}$ and $C_{\cul}$ be the corresponding optimal values of optimization~(\ref{eq:3.2.7}) given block-lengths $\cl$ and $\cul$.
		For any constants $A>0$ and $B>0$, there exist $\tau>0$ and $\gamma>0$ such that
		\begin{align*}
		C_\cl\geq& \overline{C}^{\tau}_\cl+\frac{1}{2}\log\left(1-A\right)\\
		C_\cul\leq&  \underline{C}^{\gamma}_\cul-\frac{1}{2}\log\left(1-B\right)
		\end{align*}
		when $\cl$ and $\cul$ are sufficiently large.
	\end{theorem}
	
	In the coming two sections, Theorem~\ref{thm:3} above will be used in the end of the proofs of Theorem~\ref{thm:1} and Theorem~\ref{thm:2} to help clean up the statements. Now, we are ready to move to details related to our converse and achievability. The next section specifies an attack strategy for James followed by a sketched derivation of the converse result in Theorem~\ref{thm:1}.
	

\section{Converse}
\label{sec:4}
Fix a positive constant $\varepsilon>0$ arbitrarily. 
It suffices to find an attack strategy by specifying some \textit{causal} distribution $p_{\St|\Cwd}\in\mathsf{P}$ for an adversary such that the average probability of error $\apoe$ defined in (\ref{eq:1.2}) is always a positive constant $\varepsilon^{O\left({1}/{\varepsilon}\right)}$ for any $\left(\mn,\cl\right)$-code with rate $R=\frac{1}{\cl}\log\mn=\overline{C}_{\cl}^{\tau}+3\varepsilon$ and block-length $\cl$ large enough. 

Fix any integer $\cl\geq 1$ large enough. 
In what follows, we show that for any $\left(\mn,\cl\right)$-code $\left(p_{\Cwd|\Msg},p_{\Esg|\Rwd}\right)$, if its rate $R$ satisfies
\begin{align}
 R&=\adjustlimits
 \sup_{\spa\in\spas}\inf_{1\leq\dpr\leq\cl}\inf_{\npa\in\npase}
 \frac{1}{2\cl}\sum_{T=1}^{\dpr}\log\frac{\spal_{\indt}}{\npal_{\indt}}+3\varepsilon\\
 &= \overline{C}_{\cl}^{\tau}+3\varepsilon,
\end{align}
then the average probability of error $\apoe$ is always a positive constant. Therefore Theorem~\ref{thm:1} in Section~\ref{thm:1} follows.

Below we specify a causal attack strategy, called the \textit{scaled babble-and-push attack}. The attack is motivated by the \textit{babble-and-push attack} in~\cite{dey2013upper} for causal binary bit-flipping channels. 
\subsection{Scaled Babble-and-Push Attack}

Given an $\left(\mn,\cl\right)$-code $\left(p_{\Cwd|\Msg},p_{\Esg|\Rwd}\right)$, we recall in Definition~\ref{def:7} its corresponding \textit{average power allocation sequence} $\spa=\spal_{1},\ldots,\spal_{\cl}\in\spas$ where $\spal_{\indt}=\expc[\left|\Clwd_{\indt}\right|^2]$.
Provided with a fixed average power allocation sequence $\spa$, we are ready to give the scaled babble-and-push attack.  Let $1\leq\optdpr\leq\cl$ and $\npa^*=\npal^*_1,\ldots,\npal^*_{\cl}\in\optnpase$ be the \textit{optimal} solutions of the optimization below:
\begin{align}
\label{eq:3.23}
\min_{1\leq\dpr\leq\cl}\inf_{\npa\in\npase}
\frac{1}{2\cl}\sum_{\indt=1}^{\dpr}\log\frac{\spal_{\indt}}{\npal_{\indt}}.
\end{align}

Note that such $\optdpr$ and $\npa^*$ exist since for any fixed $1\leq\optdpr\leq\cl$, both constraints and objective function of the optimization (\ref{eq:3.23}) above are convex.

The two-stage attack strategy can be summarized as follows.


\begin{mybox}{Scaled Babble-and-Push Attack}
For each $\indt$-th ($\indt=1,\ldots,\cl$) transmission, the causal adversarial noise $\Slt_\indt$ is given by
	\begin{align}
	\label{eq:3.0}
	\Slt_\indt=\begin{cases}
	\Zlt_\indt-\frac{\npal^*_{\indt}}{\spal_{\indt}}\Clwd_\indt, \qquad\qquad & \ \indt\leq \optdpr\\
\frac{1}{2}{\left(\Hlwd_\indt-\Clwd_\indt\right)}, \qquad\qquad &\ \indt> \optdpr
	\end{cases}.
	\end{align}
\end{mybox}
	Denote by $\optCwdp$ and $\optRwdp$ ($\optCwdl$ and $\optRwdl$)  the $\optdpr$-prefix ($\optdpr$-suffix) of the codewords $\Cwd$ and $\Rwd$. We describe the attack in two stages.
	In the first stage when $\indt\leq\optdpr$, each $\Zlt_\indt$ is an independent Gaussian random variable with zero mean and variance 
	\begin{align*}
	 \widetilde{{\npal}}_{\indt}^{\varepsilon}=\frac{1}{1+\varepsilon}\npal^*_{\indt}\left(1-\frac{\npal^*_{\indt}}{\spal_{\indt}}\right).
	\end{align*}
 In the second stage (the last case) when $T>\optdpr$, given a fixed prefix $\optrwdp$, James first generates a random message $\Usg$ according to the distribution $p_{\Msg|\optRwdp}$ such that
	\begin{align}
	\label{eq:4.34}
	\Pr_{\Usg}\left(\Usg=\usg\right)=\Pr_{\Msg|\optRwdp}\left(\Msg=\usg|\optRwdp=\optrwdp\right).
	\end{align}
Let $\optHwdl\left(\usg\right)$ be a randomly selected corresponding suffix of codeword according to $p_{\optHwdl|\Usg}$ given a message $\Usg=\usg$. Then James pushes the suffix $\optCwdl$ towards the middle point between $\optCwdl$ and $\optHwdl$.




\begin{remark}
Note that based on the attack construction above, for some case, a realization of the adversarial noise  $\st$ may be out of the ball $\ballN$. If such case occurs, James will simply discard the state $\st$ and the attack is unsuccessful. 
\end{remark}


We verify the two-stage attack in (\ref{eq:3.0}) indeed satisfies the causality property in Definition~\ref{def:0}, as the following theorem states.
\begin{theorem}
The distribution $p_{\St|\Cwd}$ is causal.
\end{theorem}

\begin{proof}
	We verify the claim by considering the first stage when $\indt\leq\optdpr$ and the second stage when $\indt>\optdpr$ respectively.	
\begin{enumerate}
\item When $\indt\leq \optdpr$, each $\Slt_\indt$ only depends on $\Clwd_\indt$ for all $\indt\leq \optdpr$ since $\Zlt_\indt$ is an independent Gaussian random variable.
Therefore $p_{\optStp|\optCwdp}$ is a causal distribution.

\item When $\indt>\optdpr$, the prefix $\optRwdp$ of length-$\optdpr$ is already fixed. Conditioned on $\optRwdp=\optrwdp$, James first selects a random message $\Usg$ according to $p_{\Msg|\optRwdp}$. Then, James simulates a random codeword $\optHwdl$ for $\Usg$. Therefore, each $\Slt_\indt$ only depends on $\optRwdp$ (hence $\optCwdp$) and $\Clwd_\indt$. We conclude that $p_{\St|\Cwd}$ is a causal distribution.
\end{enumerate}
\end{proof}

The probabilistic structure for the two-stage attack can be visualized as the following diagram in  Figure~\ref{fig:2}.  The dotted arrows from $\optCwdp$ to $\optStp$ and $\optCwdl$ to $\optStl$ denote the causal dependency between them.

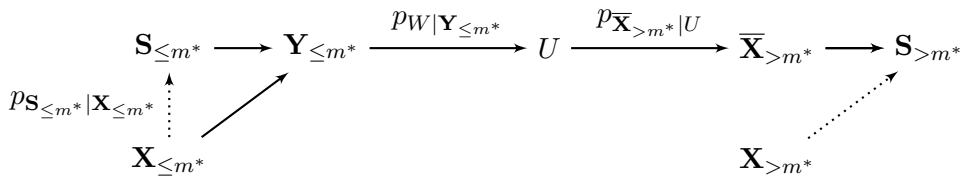
\begin{figure*}[h]
	 	 \hrule
	\begin{center}
		\begin{tikzpicture}[node distance=2cm,auto,>=latex']
		\node (begin)  {$\optStp$};
		\node (a) [right of=begin, node distance=2cm]{$\optRwdp$};
		\node (b) [right of=a, node distance=3cm] {$\Usg$};
		\node (f) [right of=b, node distance=3cm] {$\optHwdl$};
		\node (d) [right of=f, node distance=2cm] {$\optStl$};
		\node (e) [below of=f, node distance =1.5 cm] {$\optCwdl$};
		\node (c) [below of=begin, node distance =1.5 cm] {$\optCwdp$};
		\draw[thick,->] (a) edge node [name=p] {$p_{\Msg|\optRwdp}$} (b);
		\node [coordinate] (end) [below of=f, node distance=1cm]{};
		\draw[thick,->] (b) edge node [name=p] {$p_{\optHwdl|\Usg}$}  (f);
		\draw[thick,->] (f) edge node [below,pos=0.55] {}  (d);
		\draw[thick,dotted,->] (e) edge node {} (d);
		\begin{scope}[every node/.style={scale=.65}]
		\draw[thick,->] (begin) edge node [name=p] {} (a);
		\draw[thick,->] (f) edge node [name=p] {} (d);
		\end{scope}
		\draw[thick,->] (c) edge node {} (a);
		\draw[thick,dotted,->] (c) edge node {$p_{\optStp|\optCwdp}$} (begin);
		\end{tikzpicture}
	\end{center}
	\caption{Causal Distribution of $\St$ given $\Cwd$.}	
	\label{fig:2}
	 	 \hrule
\end{figure*}

\subsection{Proof Sketch}
We give some intuition first.

At the time-step $\indt=\optdpr+1$, the prefix $\optRwdp=\rwdp$ is fixed. Conditioned on $\optRwdp=\optrwdp$, James first selects a random message $\Usg$ according to $p_{\Msg|\optRwdp}$. Then pretending that $\Usg$ is the transmitted message, James simulates a random codeword $\optHwdl$ as a copy of $\optCwdl$ and they have the same distribution conditioned on $\optRwdp$ .
Therefore, as we formally state in Lemma~\ref{lemma:0} below, if $\Msg\neq U$ and at the same time the adversarial noise $\St$ is in $\ballN$, by pushing the  $\optdpr$-suffix $\optCwdl$ towards the middle point between $\optCwdl$ and $\optHwdl$, Bob will be confused and unable to distinguish the selected message from $\Msg$ and $\Usg$. Intuitively, since from the estimate $\esg$'s point of view, the truly selected message can be either $\Msg$ or $\Usg$ with equal probability but they are distinct. 

We summarize above as a lower bound on the average probability of error $\apoe$:
 \begin{lemma}
 	\label{lemma:0}
 	With $\St$ defined in (\ref{eq:3.0}) and $\Usg$ defined in (\ref{eq:4.34}), 
 	 	$$\apoe\geq\frac{1}{2}\Pr\left(\left|\left|\St\right|\right|^2\leq \cl \npal, \ \Msg\neq\Usg\right)$$ where the randomness is from the joint distributions of  $\St$, $\Msg$ and $\Usg$.
 \end{lemma}
  \begin{proof}
  	Recall $\mathsf{Q}$ defines the set of all probability density functions $p_{\Rwd|\Cwd}$ with an underlying causal distribution $p_{\St|\Cwd}$.
  	By Definition~\ref{def:1},
  	\begin{align*}
  	\apoe
  	&=\sum_{\msg=1}^{\mn}\sup_{p_{\Rwd|\Cwd}\in\mathsf{Q}}\sum_{\esg\neq\msg}\int_{\rwd\in\ballPN}p_{\Rwd,\Msg}\left(\rwd ,\msg\right)p_{\Esg|\Rwd}\left(\esg|\rwd\right)\mathrm{d}\rwd\\
  	&\geq
  	\sup_{p_{\Rwd|\Cwd}\in\mathsf{Q}}\sum_{\msg=1}^{\mn}\sum_{\esg\neq\msg}\int_{\rwd\in\ballPN}p_{\Rwd,\Msg}\left(\rwd ,\msg\right)p_{\Esg|\Rwd}\left(\esg|\rwd\right)\mathrm{d}\rwd.
  	\end{align*}
  	
  	With the causal distribution $p_{\Rwd|\Cwd}\in\mathsf{Q}$ elaborated in Figure~\ref{fig:2}, substituting the symbol $\msg$ by $\usg$ in $\apoe$ and taking the optimal decoder $p_{\Esg|\Rwd}$,
  	\begin{align}
\nonumber
  &\qquad 2\apoe\geq\\
  	  	\label{eq:2.1} 
  	&\inf_{p_{\Esg|\Rwd}}\Big(\sum_{\msg=1}^{\mn}\sum_{\esg\neq\msg}\int_{\rwd\in\ballPN}p_{\Rwd,\Msg}\left(\rwd,\msg\right)p_{\Esg|\Rwd}\left(\esg|\rwd\right)\mathrm{d}\rwd
  	\\
  	&+\sum_{\usg=1}^{\mn}\sum_{\esg\neq\usg}\int_{\rwd\in\ballPN}p_{\Rwd,\Msg}\left(\rwd,\usg\right)p_{\Esg|\Rwd}\left(\esg|\rwd\right)\mathrm{d}\rwd\Big).
  	\end{align}
  	
  	Following the attack described in (\ref{eq:3.0}), $\Rlwd_\indt=\frac{1}{2}\left(\Hlwd_\indt+\Clwd_\indt\right)$ whenever $\indt>\optdpr$. Hence the positions of $\optCwdl$ and $\optHwdl$ in the distribution $p_{\optRwdl,\optCwdl,\optHwdl}$ are exchangeable. Therefore the two messages $\Msg$ and $\Usg$ are also replaceable in $p_{\optRwdl,\Msg,\Usg}$ and we have
  	\begin{align}
  	  	\label{eq:4.27}
  	&p_{\optRwdl|\optRwdp,\Msg,\Usg}\left(\optrwdl|\optrwdp,\msg,\usg\right)\\
  	=	&p_{\optRwdl|\optRwdp,\Msg,\Usg}\left(\optrwdl|\optrwdp,\usg,\msg\right).
  	\end{align}
  	 Moreover, by the definition of $\Usg$, conditioned on $\optRwdp$, $\Msg$ and $\Usg$ have the same distribution. Hence
  	 \begin{align}
  	  	\label{eq:4.28}
  	 p_{\optRwdp,\Msg,\Usg}\left(\optrwdp,\msg,\usg\right)=	p_{\optRwdp,\Msg,\Usg}\left(\optrwdp,\usg,\msg\right).
  	 \end{align}
  	  Therefore combining (\ref{eq:4.27}) and (\ref{eq:4.28}), we have for all $\optrwdp,\optrwdl,\msg$ and $\usg$
  	\begin{align*}
  	&p_{\optRwdp,\optRwdl,\Msg,\Usg}\left(\optrwdp,\optrwdl,\msg,\usg\right)\\=&p_{\optRwdp,\optRwdl,\Msg,\Usg}\left(\optrwdp,\optrwdl,\usg,\msg\right).
  	\end{align*}

  	The joint distribution $p_{\Rwd,\Msg}$ can be decomposed as follows:
  	\begin{align}
  	\label{eq:2.1.a}
  	p_{\Rwd,\Msg}\left(\rwd,\msg\right)&=\sum_{\usg=1}^{\mn}p_{\optRwdp,\optRwdl,\Msg,\Usg}\left(\optrwdp,\optrwdl,\msg,\usg\right)\\
  	&=\sum_{\usg=1}^{\mn}p_{\Rwd,\Msg,\Usg}\left(\rwd,\msg,\usg\right)
  	\end{align}
  	and
  	\begin{align}
  	\nonumber
  	p_{\Rwd,\Msg}\left(\rwd,\usg\right)&=\sum_{\msg=1}^{\mn}p_{\optRwdp,\optRwdl,\Msg,\Usg}\left(\optrwdp,\optrwdl,\usg,\msg\right)\\
  	\label{eq:2.0}
  	&=\sum_{\msg=1}^{\mn}p_{\optRwdp,\optRwdl,\Msg,\Usg}\left(\optrwdp,\optrwdl,\msg,\usg\right)\\
  	\nonumber
  	&=\sum_{\msg=1}^{\mn}p_{\Rwd,\Msg,\Usg}\left(\rwd,\msg,\usg\right).
  	\end{align}
  	The equality (\ref{eq:2.0}) above comes from (\ref{eq:2.1.a}).
  	
  	Putting the expressions of $p_{\Rwd,\Msg}\left(\rwd,\msg\right)$ in (\ref{eq:2.1.a}) and $p_{\Rwd,\Msg}\left(\rwd,\usg\right)$ above into (\ref{eq:2.1}), we obtain
  	\begin{align*}
  	&2\apoe\geq\inf_{p_{\Esg|\Rwd}}\sum_{\msg=1}^{\mn}\sum_{\usg=1}^{\mn}\int_{\rwd\in\ballPN}\sum_{\esg=1}^{\mn}\left(\mathds{1}\left(\esg\neq \msg\right)+\mathds{1}\left(\esg\neq \usg\right)\right)\\
  	&\qquad\qquad p_{\Esg|\Rwd}\left(\esg|\rwd\right)
  	p_{\Rwd,\Msg,\Usg}\left(\rwd,\msg,\usg\right)\mathrm{d}\rwd.
  	\end{align*}
  	
  	Moreover, provided $\msg\neq\usg$, we have 
  	\begin{align*}
  	&\sum_{\esg=1}^{\mn}\left(\mathds{1}\left(\esg\neq \msg\right)+\mathds{1}\left(\esg\neq \usg\right)\right)p_{\Esg|\Rwd}\left(\esg|\rwd\right)\\
  	\geq&\min_{\esg\in\msgs}\left(\mathds{1}\left(\esg\neq \msg\right)+\mathds{1}\left(\esg\neq \usg\right)\right)\geq 1 \quad \text{for all} \ \rwd \ \text{and} \ \esg.  
  	\end{align*}
  	Therefore above yields
  	\begin{align*}
  	2\apoe\geq&\sum_{\msg=1}^{\mn}\sum_{\msg\neq\usg}\int_{\rwd\in\ballPN}p_{\Rwd,\Msg,\Usg}\left(\rwd,\msg,\usg\right)\mathrm{d}\rwd\\
  	=&\sum_{\msg=1}^{\mn}\sum_{\msg\neq\usg}\int_{\st\in\ballN}\int_{\rwd\in\ballPN}\int_{\cwd\in\ballP}\mathds{1}\left(\rwd=\cwd+\st\right)\\
  	&\qquad\qquad p_{\Cwd,\St,\Msg,\Usg}\left(\cwd,\st,\msg,\usg\right)\mathrm{d}\cwd\mathrm{d}\rwd\mathrm{d}\st.
  	\end{align*}
  	
  	Concerning the fact that each $\rwd\in\ballPN$ is a sum of some $\cwd\in\ballP$ and $\st\in\ballN$ and
  	\begin{align*}
  	p_{\St,\Msg,\Usg}\left(\st,\msg,\usg\right)&=\int_{\cwd\in\ballP}p_{\Cwd,\St,\Msg,\Usg}\left(\cwd,\st,\msg,\usg\right)\mathrm{d}\cwd\\
  	&=\int_{\rwd\in\ballPN}\int_{\cwd\in\ballP}\mathds{1}\left(\rwd=\cwd+\st\right)\\
  	&\qquad \qquad p_{\Cwd,\St,\Msg,\Usg}\left(\cwd,\st,\msg,\usg\right)\mathrm{d}\cwd\mathrm{d}\rwd,
  	\end{align*}
  	we write
  	\begin{align}
  	\label{eq:3.4}
  	\apoe\geq&\frac{1}{2}\sum_{\msg=1}^{\mn}\sum_{\msg\neq\usg}\int_{\st\in\ballN}p_{\St,\Msg,\Usg}\left(\st,\msg,\usg\right)\mathrm{d}\st\\
  	&\triangleq\frac{1}{2}\Pr_{\Msg,\Usg,\St}\left(\Msg\neq\Usg, \ \left|\left|\St\right|\right|^2\leq \cl \npal \right).
  	\end{align}
  \end{proof}

Let $\underline{\mathcal{Y}}$ denote the set containing all length-$\optdpr$ prefixes $\optrwdp$.

Next we analyze the probability above by decomposing it into three parts.

Fix a $\tau>0$ arbitrarily. The probability in (\ref{eq:3.4}) can be further bounded from below as
\begin{align*}
&\Pr\left( \Msg\neq\Usg , \ \left|\left|\St\right|\right|^2\leq \cl \npal\right)\\
\geq&\Pr\left( \Msg\neq\Usg , \ \left|\left|\St\right|\right|^2\leq \cl \npal, \ \left|\left|\optStl\right|\right|^2\leq \sum_{\indt=\optdpr+1}^{\cl}\npal_\indt \right)\\
=&\Pr\left(\Msg\neq\Usg, \ \left|\left|\St\right|\right|^2\leq \cl \npal \Big| \ \left|\left|\optStl\right|\right|^2\leq \sum_{\indt=\optdpr+1}^{\cl}\npal_\indt \right)\\
&\quad\cdot\Pr\left(\left|\left|\optStl\right|\right|^2\leq \sum_{\indt=\optdpr+1}^{\cl}\npal_\indt\right).
\end{align*}

Since $\left|\left|\St\right|\right|^2=\left|\left|\optStp\right|\right|^2+\left|\left|\optStl\right|\right|^2$, if both $\left|\left|\optStp\right|\right|\leq\sum_{\indt=1}^{\optdpr}\npal_\indt$ and $\left|\left|\optStl\right|\right|\leq\sum_{\indt=\optdpr+1}^{\cl}\npal_\indt$ hold, it is automatically true that $\left|\left|\St\right|\right|^2\leq \cl \npal$. Hence,
\begin{align*}
&\Pr\left(\Msg\neq\Usg, \ \left|\left|\St\right|\right|^2\leq \cl \npal \Big| \ \left|\left|\optStl\right|\right|^2\leq \sum_{\indt=\optdpr+1}^{\cl}\npal_\indt \right)\\
\geq &\Pr\left(\Msg\neq\Usg, \  \left|\left|\optStp\right|\right|^2\leq \sum_{\indt=1}^{\optdpr}\npal_\indt \right)
\end{align*}
yielding
\begin{align*}
&\Pr\left( \Msg\neq\Usg , \ \left|\left|\St\right|\right|^2\leq \cl \npal\right)\\
\geq&\Pr\left(\Msg\neq\Usg, \  \left|\left|\optStp\right|\right|^2\leq \sum_{\indt=1}^{\optdpr}\npal_\indt \right)\\
&\cdot\Pr\left(\left|\left|\optStl\right|\right|^2\leq \sum_{\indt=\optdpr+1}^{\cl}\npal_\indt\right).
\end{align*}

For simplicity, in the following contexts, we denote 
\begin{align*}
\mathbbm{P}_1&\triangleq\Pr\left(\Msg\neq\Usg, \  \left|\left|\optStp\right|\right|^2\leq \sum_{\indt=1}^{\optdpr}\npal_\indt \right)\\
\mathbbm{P}_2&\triangleq\Pr\left(\left|\left|\optStl\right|\right|^2\leq \sum_{\indt=\optdpr+1}^{\cl}\npal_\indt\right).
\end{align*}

We can bound them as below:
\begin{lemma}
	\label{lemma:3}
	There exist constants $\alpha>0$, $\beta>0$ and $\varsigma>0$ sufficiently small such that
	for any fixed block-length $\cl>0$, constants $\varepsilon>0$, $\tau>0$ and quadratic constraints $\spal,\npal>0$,
	\begin{align*}
	\mathbbm{P}_1&\geq\frac{{\cl\varepsilon}-2}{2\log\mn}\cdot\frac{\varsigma}{1+\varsigma}\left(1-\frac{1}{2}e^{-\frac{\alpha^2}{16\npal\spal}}\right)\cdot\frac{\beta}{1+\beta},\\
	\mathbbm{P}_2&\geq\frac{\tau}{1+\tau}.
	\end{align*}
\end{lemma}

Therefore, the probability of error can be bounded from below as
\begin{align}
\apoe\geq&\frac{1}{2}\mathbbm{P}_1\mathbbm{P}_2=\Omega\left(\varepsilon\tau\right),
\end{align}
which is a positive constant for any $\cl$ sufficiently large. As the last step we consider Theorem~\ref{thm:3}. Since $\overline{C}_{\cl}^{\tau}$ can be made arbitrarily close to $C_\cl$ for large $\cl$, Theorem~\ref{thm:1} is proved. In Appendix~\ref{app:3}, we prove the lower bounds on the probabilities $\mathbbm{P}_1$,$\mathbbm{P}_2$ as presented in Lemma~\ref{lemma:3}.

\section{Achievability}
\label{sec:5}
Suppose $P> 2N$. Fix a block-length $\cl\geq 1$ large enough. Let $R=\underline{C}_{\cul}^{\gamma}-\varepsilon$ where $\cul=\frac{\cl}{\theta}$ denotes the \textit{number of chunks} and a \textit{chunk-length} $\theta$ is set to be $\sqrt{\cl}$\footnote{Theoretically, the chunk-length $\theta$ can take a wide range of values as a function of $\cl$ as long as $\lim_{\cl\rightarrow\infty}\frac{\theta(\cl)}{\cl}=0$ and $\lim_{\cl\rightarrow\infty}\theta(\cl)\rightarrow\infty$. But for the sake of presentation, we choose $\theta=\sqrt{\cl}$ everywhere in this work.} . Then the number of codewords is $\mn=2^{nR}$. Our goal is to give a code $\left(p_{\Cwd|\Msg},p_{\Esg|\Rwd}\right)$ with rate $R$ such that for any $N$-constrained causal adversarial noise $\St$ satisfying $\left|\left|\St\right|\right|^2\leq \cl\npal$, the corresponding maximal probability of error $\mpoe$ always converges to zero as the block-length $\cl$ goes to infinity. We are not going to give an explicit construction of a code $\left(p_{\Cwd|\Msg},p_{\Esg|\Rwd}\right)$ with achievable rate $R=\underline{C}_{\cul}^{\gamma}-\varepsilon$. Instead, we use probabilistic argument and construct an ``ensemble'' of stochastic codes. This ensemble of codes follows the same construction of encoder and decoder as in~\cite{chen2015characterization}. Note that by showing the overall maximal probability of error averaging over each instance of the codes goes to zero as the block-length $\cl$ grows, it holds that there exists (implicitly) some code with achievable rate $R$. We specify the encoding and decoding for the aforementioned stochastic codes.

\subsection{Encoding}
Recall that $\msgs$ is the set of messages containing $\mn=2^{\cl R}$ ($R$ is the rate same as above) distinct messages. The \textit{collection of codewords} $\Code$ is the set of all possible codewords. The collection $\Code$ is generated according to some distribution $p_{\Code}$. Once generated, a fixed collection $\code$ is accessible to every party in the communication system (including James). For each message $\msg\in\msgs$, a codeword $\Cwd\left(\msg\right)$ is chosen uniformly at random from a subset of the collection $\code$ denoted by $\cset$. We call $\cset$ a \textit{partial collection} for convenience. We have $\Code=\bigcup_{\msg=1}^{\mn}\Code\left(\msg\right)$.

\begin{definition}[\textit{Division Point}]
	As an abuse of notation, let $\rwdsp$ denote the set containing all length-$\dpc\theta$ prefixes $\rwd_{\leq\dpc\theta}$ and let $\rwdsl$ denote the set containing all length-$\left(\cul-\dpc\right)\theta$ suffixes $\rwd_{>\dpc\theta}$. The corresponding \textit{division point} $1\leq\dpc\leq\cul$ is an integer specifying the lengths of $\rwd_{\leq\dpc\theta}$ and $\rwd_{>\dpc\theta}$ that will be stated clearly once necessary.
\end{definition}

Let $\beta>0$ be a constant. It is convenient to write a partial collection $\cset$ as a concatenation of \textit{sub-collections}: 
$$\cset=\mathcal{C}_1\left(\msg\right)\circ\mathcal{C}_2\left(\msg\right)\circ\cdots\circ\mathcal{C}_{K}\left(\msg\right)$$
 where for all $\indT=1,\ldots,K$, the sub-collection $\mathcal{C}_{\indT}\left(\msg\right)$ is a set containing $2^{\beta}$ randomly generated real-valued length-$\theta$ sequences where $\theta=\frac{\cl}{\cul}$. The notation $\left(\cdot\circ\cdot\right)$ indicates that any combination of those coordinates from the sub-collections $\mathcal{C}_{\indT_1}\left(\msg\right)$and $\mathcal{C}_{\indT_2}\left(\msg\right)$ belongs to the set $\mathcal{C}_{\indT_1}\left(\msg\right)\circ\mathcal{C}_{\indT_2}\left(\msg\right)$. In this sense, the partial collection $\cset$ has $2^{\beta \cul}$ many sequences.
 
Denote by  $\spad^*=\spadl _1^*,\ldots,\spadl _{\cul}^*\in\spasd$ and $\npad^*=\npadl _1^*,\ldots,\npadl _{\cul}^*\in\npasd$ the corresponding sequence \textit{optimizing}\footnote{We do not worry too much about the existence of a global optimal solution of (\ref{eq:3.2.6}). Since the feasible set of optimization~(\ref{eq:3.2.6}) is non-empty, there must be some sequence such that the corresponding objective value is arbitrarily close to the optimal value. Take this sequence as $\spad^*$.} (\ref{eq:3.2.6}).  Precisely, for all $\indT=1,\ldots,K$ and $\msg\in\msgs$, the sub-collection $\mathscr{C}_{\indT}\left(\msg\right)$ is a set of random real-valued  length-$\theta$ sequences such that
\begin{eqnarray}
\mathscr{C}_{\indT}\left(\msg\right)\triangleq\left\{\Clwd_{\left({\indT}-1\right)\theta}\left(i;\msg\right),\ldots,\Clwd_{\indT\theta}\left(i;\msg\right)\right\}_{i=1}^{2^{\beta}}
\end{eqnarray}
wherein for all $i=1,\ldots,2^{\beta}$ and  $\msg\in\msgs$, each of the length-$\theta$ sequence
\begin{align*}
\Cwd_{\indT}\left(\msg\right)\triangleq\Clwd_{\left(\indT-1\right)\theta}\left(i;\msg\right),\Clwd_{\left(\indT-1\right)\theta+1}\left(i;\msg\right),\ldots,\Clwd_{\indT\theta}\left(i;\msg\right)
\end{align*}
 is independently 
chosen from the $\theta$-dimensional ball 
\begin{align*}
\ballP_{\indT}\triangleq\left\{\cwd\in\mathbbm{R}^{\theta}:\left|\left|\cwd\right|\right|^2\leq\spadl _{\indT}^*\right\}
\end{align*}
 uniformly at random. Let $\mathrm{Vol}\left(\ballP_{\indT}\right)$ denote the volume of the $\theta$-dimensional ball. Then the distribution $p_{\Cwd_{\indT}}$ follows that
\begin{align}
\label{eq:4.1}
p_{\Cwd_{\indT}}\left(\cwd\right)=\begin{cases}
\frac{1}{\mathrm{Vol}\left(\ballP_{\indT}\right)} \quad &\text{if } \cwd\in\ballP_{{{t}}}\\
0 \quad &\text{otherwise}
\end{cases}.
\end{align}

To avoid confusion, we write $p_{\Cwd|\Cset}\left(\cwd|\cset\right)$ as the probability for $\Cwd=\cwd$ when the collection of codewords $\cset$ is fixed. Following what we defined above,
given a fixed collection of codewords $\cset$, for all $\msg\in\msgs$, the encoding distribution $p_{\Cwd|\Cset}$ for all $\cwd\in\ballP$ and $\msg\in\msgs$ can be expressed as
\begin{eqnarray}
\label{eq:5.3}
p_{\Cwd|\Cset}\left(\cwd|\cset\right)=\begin{cases}
\frac{1}{2^{\beta\cul}} \quad &\text{if } \ \cwd\in\mathcal{C}\left(\msg\right)\\
0 \quad &\text{otherwise}
\end{cases}.
\end{eqnarray}

It is useful to define the following two sub-collections of codewords:
\begin{align*}
\codep&\triangleq\mathcal{C}_{1}\circ\cdots\circ\mathcal{C}_{\dpc}\\
\codel&\triangleq\mathcal{C}_{\dpc+1}\circ\cdots\circ\mathcal{C}_{\cul}
\end{align*}
for some division point $1\leq\dpc\leq K$ to be be stated explicitly.


\subsection{Decoding}
\label{sec:5.b}
Given a fixed adversarial state $\st\in\ballN$, we define a sequence by considering the consumed power in each $\dpc$-th chunk of length $\theta$.

\begin{definition}[$\indT$-th \textit{Accumulated Power}]
	\label{def:2}
	The $\indT$-th \textit{accumulated power} of an adversarial state $\st\in\ballN$ is defined as
	\begin{align}
	\label{eq:4.3}
	\mathit{\Psi}_\indT(\st)\triangleq \sum_{\indt=\left(\indT-1\right)\theta+1}^{\indT\theta}|\slt_\indt|^2.
	\end{align}
\end{definition}
Let $\npad(\st)=\mathit{\Psi}_1(\st),\ldots,\mathit{\Psi}_{\cul}(\st)$ be the corresponding \textit{accumulated power allocation sequence}. For simplicity, we write $\npad$ as the concrete accumulated power allocation sequence selected by James and it is in the set $\npasd$ since $\st\in\ballN$.

Without knowing the real \textit{accumulated power allocation sequence} $\npad$, the receiver chooses a length-$\cul$  reference sequence measures the power budget that the receiver thinks to be the real one spent by James. The decoding starts at some \textit{starting point} $\dpc=\dpc_0$ to be defined later along the \textit{(consumed) budget reference sequence} $\cbd$ defined as below all the way up to until an estimated message is decoded or it reaches the end point of the chunks and $\dpc=\cul$.

\begin{definition}[\textit{Budget Reference Sequence}]
	\label{def:3}
Let $0<\delta<1$ be a constant.
The \textit{(consumed) budget reference sequence} $\cbd$ is a length-$\cul$ sequence with each $\dpc$-th coordinate defined as
	\begin{align}
	\label{eq:4.10}
	\bost\triangleq \cl \npal-  \sum_{\indT=\dpc+\cul\delta+1}^{\cul}\frac{1}{2}\spadl _\indT^*.
	\end{align}
\end{definition}


Note that some $\bost$ may take a negative value and $\cbd$ is a non-decreasing sequence. Therefore we define a \textit{starting point} $\dpc_0$ after which the budget value $\bost$ becomes positive and attains its real physical meaning. We define the following.
\begin{definition}[\textit{Starting Point of Decoding} $\dpc_0$]
\label{def:8}
An integer $1\leq\dpc_0\leq \cul$ is a \textit{starting point of decoding} if
\begin{align*}
{{F}}_{\dpc_0}\leq 0
\end{align*}
and at the same time
\begin{align*}
{{F}}_{\dpc_0+1}> 0.
\end{align*}	
\end{definition}


The following lemma guarantees that any $\delta$-fraction of the sum of $\spadl^*_{\indT}$'s is neither too large nor too small. The proof is provided in Appendix~\ref{app:5}\footnote{Note that in the proof we presume that for any block-length $\cl$, the optimal solution $C_\cul^{\gamma}({\spal}/{\npal})$ of the optimization~(\ref{eq:3.2.6}) is continuous as a function of SNR for all $\spal/\npal\in (0,1)$. This is a valid assumption in the sense that }. 
\begin{lemma}
	\label{lemma:12}
For any division point $1\leq\dpc\leq\cul-\delta\cul-1$ and any $0<\delta<1$, we have
\begin{align}
\nonumber
\sum_{\indT=\dpc+1}^{\dpc+\delta\cul}\spadl^*_{\indT}=\Theta\left(\delta\cul\theta\right).
\end{align}
\end{lemma}
In particular, the lemma above yields
\begin{align}
\label{eq:5.31}
\sum_{\indT=\dpc+1}^{\dpc+\delta\cul}\spadl^*_{\indT}=\Omega\left(\delta\cul\theta\right),\\
\label{eq:5.32}
\sum_{\indT=1}^{\delta\cul}\spadl^*_{\indT}=\mathcal{O}\left(\delta\cul\theta\right).
\end{align}

We present the following lemma stating the existence of the starting point ${\dpc_0}$. Note that the regime of interests is $\spal-2\npal>0$.

\begin{lemma}
	\label{lemma:11}
Given a small enough constant $1>\delta>0$, the starting point of decoding ${\dpc_0}$ of any budget $\cbd$ reference sequence exists. Moreover,
\begin{align*}
1\leq{\dpc_0}\leq \cul-\cul\delta-1.
\end{align*}
\end{lemma}

\begin{proof}
It suffices to check $\ost_1\leq 0$. Note that in Lemma~\ref{lemma:12}, we have
\begin{align*}
\sum_{\indT=1}^{\delta\cul}\spadl^*_{\indT}&=\mathcal{O}\left(\delta\cul\theta\right).
\end{align*}
By definition,
\begin{align*}
\ost_1=\cl\npal- \sum_{\indT=1+\delta\cul}^{\cul}\frac{1}{2}\spadl _\indT^*&=\cl\npal-\frac{1}{2}\cl\spal+\sum_{\indT=1}^{\delta\cul}\frac{1}{2}\spadl _\indT^*\\
&=\cl\npal-\frac{1}{2}\cl\spal+\mathcal{O}\left(\delta\cul\theta\right).
\end{align*}
Since $\spal>2\npal$, there exists a $\delta>0$ small enough such that $\ost_1<0$.

Moreover, $\ost_{\cul-\cul\delta-1}=\cl\npal>0$ by definition. Therefore $\dpc_0$ exists and ranges between $1$ and $\cul-\cul\delta-1$.
\end{proof}

\subsubsection{List-decoding}
Fix a division point $1\leq\dpc\leq\cul$. 
We define a list of messages:
\begin{definition}[\textit{Pre-list}]
\label{def:9}
Given a partial collection of codewords $\codep$ and a received length-$\dpc\theta$ prefix $\rwdpc$, a \textit{pre-list} $\listas$ is a subset of $\msgs$ that contains all messages with their corresponding length-$\dpc\theta$ prefixes of codewords $\cwdpc$ satisfying $\left|\left|\cwdpc-\rwdpc\right|\right|^2\leq \aost$. In our notation, we write
\begin{align}
\nonumber
&\qquad\lista\triangleq\\
&\Bigg\{\esg\in\msgs: \ \text{There exists} \ \cwdpc\in\codep\left(\esg\right) \ \\
&\qquad\qquad  \text{such that} \ \left|\left|\cwdpc-\rwdpc\right|\right|^2\leq \aost \Bigg\}.
\end{align}
\end{definition}

Intuitively, the list $\listas$ contains all possible transmitted messages assuming $\rwdpc$ has been received and the reference value $\bost$ exactly equals to James'ss consumed power budget $\sum_{T=1}^{\dpc\theta}\npal_{\indt}$. Therefore if the assumption is correct, the list $\listas$ contains the true transmitted message. Otherwise the real message may or may not be included in the list. Based on $\listas$, the receiver next implements the following consistency check, which works together with the list-decoding to select and recover the transmitted real message.

\subsubsection{Consistency Check}
Fix a division point $1\leq\dpc\leq\cul$.

Based on $\listas$, a smaller sub-list of messages can be defined as follows.
\begin{definition}[\textit{Post-list}]
	\label{def:10}
Given a received length-$\left(\cul-\dpc\right)\theta$ suffix $\rwdl$, a collection of codewords $\code$ and a pre-list $\listas$, a \textit{post-list} denoted by $\listvs$ is a subset of $\listas$ such that every message in $\listvs$ with its corresponding suffixes of codewords $\cwdlc$ satisfying $\left|\left|\cwdlc-\rwdlc\right|\right|^2\leq \cl N-\bost$. That is,
\begin{align}
\nonumber
&\qquad\listv\triangleq\\
&\Bigg\{\esg\in\listas: \ \text{There exists} \ \cwdlc\in\ \codel\left(\esg\right) \\
&\qquad \qquad \text{such that}  \ \left|\left|\cwdlc-\rwdlc\right|\right|^2\leq \cl N-\bost\Bigg\}.
\end{align}
\end{definition}

The corresponding cardinalities of $\listas$ and $\listvs$ are denoted by $\clistas$ and $\clistvs$ respectively.

Starting from $\dpc=\dpc_0$ (the starting point $\dpc_0$ is defined in Definition~\ref{def:8}) and pretending that the reference value $\bost$ represents the true remaining power of James, the receiver iteratively implements the following two-step decoding and increase $\dpc$ by one until an estimated message is obtained from $\listvs$ or $\dpc$ reaches $\cul$ (in which case an error message $\mathsf{error}$ is declared). The decoding procedure can be summarized as below.

\begin{mybox}{List-and-Check Decoding}
The estimated message $\Esg$ is set to be
\begin{align}
&\Esg=\begin{cases}
\label{eq:4.14}
\esg \qquad \ \text{if } \exists \ \dpc_0\leq\dpc_{1} \leq\cul \ \text{such that} \  \esg\in\listVs  \\ \qquad \quad \text{ and }  \clistVs=0,  \ \forall \dpc_0\leq\dpc<\dpc_{1} \\
\mathsf{error} \quad  \text{otherwise} 
\end{cases}.
\end{align}
\end{mybox}


Let $\mathcal{V}\triangleq\msgs\bigcup\left\{\mathsf{error}\right\}$ and use $\mathcal{V}_{\msg}\triangleq\mathcal{V}\backslash\left\{\msg\right\}$ for simplicity. We are ready to define the overall maximal probability of error aforementioned at the beginning of this subsection. 

Recall the meaning of maximal probability of error $\mpoe$ in Definition~\ref{def:1}. Let $\mpoec$ denote the corresponding maximal probability of error given a fixed collection of codewords $\code$. Over the randomness of the collection of codewords $\Code$, we define the \textit{overall maximal probability of error} considered in this subsection.
\begin{definition}[\textit{Overall Maximal Probability of Error}]
The \textit{overall maximal probability of error} is denoted as $\ampoe$, which is
\begin{align}
	\label{eq:4.29}
	\ampoe&\triangleq\int_{\mathcal{C}}p_{\mathscr{C}}\left(\mathcal{C}\right)\mpoec\mathrm{d}\mathcal{C}\\	
	\nonumber
	&=\int_{\mathcal{C}}p_{\mathscr{C}}\left(\mathcal{C}\right)\max_{\msg}\sup_{p_{\Rwd|\Cwd}\in\mathsf{Q}}\sum_{\esg\in\mathcal{V}_{\msg}}\int_{\cwd}\int_{\rwd}\\
	&\quad\quad p_{\Cwd|\Code\left(\msg\right)}\left(\cwd|\code\left(\msg\right)\right)p_{\Rwd|\Cwd}\left(\rwd |\cwd\right)p_{\Esg|\Rwd}\left(\esg|\rwd\right)\mathrm{d}\rwd\mathrm{d}\cwd\mathrm{d}\mathcal{C}.
\end{align}
\end{definition}

\begin{remark}
     If we are able to show that \textit{averaging} over all possible collections $\code$ distributed as $p_{\Code}$, the \textit{overall} maximal probability of error $\ampoe$
	goes to zero as $\cl$ goes to infinity, the by a random-coding argument, it is true that there exists some collection of codewords $\mathcal{C}$ such that the corresponding maximal probability of error $\mpoec$ vanishes as $\cl$ grows. In this way the achievability can be proved. In other words, it suffices to design an ensemble of codes $\left\{\left(p_{\Cwd|\Msg},p_{\Esg|\Rwd}\right)\right\}_{\code}$ all with the same rate $R=\underline{C}_{\cul}^{\gamma}-\varepsilon$ and demonstrate the averaged probability of error $\ampoe$ is vanishing (goes to zeros as $\cl$ goes to infinity).
\end{remark}

We present the decoding procedure in the following diagram.

\begin{figure*}[h]
	\begin{center}
			\footnotesize
		\begin{tikzpicture}[node distance=2cm,auto,>=latex']
		\node (q) {$\rwdpc=\rlwd_{1},\ldots,\rlwd_{\dpc\theta}$};
		\node (a) [below of=q, node distance=0.5cm]{$\rwdlc=\rlwd_{\dpc\theta+1},\ldots,\rlwd_{\cl}$};
		\node (d) [right of=q, node distance=3cm] {$\listas$};
		\node (begin) [left= 2cm of {$(a)!0.5!(q)$}] {Set ${\dpc}=\dpc_0$};
		\node (v1) [left= 1.4cm of {$(a)!0.5!(q)$}] {};
		\node (e) [right of=d, node distance=2cm] {$\listvs$};
		\node (f) [below= 1cm of e] {If $|\listvs|\geq 1$};
		\node (v3) [left= 3.2cm of $(f.north)$] {No};
		\node (v4) [left= 2.3cm of $(f.south)$] {Increase $\dpc$ by $\theta$};
		\node (v5) [left= 6cm of f] {If ${\dpc}=\cul+1$};
		\node (g) [below= 1cm of f] {Output a random $\esg\in\listvs$};
		\node (h) [left= 5cm of g] {Output an error message $\esg=\textsf{error}$};
		\draw[thick,->] (q) edge node {} (d);
		\draw[thick,->] (a) edge node {} (e);
		\node [coordinate] (end) [below of=b, node distance=1cm]{};
		\draw[thick,-]  ($(f.north west)$) rectangle ($(f.south east)$);
		\draw[thick,-]  ($(v5.north west)$) rectangle ($(v5.south east)$);
		\draw[thick,->] (d) edge node {} (e);
		\draw[thick,->] (e) edge node {} (f);
		\draw[thick,->] (f) edge node {Yes} (g);
		\draw[thick,->] (f) -| (a);
		\draw[thick,-] (f) -| (v5);
		\draw[thick,->] (v5) edge node {} (h);
		\draw[thick,->] (begin) edge node {} (v1);
		\end{tikzpicture}
	\end{center}
	\caption{Schematic Diagram of Decoding. In the proof we will use union bound and show the probability that the estimated massage is not the transmitted message $\esg\neq\msg$ and the probability that $\esg=\mathsf{error}$ are both asymptotically zero as the chunk size $\cul$ goes to infinity.}	
	\label{fig:3}
\end{figure*}
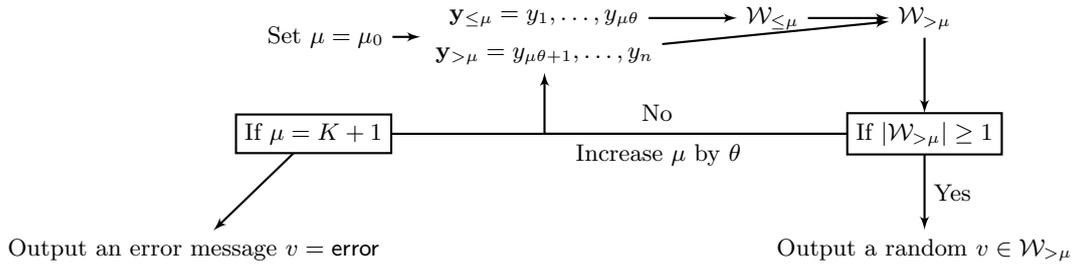

\subsection{Power Allocations}
Let $\npad=\npadl_1,\ldots,\npadl_{\cul}$ be an \textit{accumulated power allocation sequence}. 
We classify $\npad$ into two types---the \textit{high-type} and the \textit{low-type}.

\begin{definition}[\textit{High-Type}]
	 If the accumulated power allocation sequence $\npad$ satisfies 
	\begin{align*}
	\sum_{{\indT}=1}^{\dpc_0+1}\npadl_{\indT}=0,
	\end{align*}
	then we say such a $\npad$ belongs to the \textit{high-type}.
\end{definition}

\begin{definition}[\textit{Low-Type}]
	On the other hand, if the accumulated power allocation sequence $\npad$ satisfies 
	\begin{align*}
	\sum_{{\indT}=1}^{\dpc_0+1}\npadl_{\indT}>0,
	\end{align*}
	then we say such a $\npad$ belongs to the \textit{low-type}.
\end{definition}

\begin{lemma}
	\label{lemma:4}
There exists some point $\dpc_{1}$ with $\dpc_0\leq \dpc_{1} \leq\cul-\cul\delta$ such that
\begin{align*}
\sum_{\indT=1}^{\dpc_{1} }\mathit{\Psi}_{\indT}\geq\post
\end{align*}
and
\begin{align*}
\sum_{\indT=1}^{\dpc_{1} +1}\mathit{\Psi}_{\indT}<\dost.
\end{align*}
\end{lemma}

\begin{remark}
The point $\dpc_{1}$ has a critical operational meaning. When $\dpc=\dpc_{1}$, the accumulated power allocation sequence $\npad$ intersects with the optimizing sequence $\spad^*$. At the point $\dpc_{1}$, later in Lemma~\ref{lemma:5}, we will show our decoder always outputs a set of messages containing the transmitted message $\msg$. Therefore no matter what $\npad$ is, the decoding will always stop at the point $\dpc_1$.
\end{remark}

 \begin{proof}
 	Since an accumulated power allocation sequence $\npad$ either belongs to the low-type or the hight-type, we prove the existence of the point $\dpc_{1}$ for both the two types.
 	
 	Suppose $\npad$ is a high-type sequence. Then by definition, 
 	\begin{align*}
 	\sum_{{\indT}=1}^{\dpc_0+1}\npadl_{\indT}=0
 	\end{align*}
 	where $\dpc_0$ denotes the \textit{starting point of decoding}. Moreover, $\sum_{{\indT}=1}^{\dpc_0}\npadl_{\indT}=0$ since each $\npadl_{\indT}$ is non-negative. Therefore,
 	\begin{align*}
 	{{F}}_{\dpc_0}\leq 0&=	\sum_{{\indT}=1}^{\dpc_0}\npadl_{\indT},\\
 	{{F}}_{\dpc_0+1}> 0&=	\sum_{{\indT}=1}^{\dpc_0+1}\npadl_{\indT}.
 	\end{align*}	
 	
 	Note that in Lemma~\ref{lemma:11}, we demonstrated the existence of $\dpc_0$ and we know $1\leq \dpc_0\leq\cul-\cul\delta-1$.
 	Take $\dpc_{1}=\dpc_0\leq \cul-\cul\delta-1$. We have
 	\begin{align*}
 	\sum_{{\indT}=1}^{\dpc_{1}}\npadl_{\indT}&\geq{{F}}_{\dpc_1},\\
 	\sum_{{\indT}=1}^{\dpc_{1}+1}\npadl_{\indT}&<{{F}}_{\dpc_1+1},
 	\end{align*}
 	and the lemma is true for all high-type sequences.
 	
 	Next if $\npad$ is a low-type sequence. Then 	
 	$\sum_{{\indT}=1}^{\dpc_0+1}\npadl_{\indT}>0$.
 	Note that in the proof of Lemma~\ref{lemma:11}, we get
 	\begin{align*}
 	\ost_{\dpc_0}\leq 0&\leq \sum_{{\indT}=1}^{\dpc_0}\npadl_{\indT},\\
 	\ost_{\cul-\cul\delta-1}=\cl \npal&= \sum_{{\indT}=1}^{\cul}\npadl_{\indT}>\sum_{{\indT}=1}^{\cul-\cul\delta-1}\npadl_{\indT}
 	\end{align*}
 	implying that there exits a point $\dpc_1$ between $\dpc_0$ and $\cul-\cul\delta-1$ such that
 	\begin{align*}
 	\sum_{\indT=1}^{\dpc_{1} }\mathit{\Psi}_{\indT}&\geq\post,\\
 	\sum_{\indT=1}^{\dpc_{1} +1}\mathit{\Psi}_{\indT}&<\dost.
 	\end{align*}
 \end{proof}


\subsection{Sketch of Proof}
  \begin{figure*}[b]
  	\hrule
  	\begin{align}
  	\label{eq:5.34}
  	\mathsf{g}\left(\codep,\ell\right)\triangleq&\max_{\dpc_0\leq\dpc\leq\cul}\sup_{\rwdpc\in\rwdsp}\mathds{1}\left(\clistas>\ell\right),\\
  	\label{eq:5.35}
  	\mathsf{h}\left(\codel,\ell\right)\triangleq&\max_{\dpc_0\leq\dpc\leq\cul} \ \max_{\msg\in\msgs}\sup_{\listas:\clistas\leq\ell}\int_{\cwdlc}p_{\Cwdlc|\Csetl}\left(\cwdlc|\codel\left(\msg\right)\right)\\
  	&\qquad\qquad\quad \ \sup_{\rwdlc\in\ballx}\mathds{1}\left(\left|\listvs\backslash\left\{\msg\right\}\right|>0\right)\mathrm{d}\cwdlc.
  	\end{align}
  \end{figure*}
Recall by our encoding construction, the conditional probability below implies a uniform distribution over all prefixed codewords (not available to Bob and James):
	\begin{eqnarray*}
		&p_{\Cwdlc|\Csetl}\left(\cwdlc|\csetl\right)\\
		=&\begin{cases}
			\frac{1}{2^{\beta\left(\cul-\dpc\right)}} \quad &\text{if } \ \cwdlc\in\csetl\\
			0 \quad &\text{otherwise}
		\end{cases}.
	\end{eqnarray*}

 For notational convenience, define
 \begin{align*}
 &\quad \mathsf{f}\left(\code\right)\triangleq\\
 &\max_{\msg\in\msgs}\min_{\dpc_0\leq\dpc\leq\cul}\int_{\cwdpc}\int_{\cwdlc}p_{\Cwdpc,\Cwdlc|\Code\left(\msg\right)}\left(\cwdpc,\cwdlc|\code\left(\msg\right)\right)\\
 &\sup_{\rwdpc\in\rwdsp}\sup_{\rwdlc\in\rwdsl}\sum_{\esg\in\mathcal{V}_{\msg}}p_{\Esg|\Rwdpc,\Rwdlc}\left(\esg|\rwdpc,\rwdlc\right)\mathrm{d}\cwdlc\mathrm{d}\cwdpc.
 \end{align*}


The supremum over $p_{\Rwd|\Cwd}\in\mathsf{Q}$ in $\ampoe$  can be simplified into a double-supremum over the prefixes $\rwdpc\in\rwdsp$ and suffixes $\rwdlc\in\rwdsl$ with any division point $\dpc_0\leq\dpc\leq\cul$ such that
\begin{align}
\label{eq:4.2}
\ampoe\leq
&\int_{\mathcal{C}}p_{\mathscr{C}}\left(\mathcal{C}\right){ \mathsf{f}\left(\code\right)}\mathrm{d}\mathcal{C}.
\end{align}

Let $\ell$ be an integer between $0$ and $\mn$.
Moreover, we define two functions $	\mathsf{g}\left(\codep,\ell\right)$ and $\mathsf{h}\left(\codel,\ell\right)$ in (\ref{eq:5.34}) and (\ref{eq:5.35}) for notational convenience.
The conditional probability in $\mathsf{h}\left(\codel,\ell\right)$ is same as the one in (\ref{eq:4.2}). The first supremum in $\mathsf{g}\left(\codep,\ell\right)$  is taken over all $\rwdpc\in\rwdsp$ such that $\left|\left|\rwdpc\right|\right|^2\leq \cl N$  and the second supremum in $\mathsf{h}\left(\codel,\ell\right)$ is over all $\rwdlc\in\ballx$ with 
\begin{align*}
&\quad \ballx\triangleq\\
&\left\{\rwdlc\in\mathbbm{R}^{\cl-\dpc\theta}:\left|\left|\cwdlc-\rwdlc\right|\right|^2\leq \cl N-\bost\right\}.
\end{align*}



\begin{lemma}
	\label{lemma:5}
For any collection of codewords $\mathcal{C}$ and any integer  $1\leq\ell\leq\mn$,
\begin{align}
\label{eq:4.4}
&\mathsf{f}\left(\code\right)\leq
\mathsf{g}\left(\codep,\ell\right)+\mathsf{h}\left(\codel,\ell\right).
\end{align}
\end{lemma}

 \begin{proof}
 	Given a $\msg\in\msgs$ and for any $\esg\in\mathcal{V}_{\msg}$, we consider two cases -- $\esg=\mathsf{error}$ and $\esg\in\msgs\backslash\left\{\msg\right\}$. 
 	
 	When $\esg=\mathsf{error}$, we claim that for any collection of codewords $\mathcal{C}$, any selected message $\msg$, any integer  $0\leq\ell\leq\mn$ and any $\rwdpc\in\rwdsp$, $\rwdlc\in\rwdsl$, $$p_{\Esg|\Rwdpc,\Rwdlc}\left(\esg|\rwdpc,\rwdlc\right)=0.$$ This follows from our decoding process. Since from Lemma~\ref{lemma:4}, there exists some $\dpc_{1}$ with $\dpc_0\leq \dpc_{1} \leq\cul-\cul\delta-1$ such that
 	\begin{align*}
 	\sum_{\indT=1}^{\dpc_{1} }\mathit{\Psi}_\indT&\geq\post,\\
 	\sum_{\indT=1}^{\dpc_{1} +1}\mathit{\Psi}_\indT&<\dost.
 	\end{align*}
 	
 	Therefore, for any transmitted codeword $\cwd=\cwdpc^{\dpc_{1} \theta}\circ\cwdlc^{\dpc_{1} \theta}$ as a concatenation of a length-$\dpc\theta$ prefix and a length-$\left(\cl-\dpc\theta\right)$ suffix,
 	\begin{align*}
 	&\left|\left|\cwdpc-\rwdpc\right|\right|^2
 	\leq\sum_{T=1}^{\dpc_{1} \theta}\npal_{\indt}=\sum_{\indT=1}^{\dpc_{1} }\mathit{\Psi}_\indT\leq\sum_{\indT=1}^{\dpc_{1} +1}\mathit{\Psi}_\indT< \dost
 	\end{align*}
 	and
 	\begin{align*}
 	\left|\left|\cwdlc-\rwdlc\right|\right|^2
 	&\leq\sum_{T=\dpc_{1} \theta+1}^{\cl}\npal_{\indt}=\sum_{\indT=\dpc_{1} +1}^{\cul}\mathit{\Psi}_\indT\\
 	&= \cl N-\sum_{\indT=1}^{\dpc_{1} }\mathit{\Psi}_\indT\leq \cl N- \post
 	\end{align*}
 	implying that $\msg\in\listd$ no matter what $\mathcal{C}$ and $\rwd$ are. In this sense, the decoding will end at or before the point $\dpc=\dpc_1$.
 	
 	When $\esg\in\msgs\backslash\left\{\msg\right\}$, it suffices to note that an estimated message $\Esg=\esg$ is decoded only if it passes the consistency check, \textit{i.e.,} $\esg\in\listvs\backslash\left\{\msg\right\}$ for some $\dpc_0\leq\dpc\leq\cul$, received codeword $\rwd$ and the corresponding selected message $\msg$.  
 	Conditioned on whether the size of the pre-list $\listas$ is larger than $\ell$ or not, we consider $\mathsf{f}\left(\codep,\ell\right)$ and $\mathsf{g}\left(\codel,\ell\right)$ separately. 
 	
 	If the pre-list size $\clistas$ is larger than $\ell$, we declare an error directly.	Otherwise, if the size $\clistas$ is smaller than or equal to $\ell$, we condition on the fact $\clistas\leq\ell$. In this case, an error will happen only if some message $\esg\in\msgs\backslash\left\{\msg\right\}$ is in the post-list $\listvs$. By considering all such pre-lists $\listas$, we get if $\clistas\leq\ell$,
 	\begin{align*}
 \qquad\mathsf{f}\left(\code\right)
 \leq&\max_{\msg\in\msgs}\min_{\dpc_0\leq\dpc\leq\cul}\\
 &\sup_{\listas:\clistas\leq\ell}\int_{\cwdlc}p_{\Cwdlc|\Csetl}\left(\cwdlc|\codel\left(\msg\right)\right)\\
 &\quad \ \sup_{\rwdlc\in\rwdsl}\mathds{1}\left(\left|\listvs\backslash\left\{\msg\right\}\right|>0\right)\mathrm{d}\cwdlc
 	\end{align*}
 since given a worst-case $\listas\leq\ell$, the probability $p_{\Esg|\Rwdpc,\Rwdlc}\left(\esg|\rwdpc,\rwdlc\right)$  is independent with $\Rwdpc$ and $\Cwdpc$.
 	
 Moreover, the last supremum over all $\rwdlc\in\rwdsl$ can be simplified as another supremum over all $\rwdlc$ in a smaller set
 \begin{align*}
 &\quad \ballx\\
 &=\left\{\rwdlc\in\mathbbm{R}^{\cl-\dpc\theta}:\left|\left|\cwdlc-\rwdlc\right|\right|^2\leq \cl \npal-\bost\right\}.
 \end{align*}
 	This is because for all accumulated power allocation sequence $\npad$, we always have 
 	$$\sum_{{\indT}=\dpc+1}^{\cul}\npadl_{\indT}\leq \cl\npal-\bost$$ meaning that
 	James's remaining power is always bounded from above by the reference value $\npal-\bost$, no matter $\npad$ belongs to the low type or the high type. The reason is that for all $\dpc_0\leq \dpc\leq\dpc_1$,  we have both $	\sum_{{\indT}=\dpc_0+1}^{\cul}\npadl_{\indT}\leq \cl\npal-\ost_{\dpc_0}$ and $	\sum_{{\indT}=\dpc_1+1}^{\cul}\npadl_{\indT}\leq \cl\npal-\ost_{\dpc_1}$ and  $\dpc_1$ is the first point when the accumulated power allocation sequence $\npad$ intersects with our reference sequence $\cbd$.
 	
 	Therefore using union bound and maximizing over all $\dpc_0\leq\dpc\leq\cul$, the RHS of (\ref{eq:4.4}) is obtained with $\mathsf{g}\left(\codep,\ell\right)$and $\mathsf{h}\left(\codel,\ell\right)$ defined in (\ref{eq:5.34}) and (\ref{eq:5.35}) respectively.
 \end{proof}



 For simplicity, denote respectively
 \begin{align}
\nonumber
   &\mathbbm{P}_3\triangleq\int_{\mathcal{C}}p_{\mathscr{C}}\left(\mathcal{C}\right)\mathsf{g}\left(\codep,\ell\right)\mathrm{d}\mathcal{C},\\
   \label{eq:4.25}
   &\mathbbm{P}_4\triangleq\int_{\mathcal{C}}p_{\mathscr{C}}\left(\mathcal{C}\right)\mathsf{h}\left(\codel,\ell\right)\mathrm{d}\mathcal{C}.
 \end{align}
 
 Putting (\ref{eq:4.4}) above into the RHS of (\ref{eq:4.2}), we can write
\begin{align}
\label{eq:4.24}
\ampoe
&\leq\int_{\mathcal{C}}p_{\mathscr{C}}\left(\mathcal{C}\right) \mathsf{f}\left(\mathcal{C}\right)\leq\mathbbm{P}_3+\mathbbm{P}_4.
\end{align}


 We give bounds on $\mathbbm{P}_3$ and $\mathbbm{P}_4$ above respectively.

\begin{lemma}
\label{lemma:6}
Let $R=\underline{C}_{\cul}^{\gamma}-\varepsilon$. For all $\ell>0$, $\beta>0$, $\theta>0$ and constants $\gamma>0$, $\varepsilon>0$, there exists a division point $\dpc_0\leq\dpc\leq\cul-\delta\cul-1$ and a constant $0<\delta<1$ such that
	\begin{align}
\nonumber
	 \mathbbm{P}_3\leq& \cul\left(4\sqrt{\spal}\log\cl\right)^{\cul}\\
	 	   \label{eq:4.0}
	 &\quad \cdot2^{\cl R-\left(\ell+1\right)\left[\varepsilon\theta-\left(\beta+\log\left(\frac{\sqrt{\cl}\log\cl}{4\theta}\right)\right)\right]\cul},\\
	 \nonumber
	\mathbbm{P}_4\leq& \cul\mn{ \mn\choose\ell  }\left[\ell2^{\beta\left(\cul-\dpc\right)}e^{-\frac{\Omega\left(\cl\delta^2\right)}{2\spal}}\right]^{\lfloor\eta 2^{\beta\left(\cul-\dpc\right)}\rfloor}+\eta
	\end{align}
	when the block-length $\cl$ and the number of chunks $\cul$ are sufficiently large. The randomness relies on the construction of collection of codewords $\Code$.
\end{lemma}

Therefore, setting
\begin{align}
\nonumber
\beta&= 2,\\
\nonumber
\ell&= \frac{2}{\varepsilon},\\
\label{eq:5.20}
\theta&=\sqrt{\cl},\\
\nonumber
\gamma&=\Theta(\delta)=\Theta(1),
\end{align}
it follows that when the block-length $\cl$ is sufficiently large,
\begin{align*}
&\varepsilon\theta-\left(\beta+\log\left(\frac{\sqrt{\cl}\log\cl}{4\theta}\right)\right)\\
=&\varepsilon\theta+2+\log\theta-\beta-\frac{1}{2}\log\cl-\log\log\cl>\frac{\varepsilon\theta}{2}.
\end{align*}
Rewriting (\ref{eq:4.0}) based on above, since $\cul=\frac{\cl}{\theta}=\sqrt{\cl}$,
\begin{align*}
\mathbbm{P}_3\leq&\cul \left(4\sqrt{\spal}\log\cl\right)^{\cul}2^{\cl R-\left(\ell+1\right)\big[\varepsilon\theta-\left(\beta+\log\left(\frac{\sqrt{\cl}\log\cl}{4\theta}\right)\right)\big]\cul}\\
=&2^{\cul\log\left(4\sqrt{\spal}\log\cl\right)+\log\cul-\left(\ell+1\right)\frac{\varepsilon\theta\cul}{2}}\\
\leq&2^{\cul\log\left(4\sqrt{\spal}\log\cl\right)+\log\cul-\cl}\\
=&\exp\left(-\Omega\left(\cl\right)\right).
\end{align*}

Moreover, since $\dpc_0\leq\dpc\leq\cul-\delta\cul-1$, when $\cl$ is large enough we get
\begin{align}
 \label{eq:5.4}
 &\ell2^{\beta\left(\cul-\dpc\right)}\exp\left[-\frac{\Omega\left(\cl\delta^2\right)}{2\spal}\right]\\
=& \exp\left\{\ln\ell+\ln2\left[\beta\left(\cul-\dpc\right)-\frac{\Omega\left(\cl\delta^2\right)}{2\spal}\right]\right\}\\
\label{eq:5.4.1}
\leq&\exp\left[\ln\frac{2}{\varepsilon}+\ln2\left(\beta\cul\right)-\frac{\Omega\left(\cl\delta^2\right)}{2\spal}\right]\\
\nonumber
=&\exp\left[{-{\Omega\left(\cl\delta^2+\ln\varepsilon\right)}}\right]
\end{align}
where the inequality~(\ref{eq:5.4.1}) comes from the fact that $\dpc\leq \cul-\delta\cul-1$. Therefore,
\begin{align}
\nonumber
\mathbbm{P}_4
\leq& \eta+2^{\cl R\left(\ell+2\right)}\left\{e^{{-{\Omega\left(\cl\delta^2+\ln\varepsilon\right)}}}\right\}^{\eta \exp\left[{\Omega\left(\delta\cul\right)}\right]}\\
\nonumber
\leq&\eta+2^{\mathcal{O}\left(\frac{\cl}{\varepsilon}\right)} \left\{e^{{-{\Omega\left(\cl\delta^2+\ln\varepsilon\right)}}}\right\}^{\eta \exp\left[{\Omega\left(\delta\sqrt{\cl}\right)}\right]}\\
\nonumber
=&\eta+\exp\left\{-{\Omega\left(\left(\cl+\ln\varepsilon\right)\eta e^{\Omega\left(\sqrt{\cl}\right)}-\frac{\cl}{\varepsilon}\right)}\right\}
\end{align}
since $\delta$ is a constant. Taking $\eta=e^{c\sqrt{\cl}}$ for some constant $c>0$ and using (\ref{eq:4.24}), we obtain
\begin{align*}
\ampoe\leq \mathbbm{P}_3+\mathbbm{P}_4\leq\exp\left(-\Omega\left(\cl\right)\right).
\end{align*}
%
 As the last step we consider Theorem~\ref{thm:3}. Since $\underline{C}_{\cul}^{\gamma}$ can be made arbitrarily close to $C_\cul$ for large $\cul$,
 we get the desired achievability result stated in Theorem~\ref{thm:2}.

\section{Summary}
This paper explores the capacity region of a quadratically constrained channel with a causal adversary. 
The assumption of causality in the channel model makes finding the channel capacity a challenging problem. By arguing the corresponding converse and achievability in Section~\ref{sec:4} and Section~\ref{sec:5} respectively, this work characterizes the capacity region as a limit of optimal objective values of some optimization problems; it also provides numerical validation of a high rate of convergence.

Different from the binary bit-flip channel with a causal adversary as studied in~\cite{dey2013upper,chen2015characterization}, both the transmitter and adversary in the quadratically constrained channel can choose codewords from an $\cl$-dimensional ($\cl$ represents the block-length) Euclidean space. To tackle such
a discrete channel with a continuous alphabet,
 many novel techniques are developed to prove Theorem~\ref{thm:1} and Theorem~\ref{thm:2} in this work.
 
Maximizing over power allocation sequences representing the energy distribution of various codes and minimizing over power allocation sequences of adversaries together with a division point, the channel capacity can be characterized naturally under a framework of optimization problems. Interestingly, dislike the case for binary bit-flip channel, the uniform power allocation is \textit{not} the optimal solution in general. Indeed, a higher rate is achievable as our simulation in Section~\ref{sec:3} shows.

In Section~\ref{sec:4}, we design a new attack strategy called \textit{scaled babble-and-push}. Based on the attack, we prove a \textit{converse} such that for any block-length $\cl$ large enough, any stochastic code with rate larger than $\overline{C}_{\cl}^{\tau}$ is never achievable. The lower bound on the probability of error in the proof Lemma~\ref{lemma:0} is motivated by~\cite{dey2013upper}.  In~\cite{dey2013upper}, to prove a converse for stochastic codes, the conditional  probability $p_{\Cwd|\Msg}$ for each codeword given a message is quantized using a quantization level that is exponentially small. Since there are only $2^\cl$ many distinct codewords in an $\cl$-dimensional Hamming cube, the error caused by the quantization is negligible. However, in our model, there are uncountably many possible codewords even for a reasonably small block length. To overcome this issue, we directly decompose the probability in Lemma~\ref{lemma:0} into two parts, which are bounded from below separately in Lemma~\ref{lemma:3} using probabilistic methods like Markov's inequality, concentration of sum of Gaussian random variables, Fano's inequality~\cite{fano1961transmission} and so on.

In Section~\ref{sec:5}, an \textit{achievability} is proved such that for block-length $\cl$ large enough, there exists a stochastic code with achievable rate smaller than $\underline{C}_{\cul}^{\gamma}$. Moreover, both $\underline{C}_{\cul}^{\gamma}$ and $\overline{C}_{\cl}^{\tau}$ are close to each other for block-length $\cl$ sufficiently large. The code is constructed as a  concatenation of independent chunks. The idea of such a construction and the corresponding decoding procedure come from~\cite{chen2015characterization}. Nonetheless, since our channel model involves codewords from a continuous space, the probability of error is analyzed in a quite different way. For example, in Appendix~\ref{app:4.1}, we cover the space by a large hypercube. Then we divide the cube into small sub-cubes and apply union bound; we use Hoeffding's inequality in Appendix~\ref{app:4.2} to concentrate a sum of sub-Gaussian random variables. 

\section*{Acknowledgment}
This work was partially funded by a grant from the University Grants Committee of the Hong Kong Special Administrative Region (Project No.\ AoE/E-02/08), RGC GRF grants 14208315 and 14313116, NSF grant 1526771, and a grant from Bharti Centre for Communication in IIT Bombay.

\newpage
\addcontentsline{toc}{section}{Bibliography}
{\bibliographystyle{IEEEtran}
	\bibliography{ref}}

\begin{thebibliography}{10}
\providecommand{\url}[1]{#1}
\csname url@samestyle\endcsname
\providecommand{\newblock}{\relax}
\providecommand{\bibinfo}[2]{#2}
\providecommand{\BIBentrySTDinterwordspacing}{\spaceskip=0pt\relax}
\providecommand{\BIBentryALTinterwordstretchfactor}{4}
\providecommand{\BIBentryALTinterwordspacing}{\spaceskip=\fontdimen2\font plus
\BIBentryALTinterwordstretchfactor\fontdimen3\font minus
  \fontdimen4\font\relax}
\providecommand{\BIBforeignlanguage}[2]{{%
\expandafter\ifx\csname l@#1\endcsname\relax
\typeout{** WARNING: IEEEtran.bst: No hyphenation pattern has been}%
\typeout{** loaded for the language `#1'. Using the pattern for}%
\typeout{** the default language instead.}%
\else
\language=\csname l@#1\endcsname
\fi
#2}}
\providecommand{\BIBdecl}{\relax}
\BIBdecl

\bibitem{shannon1949communication}
C.~E. Shannon, ``Communication in the presence of noise,'' \emph{Proceedings of
  the IRE}, vol.~37, no.~1, pp. 10--21, 1949.

\bibitem{ganti2000mismatched}
A.~Ganti, A.~Lapidoth, and I.~E. Telatar, ``Mismatched decoding revisited:
  General alphabets, channels with memory, and the wide-band limit,''
  \emph{IEEE Transactions on Information Theory}, vol.~46, no.~7, pp.
  2315--2328, 2000.

\bibitem{haddadpour2013avcs}
F.~Haddadpour, M.~J. Siavoshani, M.~Bakshi, and S.~Jaggi, ``On avcs with
  quadratic constraints,'' in \emph{Information Theory Proceedings (ISIT), 2013
  IEEE International Symposium on}.\hskip 1em plus 0.5em minus 0.4em\relax
  IEEE, 2013, pp. 271--275.

\bibitem{hughes1987gaussian}
B.~Hughes and P.~Narayan, ``Gaussian arbitrarily varying channels,'' \emph{IEEE
  Transactions on Information Theory}, vol.~33, no.~2, pp. 267--284, 1987.

\bibitem{blachman1962capacity}
N.~Blachman, ``On the capacity of a band-limited channel perturbed by
  statistically dependent interference,'' \emph{IRE Transactions on Information
  Theory}, vol.~8, no.~1, pp. 48--55, 1962.

\bibitem{rankin1955closest}
R.~A. Rankin, ``The closest packing of spherical caps in n dimensions,''
  \emph{Glasgow Mathematical Journal}, vol.~2, no.~3, pp. 139--144, 1955.

\bibitem{kabatiansky1978bounds}
G.~A. Kabatiansky and V.~I. Levenshtein, ``On bounds for packings on a sphere
  and in space,'' \emph{Problemy Peredachi Informatsii}, vol.~14, no.~1, pp.
  3--25, 1978.

\bibitem{blackwell1960capacities}
D.~Blackwell, L.~Breiman, and A.~Thomasian, ``The capacities of certain channel
  classes under random coding,'' \emph{The Annals of Mathematical Statistics},
  vol.~31, no.~3, pp. 558--567, 1960.

\bibitem{lapidoth1998reliable}
A.~Lapidoth and P.~Narayan, ``Reliable communication under channel
  uncertainty,'' \emph{IEEE Transactions on Information Theory}, vol.~44,
  no.~6, pp. 2148--2177, 1998.

\bibitem{csiszar2011information}
I.~Csiszar and J.~K{\"o}rner, \emph{Information theory: coding theorems for
  discrete memoryless systems}.\hskip 1em plus 0.5em minus 0.4em\relax
  Cambridge University Press, 2011.

\bibitem{dey2013upper}
B.~K. Dey, S.~Jaggi, M.~Langberg, and A.~D. Sarwate, ``Upper bounds on the
  capacity of binary channels with causal adversaries,'' \emph{IEEE
  Transactions on Information Theory}, vol.~59, no.~6, pp. 3753--3763, 2013.

\bibitem{chen2015characterization}
Z.~Chen, S.~Jaggi, and M.~Langberg, ``A characterization of the capacity of
  online (causal) binary channels,'' in \emph{Proceedings of the Forty-Seventh
  Annual ACM on Symposium on Theory of Computing}.\hskip 1em plus 0.5em minus
  0.4em\relax ACM, 2015, pp. 287--296.

\bibitem{chen2016capacity}
------, ``The capacity of online (causal) q-ary error-erasure channels,'' in
  \emph{Information Theory (ISIT), 2016 IEEE International Symposium on}.\hskip
  1em plus 0.5em minus 0.4em\relax IEEE, 2016, pp. 915--919.

\bibitem{bassily2014causal}
R.~Bassily and A.~Smith, ``Causal erasure channels,'' in \emph{Proceedings of
  the Twenty-Fifth Annual ACM-SIAM Symposium on Discrete Algorithms}.\hskip 1em
  plus 0.5em minus 0.4em\relax Society for Industrial and Applied Mathematics,
  2014, pp. 1844--1857.

\bibitem{dey2010coding}
B.~K. Dey, S.~Jaggi, M.~Langberg, and A.~D. Sarwate, ``Coding against delayed
  adversaries,'' in \emph{Information Theory Proceedings (ISIT), 2010 IEEE
  International Symposium on}.\hskip 1em plus 0.5em minus 0.4em\relax IEEE,
  2010, pp. 285--289.

\bibitem{dey2016bit}
------, ``A bit of delay is sufficient and stochastic encoding is necessary to
  overcome online adversarial erasures,'' in \emph{Information Theory (ISIT),
  2016 IEEE International Symposium on}.\hskip 1em plus 0.5em minus 0.4em\relax
  IEEE, 2016, pp. 880--884.

\bibitem{sarwate2012avc}
A.~D. Sarwate, ``An avc perspective on correlated jamming,'' in \emph{Signal
  Processing and Communications (SPCOM), 2012 International Conference
  on}.\hskip 1em plus 0.5em minus 0.4em\relax IEEE, 2012, pp. 1--5.

\bibitem{plotkin1960binary}
M.~Plotkin, ``Binary codes with specified minimum distance,'' \emph{IRE
  Transactions on Information Theory}, vol.~6, no.~4, pp. 445--450, 1960.

\bibitem{elias1957list}
P.~Elias, ``List decoding for noisy channels,'' 1957.

\bibitem{sarwate2012list}
A.~D. Sarwate and M.~Gastpar, ``List-decoding for the arbitrarily varying
  channel under state constraints,'' \emph{IEEE Transactions on Information
  Theory}, vol.~58, no.~3, pp. 1372--1384, 2012.

\bibitem{fano1961transmission}
R.~M. Fano and D.~Hawkins, ``Transmission of information: A statistical theory
  of communications,'' \emph{American Journal of Physics}, vol.~29, no.~11, pp.
  793--794, 1961.

\bibitem{boyd2004convex}
S.~Boyd and L.~Vandenberghe, \emph{Convex optimization}.\hskip 1em plus 0.5em
  minus 0.4em\relax Cambridge university press, 2004.

\bibitem{chiani2003new}
M.~Chiani, D.~Dardari, and M.~K. Simon, ``New exponential bounds and
  approximations for the computation of error probability in fading channels,''
  \emph{IEEE Transactions on Wireless Communications}, vol.~2, no.~4, pp.
  840--845, 2003.

\bibitem{voelker2017efficiently}
A.~R. Voelker, J.~Gosmann, and T.~C. Stewart, ``Efficiently sampling vectors
  and coordinates from the n-sphere and n-ball,'' 2017.

\bibitem{hoeffding1963probability}
W.~Hoeffding, ``Probability inequalities for sums of bounded random
  variables,'' \emph{Journal of the American statistical association}, vol.~58,
  no. 301, pp. 13--30, 1963.

\end{thebibliography}

 \newpage
 \appendix
 \subsection{Proof of Theorem~\ref{thm:3}}
 \label{app:1}
 First, we show for any  $0<A<1$, there exists a $\tau>0$ such that  $\overline{C}^{\tau}_\cl\leq C_\cl-({1}/{2})\log\left(1-A\right)$. Notice that the only difference between optimizations~(\ref{eq:3.2.5}) and~(\ref{eq:3.2.7}) is the \textit{energy-bounding conditions} shown below respectively for~(\ref{eq:3.2.5}) and~(\ref{eq:3.2.7}):
 \begin{align}
 \label{eq:app.5}
 \textrm{~(\ref{eq:3.2.5})}:\qquad\cl \spal-\sum\limits_{\indt=1}^{\dpr}\spal_{\indt} &\leq \left(1-\tau\right)\left(2\cl\npal-\sum\limits_{\indt=1}^{\dpr}2\npal_{\indt}\right),\\
 \nonumber
 \textrm{~(\ref{eq:3.2.7})}:\qquad\cl \spal-\sum\limits_{\indt=1}^{\dpr}\spal_{\indt} &\leq 2\cl\npal-\sum\limits_{\indt=1}^{\dpr}2\npal_{\indt}.
 \end{align}
 
The constraints above imply that given any $\spa\in\spas$, a feasible $(\npa,\dpr)$ pair of the optimization~(\ref{eq:3.2.5}) satisfying the energy-bounding is also feasible to~(\ref{eq:3.2.7}). The goal is to show that the optimal values of the two optimization are close if the slackness $\tau$ is small enough.
 
 \begin{figure}[H]	
 	\centering
 	\includegraphics[scale=0.27]{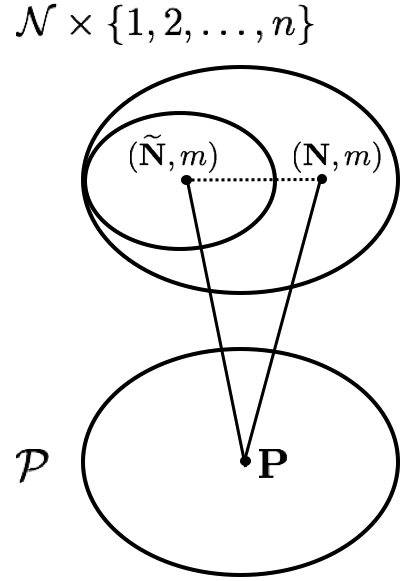}
 	\caption{The feasible sets of the two optimizations~(\ref{eq:3.2.5}) and~(\ref{eq:3.2.7}) are illustrated. }
 	 \label{Fig:5}
 \end{figure}
 
 We use Figure~\ref{Fig:5}. to explain the intuition behind the proof. For any $\spa\in\spas$, let $(\npa,\dpr)$ be the corresponding minimizer of optimization~(\ref{eq:3.2.7}).  Denote by ${C}_\cl ({\spa},{\npa},\dpr)$ and $\overline{C}^{\tau}_\cl ({\spa},{\npa},\dpr)$ the corresponding objective values of optimizations~(\ref{eq:3.2.5}) and~(\ref{eq:3.2.7}) when the variables are set to be $({\spa},{\npa},\dpr)$. We will show that for any such $\spa$ and $(\npa,\dpr)$, there exists a pair $(\widetilde{\npa},\dpc)$ with $\widetilde{\npa}\in\npas$ such that first, $(\spa,\widetilde{\npa},\dpr)$ is a feasible solution of optimization~(\ref{eq:3.2.5}); second, the objective values $\overline{C}^{\tau}_\cl ({\spa},\widetilde{\npa},\dpr)$ does not differ too much from $\overline{C}_\cl ({\spa},{\npa},\dpr)$.
 
Given $({\spa},{\npa},\dpr)$, a feasible solution of optimization~(\ref{eq:3.2.7}) where ${\npa}={\npal}_1,\ldots,{\npal}_\cl$. We construct a sequence $\widetilde{\npa}=\widetilde{\npal}_1,\ldots,\widetilde{\npal}_\cl$ with
 \begin{align*}
 \widetilde{\npal}_\indt=\begin{cases}
{\npal}_\indt-\frac{2\tau\cl\npal}{\cl-1}\quad &\text{if}\quad \indt=1,\ldots,\cl-1\\
{\npal}_\indt+2\tau\cl\npal \quad &\text{if}\quad \indt=\cl
 \end{cases}.
 \end{align*}
 
 Then $({\spa},\widetilde{\npa},\dpr)$ is a feasible solution of optimization~(\ref{eq:3.2.5}) since
 \begin{align*}
 &\left(1-\tau\right)\left(2\cl\npal-\sum\limits_{\indt=1}^{\dpr}2\widetilde{\npal}_{\indt}\right)\\
 \geq &\left(1-\tau\right)2\cl\npal-\sum\limits_{\indt=1}^{\dpr}2\widetilde{\npal}_{\indt}=2\cl\npal-\sum\limits_{\indt=1}^{\dpr}2\widetilde{\npal}_{\indt}.
 \end{align*}
 
  It follows that
 \begin{align}
 \nonumber
\overline{C}^{\tau}_\cl\left({\spa},\widetilde{\npa},\dpr\right)&=\frac{1}{2\cl}\sum_{\indt=1}^{\dpr}\log\frac{{\spal}_\indt}{\widetilde{\npal}_\indt}\\
 \label{eq:a.x.1}
 &=\frac{1}{2\cl}\sum_{\indt=1}^{\dpr}\log\frac{{\spal}_\indt}{{\npal}_\indt-\frac{2\tau\cl\npal}{\cl-1}}.
 \end{align}

Recall in Theorem~\ref{thm:1} the constant $\tau>0$ can be made arbitrarily small. Denote $\underline{{\npal}}\triangleq\min\{{\npal}_\indt\}_{\indt=1}^{\cl}>0$ the minimal value of ${\npa}$. Take $\tau={A\underline{{\npal}}}/{(4\npal)}$ where $A$ denotes an arbitrarily small constant for convenience.
 
 Continuing from (\ref{eq:a.x.1}), for large enough $\cl$, the denominator inside the logarithmic term satisfies
 \begin{align*}
  {\npal}_\indt-\frac{2\tau\cl\npal}{\cl-1}&={\npal}_\indt\left(1-\frac{2\tau\cl\npal}{{\npal}_\indt\left(\cl-1\right)}\right)\\
  &\geq {\npal}_\indt\left(1-A\right),
  \end{align*}
which gives
 \begin{align}
\overline{C}^{\tau}_\cl\left({\spa},\widetilde{\npa},\dpr\right)&\leq\frac{1}{2\cl}\sum_{\indt=1}^{\dpr}\log\frac{{\spal}_\indt}{\left(1-A\right){\npal}_\indt}\\
 &=\frac{1}{2\cl}\sum_{\indt=1}^{\dpr}\log\frac{{\spal}_\indt}{{\npal}_\indt}-\frac{\dpr}{2\cl}\log\left(1-A\right)\\
 &\leq \frac{1}{2\cl}\sum_{\indt=1}^{\dpr}\log\frac{{\spal}_\indt}{{\npal}_\indt}-\frac{1}{2}\log\left(1-A\right)\\
 \label{eq:app.1}
 &=C_\cl\left({\spa},{\npa},\dpr\right)-\frac{1}{2}\log\left(1-A\right).
 \end{align}
 
 Inequality~(\ref{eq:app.1}) works for every $\spa\in\spas$. Therefore, maximizing over all $\spa$, we obtain
 \begin{align*}
 \overline{C}^{\tau}_\cl\leq C_\cl-\frac{1}{2}\log\left(1-A\right).
 \end{align*}
 
 Following the same argument as above, we also have for any $\gamma>0$, there exists $B>0$ such that
 \begin{align*}
 \underline{C}^{\gamma}_\cul\geq C_\cul+\frac{1}{2}\log(1-B)
 \end{align*}
where $C_\cul$ is the optimal value of the optimization~(\ref{eq:3.2.7}) with block-length $\cul$.

 \subsection{Proof of Theorem~\ref{thm:4}}
 \label{app:2}
 The lower bound can be immediately proved by noticing 
 \begin{align}
 \label{eq:a.22}
 &{C}_{\cl}=\adjustlimits
 \sup_{\spa\in\spas}\min_{1\leq\dpr\leq\cl}\inf_{\npa\in\npaseu}
 \frac{1}{2\cl}\sum_{\indt=1}^{\dpr}\log\frac{\spal_{\indt}}{\npal_{\indt}}.
 \end{align}

For the upper bound, we first prove the following claim:
\begin{claim}
	\label{claim:1}
For any fixed $\spa=\spal_1,\ldots,\spal_{\cl}\in\spas$ and $1\leq\dpr\leq\cl$, there exists a two-level sequence $\spa^*(\spa,\dpr)=\spal^*_1,\ldots,\spal^*_{\cl}\in\uspas$ such that
 \begin{align}
 \label{eq:a.1}
 \inf_{\npa\in{\npas}_0\left(\dpr,\spa^*(\spa,\dpr)\right)}\frac{1}{2\cl}\sum_{\indt=1}^{\dpr}\log\frac{\spal^*_{\indt}}{\npal_{\indt}}\geq	\inf_{\npa\in{\npas}_0\left(\dpr,\spa\right)}\frac{1}{2\cl}\sum_{\indt=1}^{\dpr}\log\frac{\spal_{\indt}}{\npal_{\indt}}.
 \end{align}
\end{claim}
 
 \begin{proof}[Proof of Claim~\ref{claim:1}]
 For notational simplicity, denote two quantities\footnote{Intuitively, $B$ is equal to the total ``energy budget'' allocated for the prefix $\Cwdp$ of length $\dpr$ and $G$ is equal to the ``energy gap'' between Alice and James.} $B$ and $G$ as
 \begin{align*}
 &B\triangleq\sum_{\indt=1}^{\dpr}\frac{1}{2}\spal_{\indt},\\
 &G\triangleq\cl \left({\frac{1}{2}\spal}- \npal\right).
 \end{align*}
 
 For each $\spa=\spal_1,\ldots,\spal_{\cl}\in\spas$ and $1\leq\dpr\leq\cl$, define a corresponding $\spa^*\in\uspas$ to be
 \begin{align*}
 \spal^*_\indt=\begin{cases}
 \frac{2B}{\dpr} \quad &\text{if } 1\leq\indt\leq\dpr\\
 \frac{\cl \spal-2B}{\cl-\dpr} \quad &\text{if }\dpr< \indt\leq\cl
 \end{cases}.
 \end{align*}
 
 We claim that the corresponding optimizing $\widehat{\npa}\in\npas\left(\dpr,\spa\right)$ and $\widetilde{\npa}\in\npas\left(\dpr,\spa^*\right)$ for the LHS and RHS of (\ref{eq:a.1}) respectively are given by the following \textit{water-filling} type of solutions:
 \begin{align*}
 &\widehat{\npal}_{\indt}=\begin{cases}
 \min\left\{\spal_{\indt}, \alpha  \right\},  &\qquad \indt=1,\ldots,\dpr\\
 \frac{\cl \npal-B+G}{\cl-\dpr},  &\qquad \indt=\dpr+1,\ldots,\cl
 \end{cases}\\
 &\widetilde{\npal}_{\indt}=\begin{cases}\frac{B-G}{\mu},  &\ \quad \qquad \indt=1,\ldots,\dpr\\
 \frac{\cl \npal-B+G}{\cl-\dpr},  & \ \quad \qquad \indt=\dpr+1,\ldots,\cl
 \end{cases}
 \end{align*}
 where $ \alpha $ satisfies
 \begin{align*}
 \sum_{\indt=1}^{\dpr}\min\left\{\spal_{\indt}, \alpha \right\}={B-G}.
 \end{align*}
 
 
The following lemma indicates that the optimal sequence $\npa$ is a water-filling solution.
 \begin{lemma}
 	\label{lemma:10}
 	Consider the following optimization problem with a fixed integer $\dpr>0$, sequence $\spal_1,\ldots,\spal_\dpr$ of non-negative coordinates and variables $\npal_1,\ldots,\npal_\dpr$:
 	\begin{equation}
 	\label{eq:a.6}
 	\begin{aligned}
 	&\underset{\npal_1,\ldots,\npal_\dpr}{\text{minimize}}  & &\sum_{\indt=1}^{\dpr}\log\frac{1}{\npal_{\indt}}\\
 	&\text{subject to } & &\sum\limits_{\indt=1}^{\dpr}\npal_\indt=B-G,\\
 	& & & \npal_{\indt}\geq 0, \qquad  \indt=1,\ldots,\dpr,\\
 	& & & \npal_{\indt}\leq \spal_{\indt}, \quad \ \ \indt=1,\ldots,\dpr.
 	\end{aligned}
 	\end{equation}
 	The optimizing solution $\npal^*_1,\ldots,\npal^*_{\dpr}$ of the problem above is given by
 	\begin{align*}
 	&{\npal}^*_{\indt}=
 	\min\left\{\spal_{\indt}, \alpha  \right\},  \qquad \indt=1,\ldots,\dpr
 	\end{align*}
 	where $ \alpha >0$ satisfies
 	\begin{align*}
 	\sum_{\indt=1}^{\dpr}\min\left\{\spal_{\indt}, \alpha \right\}={B-G}.
 	\end{align*}
 \end{lemma}
 
 \begin{proof}[Proof of Lemma~\ref{lemma:10}]
 	By introducing KKT multipliers $\beta$, $\lambda_1,\ldots,\lambda_{\dpr}$ and $\nu_1,\ldots,\nu_{\dpr}$, we obtain the following KKT conditions (\textit{c.f.},~\cite{boyd2004convex}) that \textit{necessarily} hold for an \textit{optimal} solution $\npal^*_1,\ldots,\npal^*_{\dpr}$ for all $\indt=1,\ldots,\dpr$:
 	\begin{equation*}
 	\begin{aligned}
 	&\text{(Stationarity)}: \qquad\qquad\frac{1}{\npal_\indt}=-\lambda_\indt+\nu_\indt+\beta, &\\
 	&\text{(Primal Feasibility)}:
 	\qquad\sum\limits_{\indt=1}^{\dpr}\npal_\indt=B-G,\\
 	&\qquad\qquad\qquad\qquad\qquad\qquad\qquad\npal_{\indt}\geq 0,  &\\
 	&\qquad\qquad\qquad\qquad\qquad\qquad\qquad\npal_{\indt}\leq \spal_{\indt}, &\\
 	&\text{(Dual Feasibility)}:\qquad\qquad\qquad
 	\lambda_\indt\geq 0, &\\
 	&\qquad\qquad\qquad\qquad\qquad\qquad\qquad\nu_\indt\geq 0, &\\
 	&\text{(Complementary slackness)}:
 	\quad\lambda_\indt\npal_\indt=0, \qquad&\\
 	&\qquad\qquad\qquad\qquad\qquad\qquad \nu_\indt\left(\npal_\indt-\spal_\indt\right)=0.&
 	\end{aligned}
 	\end{equation*}
 	First of all, notice that $\npal_\indt\neq 0$. Therefore $\lambda_\indt=0$ for all $\indt=1,\ldots,\dpr$. The above system of inequalities can be reduced to the following:
 	\begin{align}
 	\label{eq:a.8}
 	\quad\sum\limits_{\indt=1}^{\dpr}\npal_\indt&=B-G,\\
 	\text{ for all }\indt=1,\ldots,\dpr,
 	\nonumber
 	\quad{\npal_\indt}&=\frac{1}{\nu_\indt+\beta}, \\
 	\nonumber
 	\quad \nu_\indt\left(\npal_\indt-\spal_\indt\right)&=0,\\
 	\nonumber
 	\quad\npal_{\indt}&> 0,\\
 	\quad\npal_{\indt}&\leq \spal_{\indt},\\
 	\nonumber
 	\quad \nu_\indt&\geq 0.
 	\end{align}
 	
 	Next, for all $\indt=1,\ldots,\dpr$, we consider the two cases---$\nu_\indt\neq 0$ and $\nu_\indt= 0$. In the first situation, we must have $\npal_\indt=\spal_\indt=\frac{1}{\nu_\indt+\beta}$. In the second case, we obtain $\npal_\indt=\frac{1}{\beta}$ for some multiplier $\beta>0$. Taking the first condition in (\ref{eq:a.8}) into account, we conclude that the optimal solution $\npal^*_1,\ldots,\npal^*_{\dpr}$ should \textit{necessarily} follow
 	\begin{align}
 	\label{eq:a.9}
 	&{\npal}^*_{\indt}=\min\left\{\spal_{\indt}, \frac{1}{\beta} \right\}=
 	\min\left\{\spal_{\indt}, \alpha \right\},  \qquad \indt=1,\ldots,\dpr
 	\end{align}
 	with some $ \alpha >0$ satisfying
 	\begin{align}
 	\label{eq:a.10}
 	\sum_{\indt=1}^{\dpr}\min\left\{\spal_{\indt}, \alpha \right\}={B-G}.
 	\end{align}
 	
 	The equations (\ref{eq:a.9}) and (\ref{eq:a.10}) above uniquely determine the optimal solution $\npal^*_1,\ldots,\npal^*_{\dpr}$. Therefore the lemma is proved.
 \end{proof}
 According to the  definitions of $\npaseu$ and ${\npas}_0\left(\dpr,\spa^*\right)$, sequences $\widehat{\npa}$ and $\widetilde{\npa}$ should satisfy
 \begin{align*}
 &\sum\limits_{\indt=\dpr+1}^{\cl}\spal _\indt\geq \sum\limits_{\indt=\dpr+1}^{\cl}2\widehat{\npal}_\indt,\\
 &\sum\limits_{\indt=\dpr+1}^{\cl}\spal _\indt\geq \sum\limits_{\indt=\dpr+1}^{\cl}2\widetilde{\npal}_\indt
 \end{align*}
 yielding
 \begin{align}
 \label{eq:a.5}
 &\sum\limits_{\indt=1}^{\dpr}\widehat{\npal}_\indt\leq B-G,\\
 \label{eq:a.7}
 &\sum\limits_{\indt=1}^{\dpr}\widetilde{\npal}_\indt\leq B-G.
 \end{align}
 Now consider the LHS and RHS of (\ref{eq:a.1}). Since the goal is to find some $\widehat{\npa}\in\npaseu$ and $\widetilde{\npa}\in\npas_0\left(\dpr,\spa^*\right)$ minimizing the objective values, the optimal $\widehat{\npa}$ and $\widetilde{\npa}$ are therefore those that always make the inequalities (\ref{eq:a.5}) and (\ref{eq:a.7}) equalities.  
 In other words, with fixed $\spa=\spal_1,\ldots,\spal_{\cl}\in\spas$ and $1\leq\dpr\leq\cl$, we always have
 \begin{align}
 &\sum\limits_{\indt=1}^{\dpr}\widehat{\npal}_\indt= B-G,\\
 &\sum\limits_{\indt=1}^{\dpr}\widetilde{\npal}_\indt= B-G
 \end{align}
 such that $\widehat{\npa}$ and $\widetilde{\npa}$ have the maximal total ``budget'' $B-G$.

 The last step is to show that~(\ref{eq:a.1}) holds with $\spa^*$, $\widehat{\npa}$ and $\widetilde{\npa}$ defined above. Denote $\indd\triangleq\left\{T=1,\ldots,\dpr: \widehat{\npal}_{\indt}=\alpha  \right\}$ and let $B_1\triangleq\sum_{\indt\in\left\{1,\ldots,\dpr\right\}\backslash\indd}\widehat{
 	\spal}_{\indt}$. Then for any $\indt\in\indd$, $\widehat{\npal}_{\indt}=\alpha $ where
 \begin{align}
 \label{eq:a.12}
 \alpha =\frac{B-B_1-G}{\left|\indd\right|}.
 \end{align}
 
 Substituting $\widehat{\npa}$ and $\widetilde{\npa}$ back into (\ref{eq:a.1}), we obtain
 \begin{align}
 \nonumber
 &\inf_{\npa\in\npaseu}\frac{1}{2\cl}\sum_{\indt=1}^{\dpr}\log\frac{\spal _\indt}{\npal_\indt}
 \\=&\frac{1}{2\cl}\sum_{\indt=1}^{\dpr}\log\frac{\spal _\indt}{\widehat{\npal}_\indt}\\
 \nonumber
 =&\frac{1}{2\cl}\sum_{\indt\in\indd}\log\frac{\spal _\indt}{\alpha }+\frac{1}{2\cl}\sum_{\indt\in\left\{1,\ldots,\dpr\right\}\backslash\indd}\log\frac{\spal _\indt}{\spal _\indt}\\
 \nonumber
  =&\frac{1}{2\cl}\sum_{\indt\in\indd}\log\frac{\spal _\indt}{\alpha }.
 \end{align}
 
 Since the maximal value of $\sum_{\indt\in\indd}\log\frac{\spal _\indt}{\alpha }$ is attained by setting each $\spal_\indt$ with $\indt\in\indd$ to be the same value $\frac{2B-B_1}{\left|\indd\right|}$ (recall $\sum_{\indt\in\indd}\spal_\indt=2B-B_1$), we get
  \begin{align}
\nonumber
 \frac{1}{2\cl}\sum_{\indt\in\indd}\log\frac{\spal _\indt}{\alpha }\leq& \frac{\left|\indd\right|}{2\cl}\log\frac{2B-B_1}{B-B_1-G}\\
 \label{eq:a.3}
 =&\frac{\left|\indd\right|}{2\cl}\log\left(1+\frac{B+G}{B-B_1-G}\right).
 \end{align}
Substituting (\ref{eq:a.12}) into $\sum_{\indt\in\indd}\log\frac{\spal _\indt}{\alpha }$ gives the inequality. Furthermore, we also obtain
 \begin{align}
 \nonumber
 \inf_{\npa\in{\npas}_0\left(\dpr,\spa^*\right)}\frac{1}{2\cl}\sum_{\indt=1}^{\dpr}\log\frac{\spal^* _\indt}{\npal_\indt}
 =&\frac{\dpr}{2\cl}\log\frac{2B}{B-G}\\
 \label{eq:a.2}
 =&\frac{\dpr}{2\cl}\log\left(1+\frac{B+G}{B-G}\right).
 \end{align}
 
 Note that 
 \begin{align}
 \label{eq:a.4}
 B_1=\sum_{\indt\in\left\{
 	1,\ldots,\dpr\right\}\backslash\indd}\spal_\indt\leq \left(\dpr-\left|\indd\right|\right)\frac{B-B_1-G}{\left|\indd\right|}
 \end{align}
 since $\spal_\indt\leq\alpha =\frac{B-B_1-G}{\left|\indd\right|}$ for every $\indt\notin\indd$.
 
 Moreover, (\ref{eq:a.3}) and (\ref{eq:a.2}) give
 \begin{align*}
 &\inf_{\npa\in{\npas}_0\left(\dpr,\spa^*\right)}\frac{1}{2\cl}\sum_{\indt=1}^{\dpr}\log\frac{\spal^* _\indt}{\npal_\indt}-
 \min_{\npa\in\npaseu}\frac{1}{2\cl}\sum_{\indt=1}^{\dpr}\log\frac{\spal _\indt}{\npal_\indt}\\
 \geq&\frac{\dpr}{2\cl}\log\left(1+\frac{B+G}{B-G}\right)-\frac{\left|\indd\right|}{2\cl}\log\left(1+\frac{B+G}{B-B_1-G}\right)\\
 =&\frac{\dpr}{2\cl}\log\left(1+\frac{B+G}{\left(B-B_1-G\right)+B_1}\right)\\
 &\quad-\frac{\left|\indd\right|}{2\cl}\log\left(1+\frac{B+G}{B-B_1-G}\right).
 \end{align*}
 
Rewriting above we obtain
\begin{align*}
&\frac{\dpr}{2\cl}\log\left(1+\frac{B+G}{\left(B-B_1-G\right)\left(1+\frac{\dpr-\left|\indd\right|}{\left|\indd\right|}\right)}\right)\\
&\quad-\frac{\left|\indd\right|}{2\cl}\log\left(1+\frac{B+G}{B-B_1-G}\right)\\
=&\frac{\dpr}{2\cl}\Bigg(\log\left(1+\frac{B+G}{\left(B-B_1-G\right)\left(1+\frac{\dpr-\left|\indd\right|}{\left|\indd\right|}\right)}\right).
\end{align*}
 
 Then applying (\ref{eq:a.4}), above can be bounded from below by
 \begin{align*} 
 =&\frac{\dpr}{2\cl}\Bigg(\log\left(1+\frac{B+G}{\left(B-B_1-G\right)\left(1+\frac{\dpr-\left|\indd\right|}{\left|\indd\right|}\right)}\right)\\
 &\quad-\frac{\left|\indd\right|}{\dpr}\log\left(1+\frac{B+G}{B-B_1-G}\right)\Bigg)\\
 =&\frac{\dpr}{2\cl}\Bigg(\log\left(1+\frac{B+G}{\left(B-B_1-G\right)\frac{\dpr}{\left|\indd\right|}}\right)\\
 &\quad-\frac{\left|\indd\right|}{\dpr}\log\left(1+\frac{B+G}{B-B_1-G}\right)\Bigg).
 \end{align*}
 
 Let $\psi\triangleq\frac{B+G}{B-B_1-G}$ and $\alpha\triangleq\frac{\left|\indd\right|}{\dpr}$. It suffices to show that for any $\psi>0$ and $0\leq \alpha\leq 1$,
 \begin{align*}
 \log\left(1+\alpha\psi\right)-\alpha\log\left(1+\psi\right)\geq 0,
 \end{align*}
 which is true since $\left(1+\psi\right)^\alpha\leq 1+\alpha\psi$ for all $\psi\geq -1$ and $0\leq \alpha\leq 1$. 
 Therefore we complete the proof of the claim.
\end{proof}

Now, the upper bound on $C_\cl$ can be derived by noting
 \begin{align}
 \nonumber
 {C}_{\cl}=&\adjustlimits
 \sup_{\spa\in\spas}\min_{1\leq\dpr\leq\cl}\inf_{\npa\in\npaseu}
 \frac{1}{2\cl}\sum_{T=1}^{\dpr}\log\frac{\spal_{\indt}}{\npal_{\indt}}\\
 \label{app:21}
 \leq&  \adjustlimits
 \sup_{\spa\in\spas}\min_{\dpr\in\Gamma(\spa)}\inf_{\npa\in\npaseu}
 \frac{1}{2\cl}\sum_{T=1}^{\dpr}\log\frac{\spal_{\indt}}{\npal_{\indt}}
 \end{align}
 where $\Gamma(\spa)\subseteq\left\{1,\ldots,\cl\right\}$ is defined as
 \begin{align*}
 \Gamma(\spa)=\begin{cases}
  \left\{\nu\right\} \quad &\text{ for all two-level }  \spa\in\uspas\\
 \left\{1,\ldots,\cl\right\} \quad &\text{ otherwise } 
 \end{cases}.
 \end{align*}
 
 According to the claim, we know (\ref{app:21}) equals to
 \begin{align*}
  \adjustlimits
  \sup_{\spa\in\uspas}\inf_{\npa\in\mathcal{N}\left(\nu,\spa\right)}
  \frac{1}{2\cl}\sum_{T=1}^{\nu}\log\frac{\spal_{\indt}}{\npal_{\indt}}
 \end{align*}
 since for any $\spa\in\spas\backslash\uspas$ that is not a two-level sequence with a fixed division point $\dpr$, we can always construct a two-level sequence $\spa^*$ such that the corresponding objective value is not smaller than the previous one. This implies that it suffices to only consider the subset $\uspas$. Thus
 \begin{align*}
 C_\cl\leq \sup_{\spa\in\uspas}\inf_{\npa\in\mathcal{N}\left(\nu,\spa\right)}
 \frac{1}{2\cl}\sum_{T=1}^{\nu}\log\frac{\spal_{\indt}}{\npal_{\indt}}
 \end{align*}
 follows and the proof is complete.

 \subsection{Proof of Theorem~\ref{thm:5}}
 \label{app:6}
 
The upper bound follows by weakening the adversary James such that the division point $\dpr$ should be selected before knowing the transmitted codeword $\cwd$. Given a fixed average power allocation sequence for Alice, there are two cases -- the selected $\dpr$ is a feasible division point such that there exists some power allocation sequence for James satisfying the energy-bounding condition in~(\ref{eq:energy-bounding}); or the selected $\dpr$ violates the condition for all possible power allocations of James. 
 	
James' strategy is as follows. If $\dpr$ is feasible, James chooses the corresponding feasible power allocation sequence and attacks using scaled-babble and push. This case corresponds to when $\overline{\spal}\leq 2\overline{\npal}$. Otherwise, he simply performs a two-stage babble attack as follows. First, he divides the coordinates into $\{1,\ldots,\dpr\}$ and $\{\dpr+1,\ldots,\cl\}$. Then, he adds random Gaussian noise as the first stage in (\ref{eq:3.0}) according to a two-level power allocation in $\unpas$.  This is equivalent to setting $\dpr=\cl$ and chooses a two-level noise power allocation $\npa$. Since in the second case, it always holds that $\overline{\spal}> 2\overline{\npal}$ and the selection subjects to $\underline{\spal}\geq \underline{\npal}$, the selected two-level power allocation is always feasible. 

Formally, the theorem can be proved by noticing that $C_\cl$ can be written as 
 \begin{align*}
&{C}_{\cl}=\adjustlimits
\sup_{\spa\in\spas}\min_{1\leq\dpr\leq\cl}\inf_{\npa\in\npaseu}
\frac{1}{2\cl}\sum_{\indt=1}^{\dpr}\log\frac{\spal_{\indt}}{\npal_{\indt}}.
\end{align*}
as in (\ref{eq:a.22})
 
Switching the supremum with the minimization does not decrease the optimal value. Moreover, for fixed $\dpr$ and $\spa$, if $\npaseu=\emptyset$, setting the objective value to be the one corresponding to when $\dpr=\cl$ and $\npa\in\unpas$ subjecting to $\underline{\spal}\geq \underline{\npal}$ does not alter the optimal value. This holds for the reason that setting $\dpr=\cl$ and $\npa\in\unpas$ subjecting to $\underline{\spal}\geq \underline{\npal}$ always gives a feasible solution. Therefore, an upper bound is obtained.

 \subsection{Proof of Lemma~\ref{lemma:3}}
 \label{app:3}
 \section{Lower Bound on $\mathbbm{P}_1$}
 Denote ${\Zt}_{\leq\optdpr}=\Zlt_1,\ldots,\Zlt_{\optdpr}$.
 For arbitrary constants $\alpha>0$ and $\beta>0$, Inequality~(\ref{eq:a.13}) holds.
 
 \begin{figure*}[b]
 	 \hrule
 \begin{align}
 \nonumber
 \mathbbm{P}_1&=\Pr\left(\Msg\neq\Usg, \  \left|\left|\optStp\right|\right|^2\leq \sum_{\indt=1}^{\optdpr}\npal_\indt \right)\\
  \nonumber
 &\geq \Pr\Bigg(\Msg\neq\Usg, \  \left|\left|\optStp\right|\right|^2\leq \sum_{\indt=1}^{\optdpr}\npal_\indt, \ \sum_{\indt=1}^{\optdpr}\frac{2\npal^*_{\indt}}{\spal_{\indt}}\Clwd_\indt\Zlt_{\indt}\geq-\alpha\cul, \ \sum_{\indt=1}^{\optdpr}\left(\frac{\npal^*_{\indt}}{\spal_{\indt}}\Clwd_\indt\right)^2\leq\sum_{\indt=1}^{\optdpr}\left(1+\beta\right)\frac{\left(\npal^*_{\indt}\right)^2}{\spal_{\indt}}  \Bigg)\\
  \nonumber
 &\geq \Pr\Bigg(\Msg\neq\Usg, \  \left|\left|\optStp\right|\right|^2\leq \sum_{\indt=1}^{\optdpr}\npal_\indt\Bigg|\ \sum_{\indt=1}^{\optdpr}\frac{2\npal^*_{\indt}}{\spal_{\indt}}\Clwd_\indt\Zlt_{\indt}\geq-\alpha\cul, \ \sum_{\indt=1}^{\optdpr}\left(\frac{\npal^*_{\indt}}{\spal_{\indt}}\Clwd_\indt\right)^2\leq\sum_{\indt=1}^{\optdpr}\left(1+\beta\right)\frac{\left(\npal^*_{\indt}\right)^2}{\spal_{\indt}} \Bigg)\\
  \label{eq:a.13}
 &\quad\cdot\Pr\left(\sum_{\indt=1}^{\optdpr}\frac{2\npal^*_{\indt}}{\spal_{\indt}}\Clwd_\indt\Zlt_{\indt}\geq-\alpha\cul, \ \sum_{\indt=1}^{\optdpr}\left(\frac{\npal^*_{\indt}}{\spal_{\indt}}\Clwd_\indt\right)^2\leq\sum_{\indt=1}^{\optdpr}\left(1+\beta\right)\frac{\left(\npal^*_{\indt}\right)^2}{\spal_{\indt}}\right).
 \end{align}
 \end{figure*}
 
 Recall in (\ref{eq:3.0}), for all $1\leq\indt\leq\optdpr$,
 \begin{align*}
 \Slt_\indt=\Zlt_\indt-\frac{\npal^*_{\indt}}{\spal_{\indt}}\Clwd_\indt.
 \end{align*}
 
 Therefore, the conditions 
 \begin{align*}
 &\left|\left|\optStp\right|\right|^2\leq \sum_{\indt=1}^{\optdpr}\npal_\indt,\\
 &\sum_{\indt=1}^{\optdpr}\frac{2\npal^*_{\indt}}{\spal_{\indt}}\Clwd_\indt\Zlt_{\indt}\geq-\alpha\cul, \\ &\sum_{\indt=1}^{\optdpr}\left(\frac{\npal^*_{\indt}}{\spal_{\indt}}\Clwd_\indt\right)^2\leq\sum_{\indt=1}^{\optdpr}\left(1+\beta\right)\frac{\left(\npal^*_{\indt}\right)^2}{\spal_{\indt}}
 \end{align*}
 in above can be sufficiently satisfied given
 \begin{align*}
 \left|\left|{\Zt}_{\leq\optdpr}\right|\right|^2\leq \sum_{\indt=1}^{\optdpr}\npal_\indt-\alpha\cul-\left(1+\beta\right)\sum_{\indt=1}^{\optdpr}\frac{\left(\npal^*_{\indt}\right)^2}{\spal_{\indt}}\triangleq\overline{N}_{\alpha,\beta,\tau}
 \end{align*}
 where $\overline{N}_{\alpha,\beta,\tau}>0$ with $\alpha>0$ and $\beta>0$ small enough since $\npal_\indt^*\leq\spal_\indt$ for every $1\leq\indt\leq\optdpr$ and we have $\sum_{\indt=\optdpr}^{\cl}2\npal^*_\indt\geq \sum_{\indt=\optdpr}^{\cl}\spal_\indt$ and $2\npal>\spal$.

 Hence,
 \begin{align*}
 \mathbbm{P}_1&\geq
 \Pr\left(\Msg\neq\Usg, \  \left|\left|{\Zt}_{\leq\optdpr}\right|\right|^2\leq \overline{N}_{\alpha,\beta,\tau}\right)\\
 &\quad\cdot\Pr\Bigg(\sum_{\indt=1}^{\optdpr}\frac{2\npal^*_{\indt}}{\spal_{\indt}}\Clwd_\indt\Zlt_{\indt}\geq-\alpha\cul, \\
 &\qquad \sum_{\indt=1}^{\optdpr}\left(\frac{\npal^*_{\indt}}{\spal_{\indt}}\Clwd_\indt\right)^2\leq\sum_{\indt=1}^{\optdpr}\frac{\left(1+\beta\right)\left(\npal^*_{\indt}\right)^2}{\spal_{\indt}}\Bigg)
 \end{align*}
 where $ \widetilde{{\npal}}_{\indt}^{\varepsilon}=\frac{1}{1+\varepsilon}\npal^*_{\indt}\left(1-\frac{\npal^*_{\indt}}{\spal_{\indt}}\right)$ for all $1\leq\indt\leq\optdpr$.
 
 Denote $\mathbbm{P}_{1,1}$, $\mathbbm{P}_{1,1}$, $\mathbbm{P}_{1,1}$ and $\mathbbm{P}_{1,1}$ as (\ref{eq:a.14})-(\ref{eq:a.15}).
  \begin{figure*}[b]
 	\hrule
 \begin{align}
 \label{eq:a.14}
 &\mathbbm{P}_{1,1}\triangleq\Pr\left(\Msg\neq\Usg\Bigg| \  \left|\left|{\Zt}_{\leq\optdpr}\right|\right|^2\leq \overline{N}_{\alpha,\beta,\tau}\right),\\
 &\mathbbm{P}_{1,2}\triangleq\Pr\left(\left|\left|{\Zt}_{\leq\optdpr}\right|\right|^2\leq \overline{N}_{\alpha,\beta,\tau}\right),\\
  \label{eq:a.141}
 &\mathbbm{P}_{1,3}\triangleq\Pr\left(\sum_{\indt=1}^{\optdpr}\frac{2\npal^*_{\indt}}{\spal_{\indt}}\Clwd_\indt\Zlt_{\indt}\geq-\alpha\cul\Bigg| \ \sum_{\indt=1}^{\optdpr}\left(\frac{\npal^*_{\indt}}{\spal_{\indt}}\Clwd_\indt\right)^2\leq\sum_{\indt=1}^{\optdpr}\left(1+\beta\right)\frac{\left(\npal^*_{\indt}\right)^2}{\spal_{\indt}}\right),\\
  \label{eq:a.15}
 &\mathbbm{P}_{1,4}\triangleq\Pr\left(\sum_{\indt=1}^{\optdpr}\left(\frac{\npal^*_{\indt}}{\spal_{\indt}}\Clwd_\indt\right)^2\leq\sum_{\indt=1}^{\optdpr}\left(1+\beta\right)\frac{\left(\npal^*_{\indt}\right)^2}{\spal_{\indt}}\right)
 \end{align}
   \end{figure*}
 and we have
 \begin{align}
 \label{eq:3.17}
 \mathbbm{P}_{1}\geq\mathbbm{P}_{1,1}\mathbbm{P}_{1,2}\mathbbm{P}_{1,3}\mathbbm{P}_{1,4}.
 \end{align}
 
 Next we bound $\mathbbm{P}_{1,1}$, $\mathbbm{P}_{1,2}$, $\mathbbm{P}_{1,3}$ and $\mathbbm{P}_{1,4}$ respectively.
 
 Before bounding $\mathbbm{P}_{1,1}$,  we present a useful lemma below.
 \begin{lemma}
 	\label{lemma:1}
 	Let $\Msg$ be a random variable in $\msgs$ with entropy $\ent\left(\Msg\right)\geq {\cl\varepsilon}$ and $\Usg$ be an i.i.d. copy of $\Msg$ with the same probability distribution. It follows that
 	\begin{align}
 	\label{eq:3.1}
 	\Pr_{\Msg,\Usg}\left(\Msg\neq\Usg \right)\geq \frac{\cl\varepsilon-1}{\log\mn}.
 	\end{align}
 \end{lemma}
 
 \begin{proof}
 	Fix a message $\msg\in\msgs$. Let 
 	\begin{align*}
 	E_{\msg}&=\mathds{1}\left(\Usg= \msg\right)=\begin{cases*}
 	1\quad \text{if } \Usg= \msg\\
 	0\quad \text{if } \Usg\neq \msg
 	\end{cases*}
 	\end{align*} indicate whether $\Usg$ equals to $\msg$. 
 	The probability that $\Msg$ and $\Usg$ are distinct can be decomposed as
 	\begin{align}
 	\nonumber
 	&\Pr\left(\Msg\neq\Usg \right)\\
 	\nonumber
 	=&
 	\sum_{\msg\in\msgs}\Pr_{\Msg}\left(\Msg=\msg\right)\Pr_{\Usg}\left(\Usg\neq\msg\right)\\
 	\label{eq:3.2}
 	=&
 	\sum_{\msg\in\msgs}\Pr_{\Msg}\left(\Msg=\msg\right)\Pr_{\Usg}\left(E_{\msg}=0\right).
 	\end{align}
 	Note that $\ent\left(\Usg\right)=\ent\left(\Msg\right)\geq {\cl\varepsilon}$. For any $\msg\in\msgs$, it follows that
 	\begin{align*}
 	&{\cl\varepsilon}\\
 	\leq&\ent\left(\Usg\right)\\
 	\leq& \ent\left(\Usg|E_\msg\right)+\ent\left(E_\msg\right)\\
 	=&\ent\left(\Usg|E_\msg=0\right)\Pr\left(E_\msg=0\right)\\
 	&\quad+\ent\left(\Usg|E_\msg=1\right)\Pr\left(E_\msg=1\right)+\ent\left(E_\msg\right).
 	\end{align*}
 	Since $\Usg\in\msgs$ and if $E_\msg=1$, then $\Usg=\msg$ and $\ent\left(\Usg|E_\msg=1\right)=0$, we get
 	\begin{align*}
 	{\cl\varepsilon}&\leq \log\mn\Pr\left(E_\msg=0\right)+1
 	\end{align*}
 	yielding $\Pr\left(E_\msg=0\right)\geq\frac{{\cl\varepsilon}-1}{\log\mn}$ for all $\msg\in\msgs$. Putting this into (\ref{eq:3.2}) gives (\ref{eq:3.1}).
 \end{proof}

 Now conditioned on $\left|\left|{\Zt}_{\leq\optdpr}\right|\right|^2\leq \overline{N}_{\alpha,\beta,\tau}$, we give a bound on the conditional entropy $ 	\ent\left(\Msg|\optRwdp\right)$.
 Denote by $\overline{\Zlt}$ an indicator random variable such that
 \begin{align*}
 \overline{\Zlt}=\begin{cases*}
 1 \quad \text{if} \ \left|\left|{\Zt}_{\leq\optdpr}\right|\right|^2\leq \overline{N}_{\alpha,\beta,\tau}\\
 0 \quad \text{if} \ \left|\left|{\Zt}_{\leq\optdpr}\right|\right|^2> \overline{N}_{\alpha,\beta,\tau}
 \end{cases*}.
 \end{align*}
 
 First, we notice that the Markov chain $\Msg\rightarrow\optCwdp\rightarrow\optRwdp$. By data processing inequality,
 \begin{align}
 \label{eq:3.21}
 \mut\left(\Msg;\optRwdp| \overline{\Zlt}=1\right)\leq\mut\left(\optCwdp;\optRwdp|\overline{\Zlt}=1\right).
 \end{align}
 Next by chain rule, recalling $\Rlwd_\indt=\Clwd_\indt+\Slt_\indt=\Zlt_\indt+\left(1-\frac{\npal_{\indt}}{\spal_{\indt}}\right)\Clwd_\indt$ for all $1\leq\indt\leq\optdpr$, the mutual information $\mut\left(\optCwdp;\optRwdp|\overline{\Zlt}=1\right)$ can be bounded as
 \begin{align}
 \nonumber
 &\mut\left(\optCwdp;\optRwdp|\overline{\Zlt}=1\right)\\
 =
 &\ent\left(\optRwdp|\overline{\Zlt}=1\right)-\ent\left(\optRwdp|\optCwdp,\overline{\Zlt}=1\right)\\
 \label{eq:3.22}
 &\leq\sum_{\indt=1}^{\optdpr}\ent\left(\Rlwd_\indt|\overline{\Zlt}=1\right)-\ent\left({\Zt}_{\leq\optdpr}|\overline{\Zlt}=1\right)
 \end{align}
 where ${\Zt}_{\leq\optdpr}=\Zlt_1,\ldots,\Zlt_{\optdpr}$  and each $\Zlt_\indt$ is an independent Gaussian random variable with zero mean and variance $ \widetilde{{\npal}}_{\indt}^{\varepsilon}$. Therefore by definition we have
 \begin{align}
 \nonumber
 &\ent\left({\Zt}_{\leq\optdpr}|\overline{\Zlt}=1\right)
 \\=&\int_{\left|\left|{\Zt}_{\leq\optdpr}\right|\right|^2\leq \overline{N}_{\alpha,\beta,\tau} }\frac{1}{\sqrt{\left(2\pi\right)^{\optdpr}\prod_{\indt=1}^{\optdpr}\widetilde{{\npal}}_{\indt}^{\varepsilon}}}e^{-\sum_{\indt=1}^{\optdpr}\frac{\Zlt_\indt^2}{2\widetilde{{\npal}}_{\indt}^{\varepsilon}}}\\
 \nonumber
 &\qquad\cdot\left(\frac{1}{2}\sum_{\indt=1}^{\optdpr}\log\left(2\pi\widetilde{{\npal}}_{\indt}^{\varepsilon}\right)+\sum_{\indt=1}^{\optdpr}\left(\log e\right)\frac{\Zlt_\indt^2}{2\widetilde{{\npal}}_{\indt}^{\varepsilon}}\right)
 \mathrm{d}{\Zt}_{\leq\optdpr},
  \end{align}
 which equals to
  \begin{align}
 \nonumber
 &\quad \ \ent\left({\Zt}_{\leq\optdpr}\right)\\
 &-\int_{\left|\left|{\Zt}_{\leq\optdpr}\right|\right|^2> \overline{N}_{\alpha,\beta,\tau} }\frac{1}{\sqrt{\left(2\pi\right)^{\optdpr}\prod_{\indt=1}^{\optdpr}\widetilde{{\npal}}_{\indt}^{\varepsilon}}}e^{-\sum_{\indt=1}^{\optdpr}\frac{\Zlt_\indt^2}{2\widetilde{{\npal}}_{\indt}^{\varepsilon}}}\\
 \label{eq:3.6}
 &\cdot\left(\frac{1}{2}\sum_{\indt=1}^{\optdpr}\log\left(2\pi\widetilde{{\npal}}_{\indt}^{\varepsilon}\right)+\sum_{\indt=1}^{\optdpr}\left(\log e\right)\frac{\Zlt_\indt^2}{2\widetilde{{\npal}}_{\indt}^{\varepsilon}}\right)\mathrm{d}{\Zt}_{\leq\optdpr}.
 \end{align}
 
 Next, we argue that the conditional entropy $\ent\left({\Zt}_{\leq\optdpr}|\overline{\Zlt}=1\right)$ is not too much different from the entropy $\ent\left({\Zt}_{\leq\optdpr}\right)$.

 For large enough $\optdpr$,
 \begin{align*}
 \frac{1}{2}\sum_{\indt=1}^{\optdpr}\log\left(2\pi\widetilde{{\npal}}_{\indt}^{\varepsilon}\right)+\sum_{\indt=1}^{\optdpr}\left(\log e\right)\frac{\Zlt_\indt^2}{2\widetilde{{\npal}}_{\indt}^{\varepsilon}}\leq e^{-\frac{1}{2}\sum_{\indt=1}^{\optdpr}\frac{\Zlt_\indt^2}{2\widetilde{{\npal}}_{\indt}^{\varepsilon}}},
 \end{align*}
 we denote $\overline{{\npal}}^{\varepsilon}\triangleq\min\left\{\widetilde{{\npal}}_{\indt}^{\varepsilon}\right\}_{\indt=1}^{\optdpr}$ the minimal value of all $\widetilde{{\npal}}_{1}^{\varepsilon},\ldots,\widetilde{{\npal}}_{\optdpr}^{\varepsilon}$. It follows 
 \begin{align*}
 \sum_{\indt=1}^{\optdpr}\frac{\Zlt_\indt^2}{2\widetilde{{\npal}}_{\indt}^{\varepsilon}}\geq  	\sum_{\indt=1}^{\optdpr}\frac{\Zlt_\indt^2}{2\overline{{\npal}}^{\varepsilon}}=\frac{1}{2\overline{{\npal}}^{\varepsilon}}\left|\left|{\Zt}_{\leq\optdpr}\right|\right|^2
 \end{align*}
 and
 \begin{align*}
 &\int_{\left|\left|{\Zt}_{\leq\optdpr}\right|\right|^2> \overline{N}_{\alpha,\beta,\tau} }\frac{1}{\sqrt{\left(2\pi\right)^{\optdpr}\prod_{\indt=1}^{\optdpr}\widetilde{{\npal}}_{\indt}^{\varepsilon}}}e^{-\sum_{\indt=1}^{\optdpr}\frac{\Zlt_\indt^2}{2\widetilde{{\npal}}_{\indt}^{\varepsilon}}}\\
 &\quad\left(\frac{1}{2}\sum_{\indt=1}^{\optdpr}\log\left(2\pi\widetilde{{\npal}}_{\indt}^{\varepsilon}\right)+\sum_{\indt=1}^{\optdpr}\left(\log e\right)\frac{\Zlt_\indt^2}{2\widetilde{{\npal}}_{\indt}^{\varepsilon}}\right)\mathrm{d}{\Zt}_{\leq\optdpr}\\
 \leq& \int_{\left|\left|{\Zt}_{\leq\optdpr}\right|\right|^2> \overline{N}_{\alpha,\beta,\tau} } e^{-\frac{1}{4\overline{{\npal}}^{\varepsilon}}\left|\left|{\Zt}_{\leq\optdpr}\right|\right|^2}\mathrm{d}{\Zt}_{\leq\optdpr}\\
 =&4\overline{{\npal}}^{\varepsilon}e^{-\frac{ \overline{N}_{\alpha,\beta,\tau} }{4\overline{{\npal}}^{\varepsilon}}}=\mathcal{O}\left(1\right)
 \end{align*}
 since both $\overline{{\npal}}^{\varepsilon}>0$ and  $\overline{N}_{\alpha,\beta,\tau}>0$ are constants.
 
 Moreover,
 \begin{align*}
 \ent\left({\Zt}_{\leq\optdpr}\right)=\sum_{\indt=1}^{\optdpr}\ent\left({\Zlt}_{\indt}\right)=\sum_{\indt=1}^{\optdpr}\frac{1}{2}\log 2\pi e \widetilde{{\npal}}_{\indt}^{\varepsilon}
 \end{align*}  since each $\Zlt_{\indt}$ is Gaussian with zero mean and variance $\widetilde{{\npal}}_{\indt}^{\varepsilon}$. Therefore, putting above into (\ref{eq:3.6}),
 \begin{align}
 \label{eq:3.8}
 \ent\left({\Zt}_{\leq\optdpr}|\overline{\Zlt}=1\right)\geq \sum_{\indt=1}^{\optdpr}\frac{1}{2}\log \left(2\pi e \widetilde{{\npal}}_{\indt}^{\varepsilon}\right)-\mathcal{O}\left(1\right).
 \end{align}
 
 Next, if $\left|\left|{\Zt}_{\leq\optdpr}\right|\right|^2\leq \overline{N}_{\alpha,\beta,\tau}$, we know the expectation of $\left|\Rlwd_\indt\right|^2$ can only be smaller since by (\ref{eq:3.0}), $$\Rlwd_\indt=\Clwd_\indt+\Slt_\indt=\Clwd_\indt+\Zlt_\indt-\frac{\npal^*_{\indt}}{\spal_{\indt}}\Clwd_\indt,$$ hence
 \begin{align*}
 \expc\left[\left|\Rlwd_\indt\right|^2\big|\overline{\Zlt}=1 \right]&\leq \expc\left[\left|\Rlwd_\indt\right|^2\right]\\
 &= \expc\left[\left(\Zlt_\indt+\left(1-\frac{\npal_{\indt}^*}{\spal_{\indt}}\right)\Clwd_\indt\right)^2\right]\\
 &=\expc\left[\left|\Zlt_\indt\right|^2\right]+\expc\left[\left(1-\frac{\npal_{\indt}^*}{\spal_{\indt}}\right)^2\Clwd_\indt^2\right]\\
 &\leq \widetilde{{\npal}}_{\indt}^{\varepsilon}+\spal_{\indt}\left(1-\frac{\npal^*_{\indt}}{\spal_{\indt}}\right)^2.
 \end{align*}
 Since entropy of each $\Rlwd_\indt$ is maximized by normal distribution,  it follows that
 \begin{align}
 \nonumber
 &\sum_{\indt=1}^{\optdpr}\ent\left(\Rlwd_\indt|\overline{\Zlt}=1 \right)\\
 \leq
 \nonumber &\sum_{\indt=1}^{\optdpr}\frac{1}{2}\log \left(2\pi e\expc\left[\left|\Rlwd_\indt\right|^2\big|\overline{\Zlt}=1 \right]\right)\\
 \label{eq:3.7}
 =&\sum_{\indt=1}^{\optdpr}\frac{1}{2}\log\left(2\pi e\left(\widetilde{{\npal}}_{\indt}^{\varepsilon}+\spal_{\indt}\left(1-\frac{\npal_{\indt}^*}{\spal_{\indt}}\right)^2\right)\right).
 \end{align}
 
 Now, returning back to (\ref{eq:3.21}), combining (\ref{eq:3.8}) and (\ref{eq:3.7}),
 \begin{align}
 \nonumber
 &\mut\left(\Msg;\optRwdp|\overline{\Zlt}=1\right)\\
 \leq& \sum_{\indt=1}^{\optdpr}\ent\left(\Rlwd_\indt|\overline{\Zlt}=1\right)-\ent\left({\Zt}_{\leq\optdpr}|\overline{\Zlt}=1\right)\\
 \nonumber
 \leq&\sum_{\indt=1}^{\optdpr}\frac{1}{2}\log\left(2\pi e\left(\widetilde{{\npal}}_{\indt}^{\varepsilon}+\spal_{\indt}\left(1-\frac{\npal_{\indt}^*}{\spal_{\indt}}\right)^2\right)\right)\\
 &\quad - \sum_{\indt=1}^{\optdpr}\frac{1}{2}\log \left(2\pi e \widetilde{{\npal}}_{\indt}^{\varepsilon}\right)+\mathcal{O}\left(1\right)\\
 \label{eq:3.9}
 =&\sum_{\indt=1}^{\optdpr}\frac{1}{2}\log\left(1+\frac{\spal_{\indt}}{\widetilde{{\npal}}_{\indt}^{\varepsilon}}\left(1-\frac{\npal_{\indt}^*}{\spal_{\indt}}\right)^2\right)+\mathcal{O}\left(1\right).
 \end{align}
 
 Plugging in $\widetilde{{\npal}}_{\indt}^{\varepsilon}=\frac{1}{1+\varepsilon}\npal^*_{\indt}\left(1-\frac{\npal^*_{\indt}}{\spal_{\indt}}\right)$ into (\ref{eq:3.9}) above, we obtain
 \begin{align}
 \nonumber
 &\mut\left(\Msg;\optRwdp|\overline{\Zlt}=1\right)\\
 \nonumber
 \leq& \sum_{\indt=1}^{\optdpr}\frac{1}{2}\log\left(1+\left(1+\varepsilon\right)\frac{\spal_{\indt}}{{\npal}_{\indt}^*}\left(1-\frac{\npal_{\indt}^*}{\spal_{\indt}}\right)\right)+\mathcal{O}\left(1\right)\\
 \nonumber
 \leq& \sum_{\indt=1}^{\optdpr}\frac{1}{2}\log\left(\left(1+\varepsilon\right)\frac{\spal_{\indt}}{{\npal}_{\indt}^*}\right)+\mathcal{O}\left(1\right)\\
 \label{eq:mutual}
 =&\sum_{\indt=1}^{\optdpr}\frac{1}{2}\log\left(\frac{\spal_{\indt}}{{\npal}_{\indt}^*}\right)+\frac{1}{2}\optdpr\log\left(1+\varepsilon\right)+\mathcal{O}\left(1\right).
 \end{align}
 
 Putting above into (\ref{eq:3.22}) and noting that $\ent\left(\Msg|\overline{\Zlt}=1\right)=\ent\left(\Msg\right)$ since $\Zt_{\leq\optdpr}$ is independent with $\Msg$.
 \begin{align*}
 &\ent\left(\Msg|\optRwdp,\overline{\Zlt}=1\right)\\
 =&\ent\left(\Msg\right)- \mut\left(\Msg;\optRwdp|\overline{\Zlt}=1\right)\\
 =&\cl R- \mut\left(\Msg;\optRwdp|\overline{\Zlt}=1\right)\\
 \geq&\cl R-\sum_{\indt=1}^{\optdpr}\frac{1}{2}\log\left(\frac{\spal_{\indt}}{{\npal}_{\indt}^*}\right)-\frac{1}{2}\optdpr\log\left(1+\varepsilon\right)-\mathcal{O}\left(1\right)\\
 \geq&\cl\left(3\varepsilon-\frac{1}{2}\log\left(1+\varepsilon\right)\right)-\mathcal{O}\left(1\right)\\
 \geq& \left(3-\frac{1}{\ln2}\right)\cl\varepsilon-\mathcal{O}\left(1\right)>\cl\varepsilon.
 \end{align*}

 Therefore, we show that the conditional entropy satisfies $\ent\left(\Msg|\optRwdp\right)>{\cl\varepsilon}$ given $\left|\left|{\Zt}_{\leq\optdpr}\right|\right|^2\leq \overline{N}_{\alpha,\beta,\tau}$. Moreover, as specified in (\ref{eq:4.34}), the fake message $\Usg$ is distributed according to $p_{\Msg|\optRwdp}$ and $\Usg$ is independent with $\Msg$.
 Conditioned on $\left|\left|{\Zt}_{\leq\optdpr}\right|\right|^2\leq \overline{N}_{\alpha,\beta,\tau}$,  following Lemma~\ref{lemma:1}, immediately we have
 \begin{align}
 \label{eq:3.5}
 \mathbbm{P}_{1,1}\triangleq\Pr\left(\Msg\neq\Usg\Big| \  \left|\left|{\Zt}_{\leq\optdpr}\right|\right|^2\leq \overline{N}_{\alpha,\beta,\tau}\right)\geq \frac{{\cl\varepsilon}-1}{\log\mn}.
 \end{align}

 The remaining $\mathbbm{P}_{1,2}$, $\mathbbm{P}_{1,3}$ and $\mathbbm{P}_{1,4}$ can be bounded from below by applying concentration inequalities.

 First we bound $\mathbbm{P}_{1,2}\triangleq\Pr\left(\left|\left|{\Zt}_{\leq\optdpr}\right|\right|^2\leq \overline{N}_{\alpha,\beta,\tau}\right)$.
 
 The mean value of $\left|\left|{\Zt}_{\leq\optdpr}\right|\right|^2$ equals to
 \begin{align*}
 \expc\left[\left|\left|{\Zt}_{\leq\optdpr}\right|\right|^2\right]=\sum_{\indt=1}^{\optdpr}\widetilde{{\npal}}_{\indt}^{\varepsilon}=\sum_{\indt=1}^{\optdpr}\frac{1}{1+\varepsilon}\npal^*_{\indt}\left(1-\frac{\npal^*_{\indt}}{\spal_{\indt}}\right).
 \end{align*}
 
 For a given $\varepsilon>0$, choose $\alpha>0$ and $\beta>0$ satisfying
 \begin{align*}
 \overline{N}_{\alpha,\beta,\tau}&=\sum_{\indt=1}^{\optdpr}\npal_\indt-\alpha\cul-\left(1+\beta\right)\sum_{\indt=1}^{\optdpr}\frac{\left(\npal^*_{\indt}\right)^2}{\spal_{\indt}}\\
 &>\sum_{\indt=1}^{\optdpr}\frac{1}{1+\varepsilon}\npal^*_{\indt}\left(1-\frac{\npal^*_{\indt}}{\spal_{\indt}}\right).
 \end{align*}
 
 Such $\left(\alpha,\beta,\tau\right)$ exists, since $\sum_{\indt=1}^{\optdpr}\frac{1}{1+\varepsilon}\npal^*_{\indt}\left(1-\frac{\npal^*_{\indt}}{\spal_{\indt}}\right)<\sum_{\indt=1}^{\optdpr}\npal^*_{\indt}\left(1-\frac{\npal^*_{\indt}}{\spal_{\indt}}\right)$.
 
 Now, suppose $\varsigma>0$ is the constant with
 \begin{align*}
 \overline{N}_{\alpha,\beta,\tau}&=\sum_{\indt=1}^{\optdpr}\npal_\indt-\alpha\cul-\left(1+\beta\right)\sum_{\indt=1}^{\optdpr}\frac{\left(\npal^*_{\indt}\right)^2}{\spal_{\indt}}\\
 &=\sum_{\indt=1}^{\optdpr}\frac{1+\varsigma}{1+\varepsilon}\npal^*_{\indt}\left(1-\frac{\npal^*_{\indt}}{\spal_{\indt}}\right)\\
 &=\left(1+\varsigma\right)\expc\left[\left|\left|{\Zt}_{\leq\optdpr}\right|\right|^2\right].
 \end{align*}
 
 Applying Markov's inequality,
 \begin{align}
 \label{eq:3.11}
 \mathbbm{P}_{1,2}\triangleq\Pr\left(\left|\left|{\Zt}_{\leq\optdpr}\right|\right|^2\leq \overline{N}_{\alpha,\beta,\tau}\right)\geq \frac{\varsigma}{1+\varsigma}.
 \end{align}
 
 Second, we give a lower bound on $\mathbbm{P}_{1,3}$. Recall the definition of $\mathbbm{P}_{1,3}$ in~(\ref{eq:a.141}).
%
 Since the condition is only a function of $\Cwdp$, it suffices to show that given any fixed $\cwdp=\clwd_1,\ldots,\clwd_{\optdpr}$, 
 \begin{align}
 \label{eq:3.10}
 \Pr\left(\sum_{\indt=1}^{\optdpr}\frac{2\npal^*_{\indt}}{\spal_{\indt}}\clwd_\indt\Zlt_\indt\geq-\alpha\cul\right)\geq 1-\frac{1}{2}e^{-\left(\frac{\alpha}{4\sqrt{\npal\spal}}\right)^2}.
 \end{align}
 
 
 The claim in (\ref{eq:3.10}) above can be shown by regarding $\sum_{\indt=1}^{\optdpr}\frac{2\npal^*_{\indt}}{\spal_{\indt}}\clwd_\indt\Zlt_\indt$ as a Gaussian random variable with zero mean and variance $\sigma^2\triangleq\sum_{\indt=1}^{\optdpr}\left(\frac{2\npal^*_{\indt}}{\spal_{\indt}}\clwd_\indt\right)^2\widetilde{{\npal}}_{\indt}^{\varepsilon}\leq 2\npal\spal\cul^2$. Then based on the cpf. of Gaussian distribution, denoting $\mathrm{erf}\left(x\right)=\frac{2}{\sqrt{\pi}}\int_{0}^{x}e^{-t^2}\mathrm{d}t\geq 1-e^{x^2}$ for any $x>0$ (cf.~\cite{chiani2003new}) the error function, 
 \begin{align*}
 \Pr\left(\sum_{\indt=1}^{\optdpr}\frac{\npal^*_{\indt}}{\spal_{\indt}}\clwd_\indt\Zlt_\indt\geq-\frac{\alpha\cul}{2}\right)\geq &\frac{1}{2}\left(1+\mathrm{erf}\left(\frac{\alpha\cul}{2\sqrt{2}\sigma}\right)\right)\\
 \geq& 1-\frac{1}{2}e^{-\left(\frac{\alpha\cul}{2\sqrt{2}\sigma}\right)^2}\\
 =&1-\frac{1}{2}e^{-\left(\frac{\alpha}{4\sqrt{\npal\spal}}\right)^2}.
 \end{align*}
 Therefore we obtain
 \begin{align}
 \label{eq:3.12}
 \mathbbm{P}_{1,3}\geq 1-\frac{1}{2}e^{-\left(\frac{\alpha}{4\sqrt{\npal\spal}}\right)^2}
 \end{align}
 which is a positive constant.

 Second, we bound $\mathbbm{P}_{1,4}$ using Markov's inequality again. Recall
 \begin{align*}
 &\mathbbm{P}_{1,4}\triangleq\Pr\left(\sum_{\indt=1}^{\optdpr}\left(\frac{\npal^*_{\indt}}{\spal_{\indt}}\Clwd_\indt\right)^2\leq\sum_{\indt=1}^{\optdpr}\left(1+\beta\right)\frac{\left(\npal^*_{\indt}\right)^2}{\spal_{\indt}}\right).
 \end{align*}
 
 Noticing that by linearity, the expectation of $\sum_{\indt=1}^{\optdpr}\left(\frac{\npal^*_{\indt}}{\spal_{\indt}}\Clwd_\indt\right)^2$ is
 \begin{align*}
 \expc\left[\sum_{\indt=1}^{\optdpr}\left(\frac{\npal^*_{\indt}}{\spal_{\indt}}\Clwd_\indt\right)^2\right]=\sum_{\indt=1}^{\optdpr}\left(\frac{\npal^*_{\indt}}{\spal_{\indt}}\right)^2\expc\left[\Clwd_\indt^2\right]=\sum_{\indt=1}^{\optdpr}\frac{\left(\npal^*_{\indt}\right)^2}{\spal_{\indt}}
 \end{align*}
 since $\expc\left[\Clwd_\indt^2\right]=\spal_\indt$ for every $\indt=1,\ldots,\optdpr$.
 
 Therefore, we have
 \begin{align}
 \label{eq:3.13}
 \mathbbm{P}_{1,4}\geq \frac{\beta}{1+\beta}.
 \end{align}
 
 The final step is to consider (\ref{eq:3.17}). Putting (\ref{eq:3.5}), (\ref{eq:3.11}), (\ref{eq:3.12}) and (\ref{eq:3.13}) all together, we conclude
 \begin{align*}
 \mathbbm{P}_{1}&\geq\mathbbm{P}_{1,1}\mathbbm{P}_{1,2}\mathbbm{P}_{1,3}\mathbbm{P}_{1,4}\\
 &\geq \frac{{\cl\varepsilon}-1}{\log\mn}\cdot\frac{\varsigma}{1+\varsigma}\left(1-\frac{1}{2}e^{-\left(\frac{\alpha}{4\sqrt{\npal\spal}}\right)^2}\right)\cdot\frac{\beta}{1+\beta}.
 \end{align*}
 
 \section{Lower Bound on $\mathbbm{P}_2$}
 We apply Markov's inequality to show a lower bound on $\mathbbm{P}_2$. Recall
 \begin{align*}
 \mathbbm{P}_2&\triangleq\Pr\left(\left|\left|\optStl\right|\right|^2\leq \sum_{\indt=\optdpr+1}^{\cl}\npal_\indt\right).
 \end{align*}
 
 Before bounding $\mathbbm{P}_2$, note that
 \begin{align}
 \nonumber
 \left|\left|\optStl\right|\right|^2&=\left|\left|\frac{1}{2}{\left(\optHwdl-\optCwdl\right)}\right|\right|^2\\
 \nonumber
 &=\frac{1}{4}\left|\left|\optHwdl-\optCwdl\right|\right|^2\\
 \label{eq:3.18}
 &=\frac{1}{4}\left(\left|\left|\optHwdl\right|\right|^2+\left|\left|\optCwdl\right|\right|^2-2 \optHwdl\boldsymbol{\cdot} \optCwdl\right)
 \end{align}
 where $\optHwdl\boldsymbol{\cdot} \optCwdl$ is the dot product of $\optHwdl$ and $\optCwdl$. The sequences of random variables  $\optHwdl=\Hlwd_{\optdpr+1},\ldots,\Hlwd_{\cl}$ and $\optCwdl=\Clwd_{\optdpr+1},\ldots,\Clwd_{\cl}$ are i.i.d distributed. Hence, by the linearity of expectation,
 \begin{align*}
 &\expc\left[\left|\left|\optHwdl-\optCwdl\right|\right|^2\right]\\
 =&2\expc\left[\left|\left|\optCwdl\right|\right|^2\right]-2\expc\left[2 \optHwdl\boldsymbol{\cdot} \optCwdl\right]\\
 =&2\expc\left[\left|\left|\optCwdl\right|\right|^2\right]-2\expc\left[\sum_{\indt=\optdpr+1}^{\cl}\Hlwd_\indt\Clwd_\indt\right]\\
 =&2\expc\left[\left|\left|\optCwdl\right|\right|^2\right]-2\sum_{\indt=\optdpr+1}^{\cl}\expc\left[\Hlwd_\indt\Clwd_\indt\right].
 \end{align*}
 
 Moreover, since $\Hlwd_\indt$ and $\Clwd_\indt$ are i.i.d., continuing from above, we have
 \begin{align}
 \nonumber
 &\expc\left[\left|\left|\optHwdl-\optCwdl\right|\right|^2\right]\\
 =&2\expc\left[\left|\left|\optCwdl\right|\right|^2\right]-2\sum_{\indt=\optdpr+1}^{\cl}\expc\left[\Hlwd_\indt\right]\expc\left[\Clwd_\indt\right]\\
 \nonumber
 =&2\expc\left[\left|\left|\optCwdl\right|\right|^2\right]-2\sum_{\indt=\optdpr+1}^{\cl}\left(\expc\left[\Clwd_\indt\right]\right)^2\\
 \label{eq:3.19}
 \leq& 2\expc\left[\left|\left|\optCwdl\right|\right|^2\right]=2\sum_{\indt=\optdpr+1}^{\cl}\spal_{\indt}.
 \end{align}
 
 Furthermore, since both $\spa$ and $\npa$ are feasible solutions of the optimization~(\ref{eq:3.2.5}), we must have
 \begin{align}
 \label{eq:3.20}
 \sum_{\indt=\optdpr+1}^{\cl}\spal_{\indt}\leq  \sum_{\indt=\optdpr+1}^{\cl}2\left(1-\tau\right)\npal_{\indt}.
 \end{align}
 
 Considering (\ref{eq:3.18}), ((\ref{eq:3.19})) and (\ref{eq:3.20}) above, we know the expectation of $\left|\left|\optStl\right|\right|^2$ satisfies
 \begin{align*}
 \expc\left[\left|\left|\optStl\right|\right|^2\right]\leq\sum_{\indt=\optdpr+1}^{\cl}\left(1-\tau\right)\npal_\indt.
 \end{align*}
 
 Then applying the Markov's inequality,
 \begin{align*}
 \mathbbm{P}_2&=\Pr\left(\left|\left|\optStl\right|\right|^2\leq \sum_{\indt=\optdpr+1}^{\cl}\npal_\indt\right)\geq 1-\left(1-\tau\right)=\tau.
 \end{align*}
 
 \subsection{Proof of Lemma~\ref{lemma:12}}
 \label{app:5}
 First, we show the following lower bound on the $\delta$-sum of $\spadl^*_{\indT}$'s:
 \begin{align}
\label{eq:a.20}
 \sum_{\indT=\dpc+1}^{\dpc+\delta\cul}\spadl^*_{\indT}=\Omega\left(\delta\cul\theta\right).
 \end{align}
 
 Suppose $\spad^*=\spadl^*_1,\ldots,\spadl^*_{\cul}$ is the optimal solution of the optimization~(\ref{eq:3.2.6}) and it violates the inequality above. Then we construct a new power allocation sequence $\tilde{\spad}$ as 
 \begin{align*}
 \tilde{\spadl}_\indT=\spadl^*_\indT+a \theta
 \end{align*}
 where $a$ is a small constant and $\theta$ is the chunk-length defined in Section~\ref{sec:5}. Note that $\tilde{\spad}$ becomes a feasible solution if we increase the power constraint from $\spal$ to $\spal+a$. As we presume that the optimal value $C_\cul^\gamma(\spal/\npal)$ is continuous for all $\spal/\npal\in (0,1)$, increasing the parameter by an arbitrarily small constant $a$ only increases the optimal value by an arbitrarily small amount. Hence (\ref{eq:a.20}) holds.
 
 Second, we show the $\delta$-sum of $\spadl^*_{\indT}$'s is also bounded from above as
  \begin{align}
  \label{eq:a.21}
  \sum_{\indT=\dpc+1}^{\dpc+\delta\cul}\spadl^*_{\indT}=\mathcal{O}\left(\delta\cul\theta\right).
  \end{align}
  
Suppose for contradiction that the bound above does not hold. Then since $\cul\theta=\cl$, we obtain $\cl\spal\geq \sum_{\indT=\dpc+1}^{\dpc+\delta\cul}\spadl^*_{\indT}=\omega\left(\cl\right)$, which is clearly a contradiction. 
 
 \subsection{Proof of Lemma~\ref{lemma:6}}
 \label{app:4}
 \subsubsection{Upper Bound on $\mathbbm{P}_3$}
 \label{app:4.1}
 \begin{proof}
 	Fix $\varepsilon>0$ arbitrarily. 	For each \textit{fixed} $\rwdpc\in\rwdsp$, denote
 	\begin{align*}
 	&\mathbbm{Q}_3\triangleq\int_{\mathcal{C}}p_{\mathscr{C}}\left(\mathcal{C}\right)\max_{\dpc_0\leq\dpc\leq\cul}\mathds{1}\left(\clista>\ell\right)\mathrm{d}\mathcal{C}.
 	\end{align*}
 	We first derive a lower bound on $\mathbbm{Q}_3$ and then use the $\delta$-net technique to obtain the desired upper bound in (\ref{eq:4.0}) on $\mathbbm{P}_3$. 
 	
 	We start with the following upper bound on $\mathbbm{Q}_3$:
 	\begin{align}
 	\nonumber
 	\mathbbm{Q}_3
 	&=\int_{\mathcal{C}}p_{\mathscr{C}}\left(\mathcal{C}\right)\max_{\dpc_0\leq\dpc\leq\cul}\mathds{1}\left(\clistas>\ell\right)\mathrm{d}\mathcal{C}\\
 	\nonumber
 	&\leq\int_{\code}p_{\Code}\left(\code\right)\sum_{\dpc_0\leq\dpc\leq\cul}\mathds{1}\left(\clistas>\ell\right)\mathrm{d}\code\\
 	\label{eq:4.15}
 	&\leq\sum_{\dpc_0\leq\dpc\leq\cul}\int_{\codep}p_{\Codep}\left(\codep\right)\\
 	&\left(\mathds{1}\left(\clistas=\ell+1\right)+\cdots+\mathds{1}\left(\clistas=\mn\right)\right)\mathrm{d}\codep.
 	\end{align}
 	
 	Considering only the case when $\clistas=\ell+1$ and denote
 	\begin{align*}
 	&\mathbbm{Q}_3\left(\ell+1\right)\\
 	\triangleq&\sum_{\dpc_0\leq\dpc\leq\cul}\Pr\left(\clistfs=\ell+1|\Rwdpc=\rwdpc\right)\\
 	=&\sum_{\dpc_0\leq\dpc\leq\cul}\int_{\codep}p_{\Codep}\left(\codep\right)\mathds{1}\left(\clistas=\ell+1\right)\mathrm{d}\codep.
 	\end{align*}
 	
 	It follows that
 	\begin{align}
 	\mathbbm{Q}_3\left(\ell+1\right)
 	&\leq\sum_{\dpc_0\leq\dpc\leq\cul}\sum_{\msgs_{\ell+1}}\prod_{\msg\in\msgs_{\ell+1}}\Pr_{\Codep}\left(\msg\in\listas\right),
 	\end{align}
 	which can be considered as a union bound over all possible sets $\msgs_{\ell+1}$ of $\ell+1$ messages and the corresponding sub-collection of codewords $\bigcup_{\msg\in\msgs_{\ell+1}}\codep\left(\msg\right)$. Continuing from above, since each sub-collection $\codep\left(\msg\right)$ contains $2^{\beta \dpc}$ i.i.d. distributed prefixes $\cwdpc$ according to the distribution $p_{\Cwdpc}\left(\cwdpc\right)$  in (\ref{eq:4.1}), from the definition of the set $\listas$ and the fact $\left|\left\{\msgs_{\ell+1}\right\}\right|={\mn\choose {\ell+1}}$, inequality~(\ref{eq:4.20}) follows
   \begin{figure*}[b]
 		\hrule
 	\begin{align}
 	\nonumber
 	\mathbbm{Q}_3\left(\ell+1\right)
 	\nonumber
 	\leq	&{\mn\choose {\ell+1}}\sum_{\dpc_0\leq\dpc\leq\cul}\left(2^{\beta \dpc}\int_{\cwdpc}p_{\Cwdpc}\left(\cwdpc\right)\mathds{1}\left(\left|\left|\cwdpc-\rwdpc\right|\right|^2\leq\aost\right)\mathrm{d}\codep\left(\msg\right)\right)^{\ell+1}\\
 	\label{eq:4.20}
 	=&{\mn\choose {\ell+1}}\cul\max_{\dpc_0\leq\dpc\leq\cul}\left\{\left(2^{\beta \dpc}\Pr_{\Cwdpc}\left(\left|\left|\Cwdpc-\rwdpc\right|\right|^2\leq\aost\right)\right)^{\ell+1}\right\}
 	\end{align}
 \end{figure*}
 	where the randomness are from the variable $\Cwdpc$.
 	
 	Furthermore, maximizing over all the sizes of $\clistas$ from $L=\ell+1$ to $\mn$,
 	\begin{align*}
 	\mathbbm{Q}_3\leq\mn\max_{L\geq \ell+1}\mathbbm{Q}_3\left(L\right).
 	\end{align*}	
 	
 	Now using the bound (\ref{eq:4.20}) above as a component, we derive the promised bound in (\ref{eq:4.0}). The idea is to construct subsets that cover entirely the space $\rwdsp$ containing all received prefixes $\rwdpc$.

 	Consider the following events:
 	\begin{align*}
 	&\mathcal{E}\triangleq\left\{\max_{\dpc_0\leq\dpc\leq\cul}\sup_{\rwdpc}\clistf>\ell\right\},\\
 	&\mathcal{E}\left(\rwdpc\right)\triangleq\left\{\max_{\dpc_0\leq\dpc\leq\cul}\clistf>\ell\right\}.
 	\end{align*}
 	
 	Note that for any subsets $\left\{\rwds_{1},\ldots,\rwds_{D}\right\}$ with $\rwdsp\subseteq\bigcup_{i=1}^{D}\rwds_{i}$, union bound gives
 	\begin{align}
 	\nonumber
 	\mathbbm{P}_3=\Pr_{\Code}\left(\mathcal{E}\right)&\leq\Pr_{\Code}\left\{\bigcup_{i=1}^{D}\bigcup_{\rwdpc\in\mathcal{Y}_{i}}\mathcal{E}\left(\rwdpc\right)\right\}\\
 	\label{eq:4.8}
 	&\leq D\Pr_{\Code}\left\{\bigcup_{\rwdpc\in\mathcal{Y}_{i}}\mathcal{E}\left(\rwdpc\right)\right\}.
 	\end{align}
 	
 	We have $\left|\left|\cwdpc\right|\right|^2\leq \sum_{\indT=1}^{\dpc}\spadl_\indT^*$  and by decoding assumption $\left|\left|\cwdpc-\rwdpc\right|\right|^2\leq \sum_{\indT=1}^{\dpc}\spadl_\indT^*$ for every $\cwdpc\in\Codep$ and for every possible $\rwdpc\in\rwdsp$ to make $\mathcal{E}$ occur. Hence the volume of the space of prefixes $\rwdsp$ is not larger than the volume of a $\dpc$-dimensional ball with some radius $\rho$, which is further a subset of a $\dpc$-dimensional hypercube with edge length $2\rho$, denoted by $\mathcal{H}$:
 	\begin{align*}
 	\rwdsp\subseteq\left\{\rwdpc\in\mathbbm{R}^{\dpc}: \left|\left|\rwdpc\right|\right|\leq \rho\right\}\subseteq\mathcal{H}
 	\end{align*}
 	where the radius is set to be
 	\begin{align*}
 	\rho\triangleq 2\sqrt{\cl \spal}\geq\left|\left|\rwdpc\right|\right|=\left(\sum_{i=1}^{\dpc} \spadl_i^*\right)^{\frac{1}{2}}+\left(\sum_{i=1}^{\dpc}\npadl_i\right)^{\frac{1}{2}}\\
 	 \text{ for all } \spad^*, \npad \text{ and } \dpc=1,\ldots,\cul.
 	\end{align*}

 	One appropriate choice of the subsets $\left\{\mathcal{Y}_{1},\ldots,\mathcal{Y}_{D}\right\}$ can be described as the following. For each edge of $\mathcal{H}$, we divide it equally into $\frac{2\rho}{\Delta}$ many parts. This gives us a partition of $\mathcal{H}$ with every subset in the partition being a smaller $\dpc$-dimensional hypercube of \textit{edge length} $\Delta>0$. The total number of subsets of this partition $D$ is therefore $\left(\frac{2\rho}{\Delta}\right)^{\dpc}$. For each $i$-th small hypercube, we denote its corresponding \textit{central point} by $\rwdpc^i$. In this sense we write the $i$-th small hypercube as $\mathcal{H}\left(\rwdpc^i\right)$ with edge-length $\Delta$. Note that it is possible to cover each $\mathcal{H}\left(\rwdpc^i\right)$ by a $\dpc$-dimensional ball $\mathcal{B}\left(\rwdpc^i\right)$ with radius $\Delta$ and the same center $\rwdpc^i$. Therefore we form a set of $D=\left(\frac{2\rho}{\Delta}\right)^{\dpc}$ many balls $\left\{\mathcal{B}\left(\rwdpc^1\right),\ldots,\mathcal{B}\left(\rwdpc^D\right)\right\}$ covering the space $\rwdsp$.
 	
 	The following slightly tweaked version of the budget reference sequence denoted by $\cbg=\left\{\cost\right\}_{\dpc=1}^{\cul}$ is helpful:
 	\begin{align}
 	\label{eq:4.11}
 	\cost\triangleq \left(\sqrt{\bost}+\Delta\right)^2.
 	\end{align}
 	
 	 	 \begin{figure*}[b]
 		\hrule
 		\begin{align}
 		\label{eq:a.16}
 		&\mathcal{S}_3\triangleq\left\{\mathbf{G}_{\leq\dpc}\in\mathbbm{R}^{\dpc}:  \sum_{\indT=1}^{ \dpc} G_{i}\leq\cost+  \dpc\lambda,  G_{\indT}\in \left\{\lambda, 2\lambda,\ldots,\cost\right\}, \text{ for all }  \indT=1,\ldots, \dpc\right\}.
 		\end{align}
 	\end{figure*}
 	Since for every $\mathcal{B}\left(\rwdpc^i\right)$, triangle inequality implies
 	\begin{align*}
 	 &\bigcup_{\rwdpc\in\mathcal{B}\left(\rwdpc^i\right)}\lista\\
 	 &=\listAxtend,
 	\end{align*}
 	 reprising the way for bounding (\ref{eq:4.20}) and taking a maximization over all $L=\ell+1,\ldots,\mn$, inequality~(\ref{eq:4.12}) can be derived.
  \begin{figure*}[h]
 	\begin{align}
 	\nonumber
 	\Pr_{\Codep}\left\{\bigcup_{\rwdpc\in\mathcal{B}\left(\rwdpc^i\right)}\mathcal{E}\left(\rwdpc\right)\right\}
 	\leq&\mn\max_{L\geq \ell+1}\left\{{\mn\choose {L}}\cul\max_{\dpc_0\leq\dpc\leq\cul}\left(2^{\beta \dpc}\Pr_{\Cwdpc}\left(\left|\left|\Cwdpc-\rwdpc\right|\right|^2\leq\cost\right)\right)^{L}\right\}\\
 	\label{eq:4.12}
 	\leq&\max_{L\geq \ell+1}\mn^{L+1}\cul\max_{\dpc_0\leq\dpc\leq\cul}\left\{\left(2^{\beta \dpc}\Pr_{\Cwdpc}\left(\left|\left|\Cwdpc-\rwdpc\right|\right|^2\leq\cost\right)\right)^{L}\right\}.
 	\end{align}
 	 		\hrule
 \end{figure*}
 	
 	The next step is to bound $\Pr_{\Cwdpc}\left(\left|\left|\Cwdpc-\rwdpc\right|\right|^2\leq\cost\right)$ in (\ref{eq:4.12}) by a $\Lambda$-stage approximation of $\cost$.
 	
 	Define $\lambda\triangleq\frac{4\cl\sqrt{\npal}}{\dpc\log\cl}$ satisfying
 	\begin{align}
 	\label{eq:4.5}
 	\lambda&\leq\frac{\sum_{\indT=\dpc+\delta\cul+1}^{\dpc+2\delta\cul}\spadl_{\indT}^*}{4\dpc}=\frac{\Omega\left(\delta\cul\theta\right)}{\dpc}=\frac{\Omega\left(\delta\cl\right)}{\dpc},\\
 	\label{eq:4.6}
 	\lambda&\geq \frac{2\Delta \sqrt{\bost}+\Delta^2}{\dpc}, \ \text{for all } \ \dpc=1,\ldots,\cul.
 	\end{align}
 	
 	Note that such a $\lambda$ exists. First, $\sqrt{\bost}\leq \sqrt{\cl\npal}$ and $\frac{2\Delta \sqrt{\bost}+\Delta^2}{\dpc}\leq \frac{2\Delta \sqrt{\cl\npal}+\Delta^2}{\dpc}$. Selecting an edge-length $\Delta\triangleq{\frac{\sqrt{\cl}}{\log\cl}}$ guarantees the existence of some $\lambda=\frac{4\cl\sqrt{\npal}}{\dpc\log\cl}\leq \frac{\Omega\left(\delta\cl\right)}{\dpc}$ for large $\cl$ since both $\npal>0$ and $\delta>0$ are constants.
 	
 	We divide the interval $[0,\bost]=\bigcup_{j=1}^{\Lambda} [\lambda\left(j-1\right),\lambda j]$ into $\Lambda\triangleq\frac{\cost}{\lambda}\leq\frac{\cul\log\cl}{4\sqrt{\cl}}$ closed intervals. We assume $\Lambda$ is an integer to simplify our proof. The we choose $\mathcal{G}_\indT$ with $\indT=1,\ldots, \dpc$ to be one of the closed shorter intervals and let $G_\indT$ be the corresponding maximal value in the interval, which exists since $\mathcal{G}_\indT$ is closed, \textit{i.e.,}
 	\begin{align*}
 	&\mathcal{G}_\indT\in\left\{\left[0,\lambda \right],\left[\lambda,2\lambda \right],\ldots,[\bost-\lambda,\cost]\right\}, \ \ \indT=1,\ldots, \dpc,
 	\end{align*}
 	and
 	\begin{align*}
 	&G_\indT\triangleq \left\{a\in\mathcal{G}_\indT: \ a\geq b, \ \ \text{for all} \ b\in\mathcal{G}_\indT\right\}, \ \ \indT=1,\ldots, \dpc.
 	\end{align*}

 	Denote by $\mathbf{G}_{\leq\dpc}\triangleq G_1,\ldots,G_{\dpc}$ and $\mathbf{\Psi}_{\leq\dpc}\triangleq \mathit{\Psi}_1,\ldots,\mathit{\Psi}_{\dpc}$.
 	
 	For notational convenience, we define the following sets:
 	\begin{align*}
 	\mathcal{S}_1&\triangleq\left\{\mathbf{\Psi}_{\leq\dpc}\in\mathbbm{R}^{ \dpc}:  \sum_{\indT=1}^{ \dpc} \mathit{\Psi}_{\indT}\leq\cost\right\},\\
 	\mathcal{S}_2\left(\mathbf{G}_{\leq\dpc}\right)&\triangleq\left\{\mathbf{\Psi}_{\leq\dpc}\in \mathcal{G}_1\times \mathcal{G}_2\times\cdots\times \mathcal{G}_{ \dpc}: \sum_{\indT=1}^{ \dpc} \mathit{\Psi}_{\indT}\leq\cost\right\}.
 	\end{align*}
 	
 	Note that the second set $\mathcal{S}_2\left(\mathbf{G}_{\leq\dpc}\right)$ depends on the intervals $\mathcal{G}_1,\ldots,\mathcal{G}_{\dpc}$, therefore also uniquely determined by the maximal values $G_1,\ldots,G_{\dpc}$.
 	
 	Next, by the construction of $\Cwdpc$, evaluating all possible combination of values taken by $\left|\left|\Cwd_1\right|\right|^2,\left|\left|\Cwd_2\right|\right|^2$ up to $\left|\left|\Cwd_{\dpc}\right|\right|^2$, for any $\rwdpc$,  we can write
 	\begin{align}
 	\nonumber
 	&\Pr_{\Cwdpc}\left(\left|\left|\Cwdpc-\rwdpc\right|\right|^2\leq\cost\right)\\
 	\nonumber
 	\leq&	\Pr_{\Cwdpc}\left(\left|\left|\Cwdpc\right|\right|^2\leq\cost\right)\\
 	\nonumber
 	=&
 	\Pr_{\Cwd_1,\ldots,\Cwd_{ \dpc}}\left(\left|\left|\Cwd_1\right|\right|^2+\left|\left|\Cwd_2\right|\right|^2+\cdots+\left|\left|\Cwd_\dpc\right|\right|^2\leq\cost\right)\\
 	\nonumber
 	=&\Pr_{\Cwd_1,\ldots,\Cwd_{ \dpc}}\left\{\bigcup_{\mathbf{\Psi}_{\leq\dpc}\in\mathcal{S}_1}\bigcap_{\indT=1}^{ \dpc}\left|\left|\Cwd_\indT\right|\right|^2\leq \mathit{\Psi}_\indT\right\},
 	\end{align}
 	which leads to
 	\begin{align}
 	\label{eq:4.13}
 	&\Pr_{\Cwdpc}\left(\left|\left|\Cwdpc-\rwdpc\right|\right|^2\leq\cost\right)\\
 	\leq&\Pr_{\Cwd_1,\ldots,\Cwd_{ \dpc}}\left\{\bigcup_{\mathbf{G}_{\leq\dpc}}\bigcup_{\mathbf{\Psi}_{\leq\dpc}\in\mathcal{S}_2\left(\mathbf{G}_{\leq\dpc}\right)}\bigcap_{\indT=1}^{ \dpc}\left|\left|\Cwd_\indT\right|\right|^2\leq \mathit{\Psi}_\indT\right\}
 	\end{align}
 	since according to the definitions of $  	\mathcal{S}_1$ and $  	\mathcal{S}_2\left(\mathbf{G}_{\leq\dpc}\right)$, the set $\mathcal{S}_1$ is utterly contained in the union of all $\mathcal{S}_2\left(\mathbf{G}_{\leq\dpc}\right)$ over all $G_1,\ldots,G_{ \dpc}$ defined in a $\dpc$-dimensional grid, such that $$\mathcal{S}_1\subseteq\bigcup_{\mathbf{G}_{\leq\dpc}}\mathcal{S}_2\left(\mathbf{G}_{\leq\dpc}\right).$$

 	Furthermore, define a set in~(\ref{eq:a.16}).
Applying union bound to (\ref{eq:4.13}),
 	\begin{align*}
 	&\Pr_{\Cwdpc}\left(\left|\left|\Cwdpc-\rwdpc\right|\right|^2\leq\cost\right)\\
 	\leq&\sum_{\mathbf{G}_{\leq\dpc}}\Pr_{\Cwd_1,\ldots,\Cwd_{ \dpc}}\left\{\bigcup_{\mathbf{\Psi}_{\leq\dpc}\in  	\mathcal{S}_2\left(\mathbf{G}_{\leq\dpc}\right)}\bigcap_{\indT=1}^{ \dpc}\left|\left|\Cwd_\indT\right|\right|^2\leq \mathit{\Psi}_\indT\right\}.
 	\end{align*}
 	
 	There are in total $\Lambda^{ \dpc}$ possible sequences $\mathbf{G}_{\leq\dpc}$. Taking $G_i$ as an upper bound for each $\mathit{\Psi}_i\in\mathcal{G}_\indT$ in $\mathcal{S}_2\left(\mathbf{G}_{\leq\dpc}\right)$,
 	\begin{align*}
 	&\Pr_{\Cwdpc}\left(\left|\left|\Cwdpc-\rwdpc\right|\right|^2\leq\cost\right)\\
 	\leq &\Lambda^{ \dpc}\sup_{\mathbf{G}_{\leq\dpc}\in\mathcal{S}_3}\Pr_{\Cwd_1,\ldots,\Cwd_{ \dpc}}\left\{\bigcap_{\indT=1}^{ \dpc}\left|\left|\Cwd_i\right|\right|^2\leq G_\indT\right\}.
 	\end{align*}
 	
 	Since $\Cwd_\indT$ is independent with each other,
 	\begin{align*}
 	&\Pr_{\Cwdpc}\left(\left|\left|\Cwdpc-\rwdpc\right|\right|^2\leq\cost\right)\\
 	\leq&
 	\Lambda^{ \dpc}\sup_{\mathbf{G}_{\leq\dpc}\in\mathcal{S}_3}\prod_{\indT=1}^{ \dpc}\Pr_{\Cwd_\indT}\left(\left|\left|\Cwd_\indT\right|\right|^2\leq G_\indT\right)
 	\end{align*}
 	where the probability  $\Pr_{\Cwd_\indT}\left(\left|\left|\Cwd_\indT\right|\right|^2\leq G_\indT\right)$ according to the distribution $p_{\Cwd_\indT}$ equals to $\left(\frac{G_\indT}{\spadl _\indT^*}\right)^{\frac{\theta}{2}}$ when $G_\indT\leq\spadl_\indT^*$. Therefore based on the set $\mathcal{S}_3$ we define a new set in~(\ref{eq:a.17})
 	\begin{figure*}[b]
 		\hrule
 	\begin{align}
 		\label{eq:a.17}
 	\overline{\mathcal{S}}\triangleq\Bigg\{\mathbf{G}_{\leq\dpc}\in\mathbbm{R}^{\dpc}:   \sum_{\indT=1}^{ \dpc} G_{\indT}\leq\cost+  \dpc\lambda, G_{\indT}\in \left\{\lambda, 2\lambda,\ldots,\cost\right\},
 	 G_{\indT}< \spadl_{\indT}^*,  \text{ for all }  \indT=1,\ldots, \dpc\Bigg\}
 	\end{align}
  \end{figure*}
 	and we have
 	\begin{align*}
 	\Pr_{\Cwdpc}\left(\left|\left|\Cwdpc-\rwdpc\right|\right|^2\leq\cost\right)&\leq\Lambda^{ \dpc}\sup_{\mathbf{G}_{\leq\dpc}\in\overline{\mathcal{S}}}\prod_{\indT=1}^{\dpc}\left(\frac{G_\indT}{\spadl _\indT^*}\right)^{\frac{\theta}{2}}
 	\end{align*}
 	yielding a bound on (\ref{eq:4.12}):
 	\begin{align}
 	\nonumber
 	&\Pr_{\Code}\left\{\bigcup_{\rwdpc\in\mathcal{B}\left(\rwdpc^i\right)}\mathcal{E}\left(\rwdpc\right)\right\}\\
 	\leq &\max_{L\geq \ell+1}\Bigg\{\mn^{L+1}\cul\max_{\dpc_0\leq\dpc\leq\cul}\\
 	&\left\{2^{\beta \dpc}\Pr_{\Cwdpc}\left(\left|\left|\Cwdpc-\rwdpc\right|\right|^2\leq\cost\right)\right\}^{L}\Bigg\}\\
 	\label{eq:4.7}
 	\leq &\max_{L\geq \ell+1}\cul 2^{\left(\beta+\log\Lambda\right)\cul L+\cl R-L\mathsf{G}\left(\spad^*\right)}
 	\end{align}
 	where $\mathsf{G}\left(\spad^*\right)$ is set to be
 	\begin{align*}
 	\mathsf{G}\left(\spad^*\right)&\triangleq\frac{\theta}{2}\min_{\dpc_0\leq\dpc\leq\cul}\inf_{\mathbf{G}_{\leq\dpc}\in\overline{\mathcal{S}}}\sum_{\indT=1}^{\dpc}\log\frac{\spadl _\indT^*}{G_\indT}-\cl R\\
 	&=\frac{\theta}{2}\min_{\dpc_0\leq\dpc\leq\cul}\inf_{\mathbf{G}_{\leq\dpc}\in\overline{\mathcal{S}}}\sum_{\indT=1}^{\dpc}\log\frac{\spadl _\indT^*}{G_\indT}-\cl \left(\underline{C}_{\cul}^{\gamma}-\varepsilon\right).
 	\end{align*}
 	
 	Recall the definition of $\underline{C}_{\cul}^{\gamma}$. We have $$\underline{C}_{\cul}^{\gamma}=\min_{\dpc_0\leq\dpc\leq\cul}\inf_{\npad\in\npasde}\sum_{\indT=1}^{\dpc}\log\frac{\spadl _\indT^*}{\mathit{\Psi}^*_\indT}$$ in (\ref{eq:3.14}). Hence
 	\begin{align*}
 	\mathsf{G}\left(\spad^*\right)&=\frac{1}{2}\theta	\mathsf{F}\left(\spad^*\right)+\cl\varepsilon
 	\end{align*}
 	where $\spad^*=\spadl_1^*,\ldots,\spadl_{\cul}^*\in\spasd$ is optimal and $\mathsf{F}\left(\spad^*\right)$ is a function of $\spad^*$ defined as 
 	\begin{align}
 	\label{eq:4.33}
 	\mathsf{F}\left(\spad^*\right)\triangleq&\min_{\dpc_0\leq\dpc\leq\cul}\inf_{\mathbf{G}_{\leq\dpc}\in\overline{\mathcal{S}}}\sum_{\indT=1}^{\dpc}\log\frac{\spadl _\indT^*}{G_\indT}\\
 	-&\min_{\dpc_0\leq\dpc\leq\cul}\inf_{\npad\in\npasde}\sum_{\indT=1}^{\dpc}\log\frac{\spadl _\indT^*}{\mathit{\Psi}^*_\indT}.
 	\end{align}
 	
 	The next step is to show $\mathsf{F}\left(\spad^*\right)$ in (\ref{eq:4.33}) above is non-negative and hence $\mathsf{G}\left(\spad^*\right)\geq\cl\varepsilon$. We prove this claim by showing that the sequence $\mathbf{G}_{\leq\dpc}$ is a feasible solution of the optimization~(\ref{eq:3.2.6}).
 	For any $\mathbf{G}_{\leq\dpc}\in\overline{\mathcal{S}}$, by the definition of $\cbg$ in (\ref{eq:4.11}) and the constraints of $\lambda$ in (\ref{eq:4.6}), it follows that
 	\begin{align*}
 	\sum_{\indT=1}^{\dpc}G_\indT
 	\leq&\cost+\dpc\lambda
 	\\=&\cl \npal-  \sum_{\indT=\dpc+\delta\cul+1}^{\cul}\frac{1}{2}\spadl _\indT^*+\dpc\lambda+2\Delta \sqrt{\bost}+\Delta^2,\\
 	\leq&\cl \npal-  \sum_{\indT=\dpc+\delta\cul+1}^{\cul}\frac{1}{2}\spadl _\indT^*+2\dpc\lambda
 	\end{align*}
 	leading to
 	\begin{align}
 	\nonumber
 	&2\cl N-\sum_{\indT=1}^{\dpc}2G_\indT\\
 	\nonumber
 	\geq&\sum_{\indT=\dpc+\delta\cul+1}^{\cul}\spadl _\indT^*-4\dpc\lambda\\
 	\label{eq:4.26}
 	\geq&\sum_{\indT=\dpc+2\delta\cul+1}^{\cul}\spadl _\indT^*\\
 	\label{eq:4.22}
 	\geq&\left(1-\gamma\right)\cl\spal-\sum_{\indT=1}^{\dpc}\spadl _\indT^*
 	\end{align}
 	where (\ref{eq:4.26}) follows by plugging in the condition (\ref{eq:4.5}). The last inequality (\ref{eq:4.22}) holds since Lemma~\ref{lemma:12} indicates that for any $\gamma>0$, there exists a $\delta>0$ small enough such that $$\sum_{\indT=\dpc+1}^{\dpc+2\delta\cul+1}\spadl_\indT^*\leq \gamma\cl\spal.$$

 	\begin{figure*}[h]
	\begin{align}
		\label{eq:4.17}
		\mathbbm{P}_4=
		\int_{\mathcal{C}}p_{\mathscr{C}}\left(\mathcal{C}\right) \max_{\dpc_0\leq\dpc\leq\cul}\max_{\msg\in\msgs}&\sup_{\listas:\clistas\leq\ell}
		\int_{\cwdlc}p_{\Cwdlc|\Codel\left(\msg\right)}\left(\cwdlc|\codel\left(\msg\right)\right)\\
		\nonumber
		&\sup_{\rwdlc\in\ballx}\mathds{1}\left(\left|\listvs\backslash\left\{\msg\right\}\right|>0\right)\mathrm{d}\cwdlc\mathrm{d}\mathcal{C}.
	\end{align}
		\hrule
			\begin{align}
			\label{eq:a.18}
			\quad\quad\mathbbm{P}_4\leq \int_{\code}p_{\Code}\left(\code\right)\max_{\dpc_0\leq\dpc\leq\cul}\max_{\msg\in\msgs}\sup_{\listas:\clistas\leq\ell}\int_{\cwdlc}p_{\Cwdlc|\Codel\left(\msg\right)}\left(\cwdlc|\codel\left(\msg\right)\right) \mathds{1}\left(\cwdlc\in\listxs\right)\mathrm{d}\cwdlc\mathrm{d}\code.
			\end{align}
\end{figure*}
%

\begin{figure*}[h]
	\hrule
		\begin{align}
		\label{eq:4.18}
		\mathbbm{P}_4
		\leq \int_{\code}p_{\Code}\left(\code\right)\underbrace{\max_{\dpc_0\leq\dpc\leq\cul}\max_{\msg\in\msgs}\sup_{\listas:\clistas\leq\ell} \ \ \Pr_{\Cwdlc|\Csetl}\left(\Cwdlc\in\listxs\big|\csetl\right)}_{\mathsf{q}\left(\code\right)}\mathrm{d}\code.
		\end{align}
				\hrule
	\begin{align}
		\nonumber
		\mathbbm{P}_4\leq&\int_{\code}p_{\Code}\left(\code\right)\max_{\dpc_0\leq\dpc\leq\cul}\max_{\msg\in\msgs}\sup_{\listas:\clistas\leq\ell}\mathds{1}\left(\Pr_{\Cwdlc|\Csetl}\left(\Cwdlc\in\listxs\big|\csetl\right)>\eta\right)\mathrm{d}\code+\eta\\
		\label{eq:4.19}
		\leq& \sum_{\dpc_0\leq\dpc\leq\cul}\sum_{\msg\in\msgs}\sum_{\listas:\clistas\leq\ell}
		\ \int_{\codel}p_{\Codel}\left(\codel\right)\mathds{1}\left(\Pr_{\Cwdlc|\Csetl}\left(\Cwdlc\in\listxs\big|\csetl\right)>\eta\right)\mathrm{d}\codel+\eta.
	\end{align}
		\hrule
\end{figure*}

 	
 	Recall the definitions in (\ref{eq:3.15})-(\ref{eq:3.16}).
 	Note that (\ref{eq:4.22}) implies that any $\mathbf{G}_{\leq\dpc}\in\overline{\mathcal{S}}$ is a feasible prefix given the optimizing sequence $\spad^*$ of the optimization problem (\ref{eq:3.2.6}) in the set specified by (\ref{eq:3.15})-(\ref{eq:3.16}). Hence the corresponding objective value satisfies
 	\begin{align*}
 	\min_{\dpc_0\leq\dpc\leq\cul}\inf_{\mathbf{G}_{\leq\dpc}\in\overline{\mathcal{S}}}\sum_{\indT=1}^{\dpc}\log\frac{\spadl _\indT^*}{G_\indT}\geq \min_{\dpc_0\leq\dpc\leq\cul}\inf_{\npad\in\npasde}\sum_{\indT=1}^{\dpc}\log\frac{\spadl _\indT^*}{\mathit{\Psi}^*_\indT}
 	\end{align*}
 	and as a conclusion
 	$\mathsf{G}\left(\spad^*\right)\geq\cl\varepsilon$. Substituting this back into (\ref{eq:4.7}), it follows that for every fixed $\rwdpc\in\rwdsp$, and we get
 	\begin{align}
 	\label{eq:4.9}
 	\Pr_{\Code}\left\{\bigcup_{\rwdpc\in\mathcal{B}\left(\rwdpc^i\right)}\mathcal{E}\left(\rwdpc\right)\right\}&\leq\cul 2^{\cl R-\left(\ell+1\right)\big[\varepsilon\theta-\left(\beta+\log\Lambda\right)\big]\cul}
 	\end{align}
 	when $\cl$ and $\cul$ are large enough.
 	
 	The last step is to consider a union of all subsets $\mathcal{B}\left(\rwdpc\left(1\right)\right),\ldots,\mathcal{B}\left(\rwdpc\left(D\right)\right)$ covering $\rwdsp$. By combining (\ref{eq:4.7}) with (\ref{eq:4.9}) and noting that $D=\left(\frac{2\rho}{\Delta}\right)^{\dpc}\leq\left(\frac{2\rho}{\Delta}\right)^{\cul}$, we obtain
 	\begin{align}
 	\nonumber
 	\mathbbm{P}_3
 	&\leq D\max_{i}\Pr_{\Code}\left\{\bigcup_{\rwdpc\in\mathcal{B}\left(\rwdpc^i\right)}\mathcal{E}\left(\rwdpc\right)\right\}\\
 	\label{eq:5.33}
 	&\leq \left(\frac{2\rho}{\Delta}\right)^{\cul}\cul 2^{\cl R-\left(\ell+1\right)\big[\varepsilon\theta-\left(\beta+\log\Lambda\right)\big]\cul}
 	\end{align}
 	where $\theta>0$, $\ell>0$, $\beta>0$ are arbitrary, $\cul=\frac{\cl}{\theta}$ and
 	\begin{align*}
 	&\Delta={\frac{\sqrt{\cl}}{\log\cl}},\\
 	&\rho= 2\sqrt{\cl \spal},\\
 	&\Lambda\leq\frac{\cul\log\cl}{4\sqrt{\cl}}.
 	\end{align*}
 	
 	Substituting above into (\ref{eq:5.33}), after simplification we conclude
 	\begin{align*}
 	\mathbbm{P}_3\leq \left(4\sqrt{\spal}\log\cl\right)^{\cul}\cul 2^{\cl R-\left(\ell+1\right)\big[\varepsilon\theta-\left(\beta+\log\left(\frac{\sqrt{\cl}\log\cl}{4\theta}\right)\right)\big]\cul}.
 	\end{align*}
 \end{proof}
 
 \subsubsection{Upper Bound on $\mathbbm{P}_4$}
 \label{app:4.2}
 

 \begin{lemma}
 	\label{lemma:8}
 	For any collection of codewords $\code$, chunk index $\dpc_0\leq\dpc\leq\cul$, transmitted message $\msg\in\msgs$, pre-list $\listas$ of size no greater than $\ell$ and any possible suffix $\cwdlc\in\csetl$,
 	\begin{align}
 	\label{eq:4.23}
 	&\sup_{\rwdlc\in\ballx}\mathds{1}\left(\left|\listvs\backslash\left\{\msg\right\}\right|>0\right)
 	\\ \leq&	 \mathds{1}\left(\cwdlc\in\listx\right)
 	\end{align}
 	where $\listx$ denotes the set 
 	\begin{align*}
 	&\Bigg\{\cwdlc\in\csetl: \text{There exists} \ \cwdlc'\in\ \codel\left(\esg\right) \ \text{with} \\ &\esg\in\listas\backslash\left\{\msg\right\} 
 	\text{ such that} \ \left|\left|\cwdlc-\cwdlc'\right|\right|^2\leq 2\left(\cl N-\bost\right)\Bigg\}.
 	\end{align*}
 \end{lemma}
 
 \begin{proof}
 	Recall
 	\begin{align*}
 	\ballx\triangleq\left\{\rwdlc\in\mathbbm{R}^{\cl-\dpc\theta}:\left|\left|\cwdlc-\rwdlc\right|\right|^2\leq \cl N-\bost\right\}.
 	\end{align*}
 	
 	To see the inequality (\ref{eq:4.23}) in above, we fix any $\rwdlc\in\ballx$. If the indicator function $\mathds{1}\left(\left|\listvs\backslash\left\{\msg\right\}\right|>0\right)$ gives a $1$, by Definition~\ref{def:10}, it means there exists some $\esg\neq\msg$ and $\cwdlc'\in\ \codel\left(\esg\right)$ such that $ \left|\left|\cwdlc'-\rwdlc\right|\right|^2\leq \cl N-\bost$. Since $\rwdlc\in\ballx$, we also have $\left|\left|\cwdlc-\rwdlc\right|\right|^2\leq \cl N-\bost$.
 	Note that for any $\cwdlc,\cwdlc'\in\overline{\Code}$, triangle inequality implies
 	\begin{align*}
 	\left|\left|\cwdlc-\cwdlc'\right|\right|^2\leq&\left|\left|\cwdlc-\rwdlc\right|\right|^2+\left|\left|\cwdlc'-\rwdlc\right|\right|^2\\
 	\leq &2\left(\cl N-\bost\right).
 	\end{align*}
 	
 	Therefore the indicator function $\mathds{1}\left(\cwdlc\in\listxs\right)$ also takes value one.
 \end{proof}
 
 From Lemma~\ref{lemma:8}, by definition, probability $\mathbbm{P}_4$ has the expression in~(\ref{eq:4.17}).
 Continuing from (\ref{eq:4.17}) and applying Lemma~\ref{lemma:8}, the bound in (\ref{eq:a.18}) follows.
 Writing everything inside the second integral as a conditional probability, we can further bound $\mathbbm{P}_4$ in (\ref{eq:4.18}).

 Let $\eta>0$ be any constant and define a function $\mathsf{q}\left(\code\right)$ of $\code$ in (\ref{eq:4.18}).
 
 Truncating the RHS of (\ref{eq:4.18}) into two parts corresponding to $\mathsf{q}\left(\code\right)\leq\eta$ and $\mathsf{q}\left(\code\right)>\eta$ respectively,
 \begin{align}
 \nonumber
 \mathbbm{P}_4&\leq\int_{\code}p_{\Code}\left(\code\right)\mathsf{q}\left(\code\right)\mathds{1}\left(\mathsf{q}\left(\code\right)>\eta\right)\mathrm{d}\code\\
 &\quad+\int_{\code}p_{\Code}\left(\code\right)\mathsf{q}\left(\code\right)\mathds{1}\left(\mathsf{q}\left(\code\right)\leq\eta\right)\mathrm{d}\code\\
 \nonumber
 &\leq\int_{\code}p_{\Code}\left(\code\right)\mathds{1}\left(\mathsf{q}\left(\code\right)>\eta\right)\mathrm{d}\code+\eta.
 \end{align}
 
 By the definition of $\mathsf{q}\left(\code\right)$, $\mathbbm{P}_4$ can be bounded in (\ref{eq:4.19}).

 It remains to show that for any chunk index $\dpc_0\leq\dpc\leq\cul$, transmitted message $\msg\in\msgs$ and pre-list $\listas$ of size no greater than $\ell$,
 \begin{align}
 \nonumber
 &\int_{\codel}p_{\Codel}\left(\codel\right)\\
 &\quad\mathds{1}\left(\Pr_{\Cwdlc|\Csetl}\left(\Cwdlc\in\listxs\big|\csetl\right)>\eta\right)\mathrm{d}\codel\\
 \label{eq:5.13}
 \leq
 &\left\{\ell2^{\beta\left(\cul-\dpc\right)}e^{-{\Omega\left(\cl\delta^2\right)}/{2\spal}}\right\}^{\lfloor\eta 2^{\beta\left(\cul-\dpc\right)}\rfloor}
 \end{align}
 where $\underline{\spadl}^*$ denotes the minimal value of all coordinates of the sequence $\spad^*$.
 
 Therefore (\ref{eq:4.19}) yields the desired bound on $\mathbbm{P}_4$ such that 
 \begin{align*}
 \mathbbm{P}_4\leq  \cul\mn{ \mn\choose\ell  }\left\{\ell2^{\beta\left(\cul-\dpc\right)}e^{-{\Omega\left(\cl\delta^2\right)}/{2\spal}}\right\}^{\lfloor\eta 2^{\beta\left(\cul-\dpc\right)}\rfloor}+\eta
 \end{align*}
 by counting the total number of possible $\dpc$, $\msg$ and $\listas$ in the triple of summations.
 
 
 The claim above is true by analysing the conditional probability $$\Pr_{\Cwdlc|\Csetl}\left(\Cwdlc\in\listxs\big|\csetl\right)$$ in (\ref{eq:4.19}) as follows.
 Note that by our construction of the encoder, for every suffix of codeword $\cwdlc$ with a given $\dpc$,
 \begin{eqnarray}
 &p_{\Cwdlc|\Csetl}\left(\cwdlc|\csetl\right)\\
 =&\begin{cases}
 \frac{1}{2^{\beta\left(\cul-\dpc\right)}} \quad &\text{if } \ \cwdlc\in\csetl\\
 0 \quad &\text{otherwise}
 \end{cases}
 \end{eqnarray}
 meaning that given some division point $\dpc$ and a fixed sub-collection of codewords $\codel$, the probability for each possible suffix of codeword to be transmitted is equal to each other (and there are in total $2^{\beta\left(\cul-\dpc\right)}$ possible suffixes).
 Therefore once the probability $\Pr_{\Cwdlc|\Csetl}\left(\Cwdlc\in\listxs\big|\csetl\right)$ is greater than $\eta$, sufficiently, it leads to the fact that inside the sub-collection of codewords $\csetl$, more than $\alpha\triangleq\lfloor\eta 2^{\beta\left(\cul-\dpc\right)}\rfloor$ many suffixes of codewords are in the set $\listxs$. In this way we have
 \begin{align}
 \nonumber
 &\int_{\csetl}p_{\Csetl}\left(\csetl\right)\\
 &\quad\mathds{1}\left(\Pr_{\Cwdlc|\Csetl}\left(\Cwdlc\in\listxs\big|\csetl\right)>\eta\right)\mathrm{d}\csetl\\
 \label{eq:5.12}
 \leq&\Pr_{\Codel}\left(\left|\Csetl\bigcap\listxs\right|>\alpha\right)
 \end{align}
 where the randomness in the RHS is from the random sub-collection $\Codel$, more specifically, the union of sub-collection  $\cup_{\esg\in\listas+\msg}\Codel\left(\esg\right)$. It equals
 \begin{align}
 \label{eq:5.10}
 \Pr_{\Codel}\left(\left|\Csetl\bigcap\listxs\right|>\alpha\right)\leq \left( \Pr_{\Codel}\left(\Cwdlc\in\listxs\right)\right)^{\alpha}
 \end{align}
 since the suffixes in $\Csetl$ are selected uniformly and independently.

 Since each sub-collection $\codel\left(\esg\right)$ with $\esg\neq\msg$ contains exactly $2^{\beta\left(\cul-\dpc\right)}$ suffixes of codewords, and there are at most $\ell$ many such $\esg$ in $\listas$, we define
 \begin{align*}
 \xi\triangleq  \ell2^{\beta\left(\cul-\dpc\right)}
 \end{align*}
 and denote all suffixes in $\cup_{\esg\in\listas}\Codel\left(\esg\right)$ by
 $
 \left\{\Cwdlc\left(1\right),\ldots,\Cwdlc\left(\xi\right)\right\}
 $
 then for a randomly selected suffix $\Cwdlc$, from the definition of the set $\listxs$,
 \begin{align}
 \nonumber
& \Pr_{\Codel}\left(\Cwdlc\in\listxs\right)\\
 =&\Pr_{\Codel}\left\{\bigcup_{j=1}^{\xi}\left|\left|\Cwdlc-\Cwdlc\left(j\right)\right|\right|^2\leq 2\left(\cl N-\bost\right)\right\}\\
 \label{eq:5.11}
 \leq& \xi\max_{j}\Pr_{\Cwdlc,\Cwdlc\left(j\right)}\left(\left|\left|\Cwdlc-\Cwdlc\left(j\right)\right|\right|^2\leq 2\left(\cl N-\bost\right)\right)
 \end{align}
 where the last inequality follows after applying union bound.

 For all $\Cwdlc\left(j\right)$, over the randomness of $\Codel$ itself and the selection of codewords in $\Codel$,
 \begin{align}
 \nonumber
 &\Pr_{\Cwdlc,\Cwdlc\left(j\right)}\left(\left|\left|\Cwdlc-\Cwdlc\left(j\right)\right|\right|^2\leq 2\left(\cl N-\bost\right)\right)\\
 \nonumber
 =& \Pr_{\Cwdlc,\Cwdlc\left(j\right)}\left(\sum_{i=\dpc+1}^{\cul}\left|\left|\Cwd_{i}-\Cwd_{i}\left(j\right)\right|\right|^2\leq 2\left(\cl N-\bost\right)\right)\\
 \label{eq:5.5}
 \leq&\Pr_{\Cwd_{\dpc+1},\ldots,\Cwd_{ \cul}}\left(\left|\left|\Cwd_{\dpc+1}\right|\right|^2+\cdots+\left|\left|\Cwd_{\cul}\right|\right|^2\leq 2\left(\cl N-\bost\right)\right).
 \end{align}
 The first inequality holds since the distribution and construction of $\Cwdlc=\Cwd_{\dpc+1}\circ\Cwd_{\dpc+2}\circ\cdots\circ\Cwd_{\cul}$. The last inequality (\ref{eq:5.5}) holds since $\left|\left|\Cwd_{\indT}-\Cwd_{\indT}\left(j\right)\right|\right|^2=\left|\left|\Cwd_{\indT}\right|\right|^2+\left|\left|\Cwd_{\indT}(j)\right|\right|^2-2\Cwd_{\indT}\boldsymbol{\cdot}\Cwd_{\indT}$ and the term $\left|\left|\Cwd_{\indT}(j)\right|\right|^2-2\Cwd_{\indT}$ is positive with probability larger than one-half. Therefore substituting $\left|\left|\Cwd_{\indT}(j)\right|\right|^2-2\Cwd_{\indT}$ by zero does not decrease the probability.

 %
 
 The last step is to show that for every $\dpc+1\leq\indT\leq\cul$ and $a\in\mathbbm{R}$, there exists a $\theta$ large enough such that
 \begin{align}
 \label{eq:5.6}
 \expc\left[e^{a\left(\left|\left|\Cwd_{\indT}\right|\right|^2-\expc\left[\left|\left|\Cwd_{\indT}\right|\right|^2\right]\right)}\right]\leq e^{\frac{a^2 \spadl _{\indT}^*}{2}}.
 \end{align}

 Note that by the construction of the encoder, each $\Cwd_{\indT}$  ($1\leq\indT\leq\cul$) is chosen \textit{uniformly} from a $\theta$-dimensional ball with radius $\spadl _{\dpc}^*$. Therefore as proved in~\cite{voelker2017efficiently}, if we denote $I$ as a random variable uniformly chosen from the unit interval $[0,1]$, $\Cwd_{\indT}$ can be regarded as $I^{\frac{1}{\theta}}$ times a $\theta$-dimensional vector $\mathbf{A}$ uniformly chosen from a sphere of radius $\sqrt{\spadl _{\indT}^*}$ such that
 \begin{align*}
 \Cwd_{\indT}=I^{\frac{1}{\theta}}\cdot \mathbf{A}
 \end{align*}
 where $\left|\left|\mathbf{A}\right|\right|^2=\spadl _{\indT}^*$. In this sense, we know
 \begin{align}
 \nonumber
 \expc\left[\left|\left|\Cwd_{\indT}\right|\right|^2\right]&=\expc\left[I^{\frac{1}{\theta}}\left|\left|\mathbf{A}\right|\right|^2\right]\\
 \nonumber
 &=\expc\left[I^{\frac{1}{\theta}}\right]\spadl _{\indT}^*\\
 \label{eq:5.7}
 &=\frac{\theta}{1+\theta}\spadl _{\indT}^*
 \end{align}
 and
 \begin{align}
 \nonumber
 \expc\left[e^{a\left|\left|\Cwd_{\indT}\right|\right|^2}\right]&=\expc\left[e^{aI^{\frac{1}{\theta}}\left|\left|\mathbf{A}\right|\right|^2}\right]\\
 &=\expc\left[e^{aI^{{1}/{\theta}}\spadl _{\indT}^*}\right]\\
 \label{eq:5.8}
 &=\int_{0}^{1}e^{aI^{\frac{1}{\theta}}\spadl _{\indT}^*}\mathrm{d}I\leq e^{a\spadl _{\indT}^*}.
 \end{align}
 
 Putting (\ref{eq:5.7}) and (\ref{eq:5.8}) together, we obtain
 \begin{align*}
 \expc\left[e^{a\left(\left|\left|\Cwd_{\indT}\right|\right|^2-\expc\left[\left|\left|\Cwd_{\indT}\right|\right|^2\right]\right)}\right]&\leq e^{a\spadl _{\indT}^*}\cdot e^{-\frac{a\theta}{1+\theta}\spadl _{\indT}^*}\\
 &=e^{\frac{a}{1+\theta}\spadl _{\indT}^*}
 \leq e^{\frac{a^2 \spadl _{\indT}^*}{2}}
 \end{align*}
 where the last inequality follows by setting the dimension $\theta$ large enough. This is possible since we set $\theta=\sqrt{\cl}$. Therefore, (\ref{eq:5.6}) holds, which implies that $\left|\left|\Cwd_{\indT}\right|\right|^2$ is \textit{sub-Gaussian} with parameter $\sqrt{ \spadl _{\indT}^*}$. 
 Note that $\expc\left[\Cwd_{\indT}\right]=\spadl^*_{\indT}$ for all $\indT=\dpc+1,\ldots,\cul$. Hence $\sum_{\indT=\dpc+1}^{\cul}\expc\left[\Cwd_{\indT}\right]=\sum_{\indT=\dpc+1}^{\cul}\spadl^*_{\indT}$. Also, from (\ref{eq:4.10}) in Definition~\ref{def:3} of the budget reference sequence, we have 
 \begin{align*}
 2\left(\cl N-\bost\right)=\sum_{\indT=\dpc+\cul\delta+1}^{\cul}\spadl^*_{\indT}.
 \end{align*}
Therefore,
 \begin{align*}
 & \Pr_{\Cwd_{\dpc+1},\ldots,\Cwd_{ \cul}}\left(\left|\left|\Cwd_{\dpc+1}\right|\right|^2+\cdots+\left|\left|\Cwd_{\cul}\right|\right|^2\leq 2\left(\cl N-\bost\right)\right)\\
 =&\Pr\left\{\sum_{\indT=\dpc+1}^{\cul}\left|\left|\Cwd_{\indT}\right|\right|^2-\sum_{\indT=\dpc+1}^{\cul}\expc\left[\left|\left|\Cwd_{\indT}\right|\right|^2\right]\leq -\sum_{\indT=\dpc+1}^{\dpc+\cul\delta}\spadl^*_{\indT}\right\}.
 \end{align*}
 
 Considering above and applying Hoeffding's inequality~\cite{hoeffding1963probability},
 \begin{align}
 \nonumber
 &\Pr_{\Cwd_{\dpc+1},\ldots,\Cwd_{ \cul}}\left(\left|\left|\Cwd_{\dpc+1}\right|\right|^2+\cdots+\left|\left|\Cwd_{\cul}\right|\right|^2\leq 2\left(\cl N-\bost\right)\right)\\
 \leq&\exp\left(-\frac{\left(\sum_{\indT=\dpc+1}^{\dpc+\cul\delta}\spadl^*_{\indT}\right)^2}{2\sum_{\indT=\dpc+1}^{\cul}\spadl^*_{\indT}}\right)\\
 \label{eq:5.9}
=&\exp\left(-\frac{\Omega\left(\left(\delta\cul\theta\right)^2\right)}{2\cl\spal}\right)\\
 \label{eq:5.14}
 =&e^{-{\Omega\left(\cl\delta^2\right)}/{2\spal}}.
 \end{align}
 The inequality (\ref{eq:5.9}) holds since $\sum_{\indT=\dpc+1}^{\cul}\spadl^*_{\indT}\leq\cl\spal$ and $\sum_{\indT=\dpc+1}^{\dpc+\delta\cul}\spadl^*_{\indT}=\Omega\left(\delta\cul\theta\right)$  in Lemma~\ref{lemma:12}. The inequality (\ref{eq:5.14}) holds since $\cul\theta=\cl$. Combining inequality (\ref{eq:5.14}) with (\ref{eq:5.12}), (\ref{eq:5.10}) and (\ref{eq:5.11}), the result in (\ref{eq:5.13}) follows. This completes the proof.

\end{document}